\documentclass[10pt]{article}

\usepackage{_style}

\usepackage{quoting}
\quotingsetup{vskip=0pt}

\usepackage[compact]{titlesec}
\usepackage{titletoc}

\usepackage[compact]{titlesec}
\titlespacing{\section}{0pt}{1.5ex}{0ex}
\titlespacing{\subsection}{0pt}{1.5ex}{0ex}
\titlespacing{\subsubsection}{0pt}{1ex}{0ex}
\titlespacing{\paragraph}{0pt}{1.5ex}{1ex}
\allowdisplaybreaks

\newtheorem{property}{Property}
\newtheorem{assumption}{Assumption}

\usepackage{url}
\newcommand{\csp}{\ensuremath{\mathsf{CSP}}}%
\newcommand{\D}{\mathsf{D}}
\usepackage{bbm}

\def\cA{{\cal A}}

\def\cF{{\cal F}}
\def\cG{{\cal G}}

\def\cI{{\cal I}}

\def\cN{{\cal N}}

\def\cP{{\cal P}}

\def\B{{\mathsf{B}}}
\newcommand{\cspfa}{\csp_{f}(\alpha)}
\newcommand{\trunc}{\mathrm{trnc}}

\begin{document}
\author{Antares Chen\thanks{\texttt{antaresc@uchicago.edu}}}
\author{Neng Huang\thanks{\texttt{nenghuang@uchicago.edu}}}
\author{Kunal Marwaha\thanks{\texttt{kmarw@uchicago.edu}}}
\affil{University of Chicago}
\title{Local algorithms and the failure of log-depth
quantum advantage on sparse random CSPs}
\date{\today}
\maketitle
\begin{abstract}
We construct and analyze a message-passing algorithm for random constraint satisfaction problems (CSPs) at large clause density, generalizing work of El Alaoui, Montanari, and Sellke for Maximum Cut~\cite{ams21} through a connection between random CSPs and mean-field Ising spin glasses~\cite{ams20,jmss22}.
For CSPs with even predicates, the algorithm 
asymptotically solves a stochastic optimal control problem dual to an extended Parisi variational principle.
This gives an optimal fraction of satisfied constraints among algorithms obstructed by the branching overlap gap property of Huang and Sellke~\cite{huang2022tight},  notably including the Quantum Approximate Optimization Algorithm and all quantum circuits on a bounded-degree architecture of up to $\epsilon \cdot \log n$ depth~\cite{chou2022limitations}.
\end{abstract}

\pagenumbering{arabic}

\thispagestyle{empty}
\newpage
\setcounter{tocdepth}{2}
{
    \hypersetup{linkcolor=blue}
    \tableofcontents
}
\newpage

\section{Introduction}
Today's quantum computers are starting to show signs of advantage over classical machines.
Many people are seeking to use this technology in a practical application as soon as possible.
One such area is combinatorial optimization: even though the problems are $\NP$-hard, optimization has incredible value to the contemporary world. 
This industry is also eager to experiment with quantum computers, perhaps because rigorous worst-case analysis is not always relevant.

The most popular quantum protocol for classical optimization is the \emph{Quantum Approximate Optimization Algorithm} (QAOA), proposed in~\cite{farhi2014quantum}.
Given an optimization problem, the protocol associates a qubit to each variable, and evolves a quantum state to a near-ground state of the problem, depending on choice of hyperparameters.
The QAOA has also sparked the interest of theorists, characterizing the protocol's performance on canonical problems, and debating whether it is any more powerful than existing classical algorithms (e.g. ~\cite{barak2015beating,hastings2019classical,basso2021quantum}).

One key property of the QAOA is its \emph{locality}: at each layer of the algorithm, a qubit interacts only with qubits associated with adjacent variables in the optimization problem. Locality is known to limit the QAOA up to logarithmic circuit depth when the optimization problem is \emph{sparse}~\cite{farhi2020quantumtypical,farhi2020quantumworst}. 
This limitation is due to a phenomenon known as the \emph{overlap gap property} (OGP), a type of clustering behavior of near-optimal solutions~\cite{gamarnik2014limits}.
Notably, this is an \emph{average-case} property of many random optimization problems. Recent formulations of the OGP have led to quantitative obstructions for broad families of algorithms, both quantum and classical~\cite{cgpr19,chou2022limitations,huang2022tight}.

In this work, we show that \emph{with high probability} over a large family of sparse optimization problems, there is an \emph{optimal} algorithm among those obstructed by the OGP. Furthermore, this algorithm is \emph{classical}, so there can be no quantum advantage by the QAOA up to logarithmic depth.
In fact, this no-go result extends to logarithmic-depth quantum circuits on fixed, bounded-degree architectures, such as the \emph{brickwork} circuits~\cite{Broadbent_2009} used in random circuit sampling~\cite{Bouland_2018}.
This is the strongest \emph{average-case} result known to the authors regarding quantum advantage in classical optimization.

Our algorithm is related to the theory of \emph{message-passing} used in random inference problems~\cite{bayati11}, and inspired by rich connections to statistical physics models such as the mean-field spin glass~\cite{sherrington1975solvable}. 
The parameters of the algorithm are chosen by solving a stochastic control problem dual to a problem determining the quantitative barrier of the OGP, as pioneered in~\cite{ams20,huang2022tight}.
In particular, our algorithm extends work of \cite{ams21}, and uses a sparse-to-dense connection observed in e.g. \cite{dms15,jmss22}. The performance of the algorithm, like that of \cite{montanari2019optimization,ams20,ams21}, is dependent on a technical condition (\Cref{assn:alg_minimizer_exists}) which we explain in the text.
\subsection{Background}

We study Constraint Satisfaction Problems (CSPs). A CSP instance consists of $n$ variables and a set of constraints associated with a \emph{predicate} $f: \{\pm 1\}^r \to \{0,1\}$.
\begin{definition}[Instance of a random CSP]
\label{defn:randomcsp}
Consider a function $f : \{\pm 1 \}^r \to \{0,1\}$. Fix a constant $\alpha > 0$.
A \emph{random CSP instance} $\cI$ over $n$ variables $\{x_i\}_{i \in [n]}$ and $m = \alpha \cdot n$ clauses is generated by the following procedure: For each $i = 1, \ldots, m$, draw $i_1, \ldots, i_r$ uniformly i.i.d from $[n]$, draw $r$ random signs $\epsilon_{i_1}, \dots, \epsilon_{i_r}$ uniformly i.i.d from $\{ \pm 1\}$, and add the constraint $e$ to $E(\cI)$, where $e$ describes the clause $f_e(x_e) \defeq f(\epsilon_{i_1} x_{i_1},\ldots, \epsilon_{i_r} x_{i_r})$. A clause is satisfied if it evaluates to $1$.
\end{definition}
We denote this random model as $\cspfa$. We are interested in the best \emph{satisfying fraction} (i.e. largest fraction of clauses satisfied by some assignment in $\{\pm 1\}^n$) an algorithm can achieve as $n \to \infty$ when $\alpha$ is a very large constant, i.e. so that instances of $\cspfa$ are typically unsatisfiable.
This value is related to the rich literature of \emph{mean-field Ising spin glasses}. We summarize some relevant results, especially drawing on~\cite{ams20,ams21,jmss22}.
Let $\mathscr{U}$ denote the following set of functions:
\begin{equation*}
\mathscr{U}
\defeq \bigg\{ \mu: [0,1) \rightarrow \R_{\geq 0} : \, \mu \textup{ is non-decreasing, right-continuous, and } \int_{0}^{1} \mu(t) \, dt < \infty \bigg\}\,.
\end{equation*}
Given a univariate polynomial $\xi: \R \to \R$ with non-negative coefficients, define the functional $\sfP_\xi: \mathscr{U} \rightarrow \R$ given by 
\begin{equation*}
\label{eq.parisi-functional}
\sfP_\xi(\mu) \defeq \Phi^\mu(0, 0) - \frac{1}{2} \cdot \int_{0}^{1} \xi''(t) \cdot t \cdot \mu(t) \, dt\,,
\end{equation*}
where $\Phi^\mu: [0, 1) \times \R \rightarrow \R$ is a solution of the following partial differential equation:\footnote{For the multivariable function $\Phi^\mu(t,x)$, we use subscripts to denote partial derivatives.} 
\begin{equation}
\label{eq.parisi-equations}
\begin{aligned}
\Phi^\mu_t(t, x) &= - \frac{\xi''(t)}{2} \cdot \Big( \Phi^\mu_{xx}(t, x) + \mu(t) \cdot \big( \Phi^\mu_x(t, x) \big)^2 \Big) & & \forall t \in [0,1), \forall x \in \R\,, \\
\Phi^\mu(1, x) &= \lvert x \rvert & & \forall x \in \R\,.  
\end{aligned}
\end{equation}
$\sfP_{\xi}$ is known as the Parisi functional, while \Cref{eq.parisi-equations} are known as the Parisi equations~\cite{parisi79}. Let the infimum of $\sfP_\xi$ be defined as $\OPT_\xi \defeq  \inf_{\mu \in \mathscr{U}} \sfP_\xi(\mu)$.
\begin{theorem}[{\cite{auffingerchenrep,Jagannath_2015}}]
\label{thm:auffingerchen}
    Consider any polynomial $\xi$ with non-negative coefficients. Then $\OPT_\xi$ is achieved by some unique $\mu \in \mathscr{U}$.
\end{theorem}
Surprisingly, the maximum satisfying fraction of a typical instance $\cI \sim \cspfa$ is related to $\OPT_\xi$:
\begin{theorem}[{\cite[Corollary 1.2]{jmss22}, building on \cite{parisi79,talagrand2006parisi,panchenko2013parisi,panchenko2013sherrington,dms15,sen16,panchenko2017ksat}}]
Consider $\cI \sim \cspfa$. With high probability as $n \to \infty$, the maximum satisfying fraction of $\cI$ is\footnote{We use the notation $o_{x}(g(x))$ to mean the value is at most some function $\kappa(n,x)$ satisfying$ \frac{\kappa(n, x)}{g(x)} \to 0$ as $n \to \infty$, then $x \to \infty$.}
\begin{align*}
    \E[f] + \frac{\OPT_{\xi}}{\sqrt{\alpha}} + o_{\alpha}\left(\frac{1}{\sqrt{\alpha}}\right)\,, 
\end{align*}
where $\xi(s) \defeq \sum_{j=1}^r  \| f^{=j}\|^2 s^j$, and $\| f^{=j}\|^2$ is the Fourier weight of $f$ at degree $j$.
\end{theorem}
We also know that some algorithms are quantitatively obstructed due to the geometry of the solution space.
\begin{definition}[{Overlap-concentrated algorithm, \cite[Definition 2.1]{huang2022tight}}]
\label{defn:overlap-concentrated}
    A deterministic algorithm $\cA$ is \emph{overlap-concentrated} with respect to a random family of problems $\cF$ if the following is true for every $t \in [0,1]$ and constant $\delta > 0$: Draw two instances $\cI_1, \cI_2 \sim \cF$ that are $t$-correlated.\footnote{Instances $\cI_1, \cI_2$ are $t$-correlated if a $t$ fraction of constraints are identical, and the rest are independent; see also \cite[Definition 6.1]{jmss22}.}
    Let $\cA(\cI)$ be the assignment in $\{\pm 1\}^n$ that $\cA$ generates on instance $\cI$. Then the event
    \begin{align}
    \label{eqn:overlap-concentrated}
        \frac{1}{n}\left| \langle \cA(\cI_1), \cA(\cI_2) \rangle - \E_{\cI_1, \cI_2}[ \langle \cA(\cI_1), \cA(\cI_2) \rangle]  \right| \le \delta
    \end{align}
    occurs with high probability over choice of~$\cI_1, \cI_2$.
\end{definition}

For example, \emph{local} algorithms in the factors of i.i.d. model and certain quantum circuits up to logarithmic depth~\cite{jmss22} are overlap-concentrated w.h.p. over instances of $\cspfa$. A randomized algorithm is overlap-concentrated if \Cref{eqn:overlap-concentrated} holds additionally over its internal randomness, both w.h.p. and in the expectation.

We briefly explain the intuition of why \Cref{defn:overlap-concentrated} prevents algorithms from succeeding on certain optimization problems; see \cite{gamarniksurvey} for more detail. Suppose the optimal assignments are \emph{clustered}; i.e. they are ``needles in the haystack''. Then random instances will have ``needles'' in random locations in the solution space. But for $t \in [0,1]$, consider instances $\cI_{2}(t)$ that are $t$-correlated with the initial instance $\cI_1$. If the algorithm $\cA$ is successful, the correlation $\E_{\cI_{2}(t)}[\langle \cA(\cI_1), \cA( \cI_{2}(t)) \rangle]$ measures the correlation of these ``needles'' (i.e. good assignments), and so it varies smoothly from $0$ to $1$ as $t$ goes from $0$ to $1$. The \emph{overlap-concentrated} property is a bound on the variance of this quantity, meaning that with high probability over $\{\cI_{2}(t)\}_{t \in [0,1]}$, $\langle \cA(\cI_1), \cA( \cI_{2}(t)) \rangle$ varies smoothly with $t$ from $0$ to $1$. 
But this is not possible on problems that exhibit an \emph{overlap gap property} (OGP)\footnote{An \emph{overlap gap property} exists when an inner product, for example $\langle \cA(\cI_1), \cA( \cI_2) \rangle$, is forbidden to take values within a certain range $[a,b]$.} across pairs of instances; so $\cA$ cannot be successful at locating the ``needles'' or good assignments.

There is an emerging literature generalizing the \emph{overlap gap property} with increasingly strong quantitative obstructions~\cite{gamarnik2014limits,cgpr19,gjw20,huang2022tight}. We use the \emph{branching OGP} of~\cite{huang2022tight}, which gives the strongest quantitative bounds of algorithms known on spin glasses and random CSPs~\cite{jmss22}.

Given a polynomial $\xi$ with non-negative coefficients, let $\mathscr{L}_\xi$ denote the following set of functions:
\begin{equation*}
\mathscr{L}_\xi
\defeq \bigg\{ \mu: [0,1) \rightarrow \R_{\geq 0} : \, \mu \textup{ is right-continuous, } \| \xi'' \mu \|_{\text{TV}[0,t]} < \infty\ \forall t \in [0,1)\textup{, and }\int_{0}^{1} \xi''(t) \mu(t) \, dt < \infty \bigg\}\,.
\end{equation*}
Here, $\text{TV}[a,b]$ over an interval $[a,b]$ is the total variance over partitions:
\begin{align*}
\| f \|_{\text{TV}[a,b]} \defeq \sup_n  \sup_{a \le t_0<\dots<t_n \le b} \sum_{i\in[n]} |f(t_i) - f(t_{i-1})| \,.
\end{align*}
Furthermore, the definition of $\sfP_\xi$ over $\mathscr{U}$ can be naturally extended to $\mathscr{L}_\xi$, with infimum
\begin{align}
\label{eqn:alg_defn}
    \ALG_\xi \defeq \inf_{\mu \in \mathscr{L}_\xi} \sfP_\xi(\mu)\,.
\end{align}
Since $\mathscr{L}_\xi \supseteq \mathscr{U}$, we have $\ALG_\xi \le \OPT_\xi$.
It turns out that all \emph{overlap-concentrated} algorithms over $\cspfa$ are with high probability obstructed to some $\ALG_\xi$:

\begin{theorem}[{\cite[Corollary 6.10]{jmss22}}]
Choose any \emph{even} predicate $f$. 
For all $\epsilon > 0$ there is $\alpha_0 > 0$ such that the following is true:  For all $\alpha \ge \alpha_0$, for every \emph{overlap-concentrated} algorithm $\mathcal{A}$ over $\cspfa$, and with high probability over instances $\cI \sim \cspfa$ as $n \to \infty$, $\mathcal{A}$ returns an assignment to $\cI$ with satisfying fraction \emph{at most}
\begin{align*}
    \E[f] + \frac{\ALG_{\xi} + \epsilon}{\sqrt{\alpha}}\,,
\end{align*}
where $\xi(s) \defeq \sum_{j=1}^r \| f^{=j}\|^2 s^j$, and $\| f^{=j}\|^2$ is the Fourier weight of $f$ at degree $j$.
\end{theorem}
Note that if $\ALG_\xi \ne \OPT_\xi$, then these \emph{overlap-concentrated} algorithms are obstructed from reaching the optimal solution for typical instances of $\cspfa$.

\subsection{Statement of results}
The main result of our work is that there is an \emph{optimal} algorithm among \emph{overlap-concentrated} algorithms over $\cspfa$, subject to a technical condition. This matches analogous results for Ising spin glasses~\cite{ams20}, and very recently for multi-species spherical spin glasses~\cite{subag2019following,hs23_spherical_alg}.
\begin{assumption}
\label{assn:alg_minimizer_exists}
    For any polynomial $\xi$ with non-negative coefficients, $\ALG_\xi$ is achieved by some $\mu_\ast \in \mathscr{L}_\xi$. That is, the minimizer of \Cref{eqn:alg_defn} exists.
\end{assumption}
We include \Cref{assn:alg_minimizer_exists} because our algorithm uses the minimizer $\mu_\ast$.
As noted in \cite[journal version, Remark 2.4]{ams20},
\Cref{assn:alg_minimizer_exists} should not have a major impact on the algorithm's performance, since by an informal continuity argument, using a near-minimizer $\widetilde{\mu}_\ast \in \mathscr{L}_\xi$ instead of $\mu_\ast$ incurs negligible error. We further discuss this in \Cref{sec:nonlinearities}.
\begin{theorem}
\label{thm:main_theorem}
Suppose \Cref{assn:alg_minimizer_exists} holds.
Choose any predicate $f$ without linear-weight terms; i.e. $\| f^{=1}\|^2 = 0$. 
For all $\epsilon > 0$ there is $\alpha_0 > 0$ such that the following is true:  For all $\alpha \ge \alpha_0$, there exists a \emph{randomized} algorithm $\mathcal{A}$ such that $\mathcal{A}$ is \emph{overlap-concentrated} over $\cspfa$, and  with high probability over \emph{both} instances $\cI \sim \cspfa$ and internal randomness of $\mathcal{A}$ as $n \to \infty$, $\mathcal{A}$ returns an assignment to $\cI$ with satisfying fraction \emph{at least}
\begin{align*}
    \E[f] + \frac{\ALG_{\xi} - \epsilon}{\sqrt{\alpha}}\,,
\end{align*}
where $\xi(s) \defeq \sum_{j=1}^r  \| f^{=j}\|^2 s^j$, and $\| f^{=j}\|^2$ is the Fourier weight of $f$ at degree $j$.
\end{theorem}
This is analogous to work of~\cite{ams20}, which proves a similar result for mean-field spin glasses. This also extends the algorithm of~\cite{ams21}, which studies Maximum Cut on random regular graphs. In fact, all of these algorithms are related to the Approximate Message Passing (AMP) framework, popularized by the results of~\cite{bayati11}; see \cite{ams_survey} for a recent survey. It was suspected that such an algorithm could be written down for $\cspfa$~\cite{jmss22,gamarnik2023barriers}.

One immediate impact of this result is in quantum computing.  In the gate-based circuit model, many quantum circuits are \emph{overlap-concentrated} for instances of $\cspfa$ of depth up to $c \cdot \log n$ for some $c$~\cite{chou2022limitations}.
This work implies that such quantum circuits, even up to $c \cdot \log n$ depth, cannot outperform classical algorithms on typical instances of a random CSP.
\begin{corollary}[Failure of quantum advantage]
\label{cor:no_quantum_advantage}
Suppose \Cref{assn:alg_minimizer_exists} holds.
Choose any \emph{even} predicate $f$. 
For all $\epsilon > 0$ there exist $\alpha_0 > 0$ and $c: \mathbb{R}^{> 0} \to \mathbb{R}^{> 0}$ such that the following is true:  For all $\alpha \ge \alpha_0$, and
with high probability over instances $\cI \sim \cspfa$ as $n \to \infty$,
\emph{no quantum circuit} with depth up to $c(\alpha) \cdot \log n$ 
that applies 2-local gates on a fixed, bounded-degree architecture
can produce an assignment with satisfying fraction that outperforms classical algorithms by more than $\frac{\epsilon}{\sqrt{\alpha}}$.
\end{corollary}
Even until now, the promise of quantum advantage for optimization problems has been unclear. Multiple times, improved analyses of the Quantum Approximate Optimization Algorithm (QAOA)~\cite{farhi2014quantum} suggesting a such an advantage~\cite{farhi2015quantum,wang2018,ryananderson2018quantum,basso2021quantum} led to better classical algorithms~\cite{barak2015beating,hastings2019classical,Marwaha_2021,barak21,ams21}.
This work shows that for the QAOA and many other quantum circuits up to logarithmic depth, there can be \emph{no quantum advantage} on random CSPs with even predicates of sufficiently large clause density. Note that unlike related work on the QAOA~\cite{farhi2020quantumtypical,basso_spinglass_qaoa,Anshu2023concentrationbounds}, this rules out quantum advantage \emph{regardless} of the presence of overlap gap phenomena.

We remark on the technical condition of \emph{even predicates}. This can likely be removed with two steps. First, the obstruction of \emph{overlap-concentrated} algorithms only applies to \emph{even} predicates. New techniques to prove OGPs for spherical spin glasses~\cite{huang2023algorithmic} can possibly be ported to all Ising spin glasses; these results would immediately transfer to random instances of $\cspfa$ via~\cite{jmss22}. Second, the algorithm we analyze applies when predicates have no linear-weight terms; adjustments taken in~\cite{sellke2021optimizing} will likely generalize this algorithm. These steps would extend \Cref{cor:no_quantum_advantage} to all predicates $f$.

Previous works noting limitations of logarithmic-depth quantum algorithms on sparse random optimization problems~\cite{farhi2020quantumtypical,farhi2020quantumworst,chou2022limitations} also rely on the concept of \emph{overlap gap property}. 
As in~\cite{ams20,hs23_spherical_alg}, this work shows the existence of a family of random optimization problems with an \emph{optimal} algorithm up to the OGP barrier. 
All of the optimal algorithms are classical, which hints that quantum algorithms that cannot pass the OGP barrier offer no additional power for classical optimization in general.

\subsection{Proof overview}
Our proof builds on a number of other works, especially~\cite{montanari2019optimization,ams20,ams21}.
We outline the proof here.

In \Cref{sec:localalgdefns}, we define our model of CSP and describe the algorithm. Instead of $\cspfa$, we run the algorithm on a CSP whose associated constraint hypergraph is locally a regular \emph{hypertree}.\footnote{In the hypergraph, the variables are vertices and the constraints are hyperedges. We show this change has negligible effect in~\Cref{sec:indexregular_vs_regular_vs_avgdegree}.} 
This algorithm runs in \emph{iterations}, where on each iteration, every vertex recomputes its value using information from adjacent vertices; we parameterize this algorithm by the functions (two per \emph{iteration number}) that recompute each vertex's value.

We can already study the performance of the parameterized algorithm, analyzing how information propagates locally through the hypertree.
In \Cref{sec:independence}, we generalize a set of $\sigma$-algebras from~\cite{ams21} to describe this behavior. 
This allows us to analyze the correlation of a variable from the algorithm with any other variable from a particular $\sigma$-algebra. With these tools in hand, we simplify the expected performance of the parameterized algorithm, deferring some detailed calculations to \Cref{sec:appendix_moments}.

We then show how the algorithm, with the right choice of parameters, solves a stochastic optimal control problem that is dual to the minimization problem in \Cref{eqn:alg_defn}.
\begin{enumerate}
    \item First, we show that with high probability as $\alpha \to \infty$, the intermediate values in the algorithm are asymptotically \emph{Gaussian}. We prove this in~\Cref{sec:stateevo} with an argument known as \emph{state evolution}. Notably, this argument works for sparse models (as in~\cite{ams21}), and for message-passing on hypergraphs, where some values are products of Gaussians. This allows us to study quantities in the algorithm output as \emph{discrete} stochastic processes. We defer the calculation of moments to \Cref{sec:appendix_moments}, and the full proof of state evolution to \Cref{sec:appendix_stateevo}.
    \item Then, we take the number of iterations to a very large constant while proportionally decreasing the step size. Under this process, each discrete stochastic process asymptotically approximates a \emph{continuous} stochastic process. In \Cref{sec:nonlinearities}, we show how to choose the optimal parameters. We also show that rounding the algorithm's output to a value in $\{\pm 1\}^n$ does not meaningfully affect its performance.
\end{enumerate}

This work extends \cite{ams21} in a spiritually similar way that \cite{ams20} extends \cite{montanari2019optimization}; roughly, we extend from graphs to \emph{hypergraphs}.
This requires more involved calculations and in some cases new arguments.
For example, we require a more complicated proof of state evolution, since the messages in our algorithm are not Gaussian at fixed degree. Furthermore, our proof of rounding is new, since we cannot rely on spectral statistics of random \emph{graphs}.
Finally, our analysis uses an additional symmetry we call \emph{index-regularity} (defined in \Cref{sec:localalgdefns}) which, to the best of our knowledge, does not come up in prior work. We show that enforcing this symmetry has negligible effect in \Cref{sec:indexregular_vs_regular_vs_avgdegree}, and use this symmetry to simplify some of our analysis.

There are a number of quantities to keep track of.
We annotate the important symbols in \Cref{table:notation}.
Following notation in~\cite{ams21}, we typically use lowercase for a quantity in the algorithm (e.g. $u^\ell_{i \to a}$), uppercase with step size $\delta$ for a \emph{discrete} stochastic process (e.g. $U_\ell^\delta$),
 and uppercase with time $t$ for a \emph{continuous} stochastic process (e.g. $U_t$).
\begin{table}[t]
\centering
\begin{tabular}{|l|l|l|l|l|}
\hline
\textbf{Term}                                 & \textbf{\cite{montanari2019optimization}} & \textbf{\cite{ams20}} & \textbf{\cite{ams21}} & \textbf{(our work)}    \\ \hline
Brownian motion (asymptotically independent)                               & $u$                                       & $\Delta$              & $u$                   & $u$                    \\ \hline
sum of Brownian motions &        --                                   & $z$                   &              --         &  $w$          \\ \hline
nonlinearity controller                & $x$                                       & $x$                   & $x$                   & $x$                    \\ \hline
value of spin                                 & $z$                                       & $m = f(\ldots)$        & $z$                   & $z$                    \\ \hline
\end{tabular}  
\caption{
\label{table:notation}
\footnotesize List of important terms used in the analysis of message-passing algorithms for mean-field spin glasses and random CSPs.
Terms in the first row asymptotically approach
independent Gaussians after applying state evolution, and so constitute the \emph{Brownian motion} (up to normalization) in the stochastic picture.
}
\end{table}

\subsection{Related work}

\paragraph{Approximate Message Passing}

Approximate Message Passing (AMP) is an algorithmic framework used to design locally computable, iterative methods for a variety of high-dimensional statistical inference and recovery tasks~\cite{lesieur2017constrained, rush2017capacity, mondelli2021approximate, maskey2022generalization, baranwal2023optimality}.
Many of the ideas motivating AMP originate from the study of \emph{belief propagation}, an iterative method rediscovered multiple times in different scientific contexts; see~\cite{mezard2009information, koller2009probabilistic} for a broader history.
AMP was first introduced in the context of compressed sensing and signal recovery~\cite{donoho2009message, kabashima2003cdma}, partly inspired by the use of structured priors in decoding low-density parity check codes~\cite{richardson2008modern, montanari2012graphical}.
In spin glass theory, AMP is synonymous with computing the fixed points of a certain Thouless--Anderson--Palmer (TAP) free energy functional~\cite{thouless1977solution, krzakala2012probabilistic, zdeborova2016statistical, chen2018tap, chen2023generalized}.
In fact, our algorithm is inspired by similar connections between the TAP equations and belief propagation~\cite{mezard1987spin, mezard2009information}.

A beneficial feature of AMP is that one can rigorously analyze the asymptotic behavior of statistics in the framework. This is typically done through the use of \emph{state evolution} statements~\cite{bolthausen2014iterative, bayati11}: laws of large numbers which dictate the behavior of such statistics when the dimension in the problem approaches infinity.
This has been subsequently used to compute approximate ground states of both Ising and spherical mean-field spin glasses~\cite{montanari2019optimization,ams20,huang2023algorithmic}, as well as Max-Cut on random regular graphs of large degree~\cite{ams21}. 
See the surveys of~\cite{feng2021unifying, ams_survey} for further discussion on AMP and its variants.

\paragraph{Ising spin glasses and CSPs}
Although mean-field spin glasses with Ising spins were proposed as mathematical models of materials~\cite{sherrington1975solvable}, their connection to combinatorial optimization has been observed for some time~\cite{Fu_1986,mezard1987spin}. 
More recently, a duality has emerged between Ising spin glasses and sparse random CSPs at large enough clause density; this started with relating the supremum of the two models \cite{dms15,sen16,panchenko2017ksat}. This supremum was calculated using a novel heuristic argument known as ``replica symmetry breaking'' \cite{parisi79}, and proven to be correct in \cite{talagrand2006parisi,panchenko2013parisi} (see also \cite{panchenko2013sherrington}). Moreover, the duality preserves a phenomenon known as \emph{overlap gap}~\cite{gamarnik2014limits}, a generalization of solution clustering that obstructs certain algorithms; see \cite{cgpr19,gjw20,gjw21,chou2022limitations,huang2022tight,jmss22}, and the expositions of \cite{gamarniksurvey,huang2022computational}. Notably, a strong version of overlap gap phenomena~\cite{huang2022tight,huang2023algorithmic,hs23_spherical_alg} is connected to the message-passing algorithms of~\cite{montanari2019optimization,ams20,ams21}, and this work.
Other topics at the intersection of spin glasses and CSPs include approximating partition functions (e.g. \cite{sly2012computational}), and computing thresholds for random CSPs (e.g. \cite{ding2013maximum,ding2021proof}).

\paragraph{Quantum algorithms for classical optimization}
Recent demonstrations of quantum advantage (e.g. \cite{Arute_2019,ustc_advantage}) have led to increased excitement around using contemporary (i.e. NISQ~\cite{Preskill_2018}) quantum computers for a practically relevant problem.
Although classical optimization is such a problem, quantum algorithms here are not proven to have any benefit over classical algorithms. The most popular proposal is the Quantum Approximate Optimization Algorithm (QAOA)~\cite{farhi2014quantum}. 
Earlier suggestions of quantum advantage with the QAOA~\cite{farhi2015quantum,wang2018,ryananderson2018quantum,Wurtz_2021} led to better classical algorithms~\cite{barak2015beating,hastings2019classical,Marwaha_2021,barak21}; note that all of these comparisons were for very small circuit depth.
Recently, \cite{farhi2019quantum,basso2021quantum, basso_spinglass_qaoa} find an expression for the performance of the QAOA on spin glasses and sparse CSPs at any constant depth.

The QAOA is a \emph{local} algorithm. It was noted that this feature can limit its performance~\cite{farhi2020quantumtypical,farhi2020quantumworst}. The authors in \cite{chou2022limitations} unify definitions of quantum and classical local algorithms to show that the QAOA is obstructed by overlap gap phenomena on sparse problems. Because of this, \cite{ams21} rules out quantum advantage through $\epsilon \cdot \log n$-depth QAOA on random regular Max-Cut instances of large degree (up to \Cref{assn:alg_minimizer_exists}); our work generalizes this fact to CSPs with even predicates.
Regardless of theory, there are many recent demonstrations of the QAOA on state-of-the-art quantum computing systems; see for example \cite{ebadi2022quantum, shaydulin2023qaoa, maciejewski2023design}.
It is also known that the QAOA can solve satisfiability faster than brute force~\cite{boulebnane2021}; see also \cite{dalzell2022mind} for a related protocol with a more-than-quadratic improvement.
Some quantum algorithms for machine learning are known to provide no quantum advantage~\cite{tang_dequantize_comment}.

\section{Local algorithms for sparse CSPs}
\label{sec:localalgdefns}

A CSP instance has an implicit \emph{directed} hypergraph, with variables as vertices and constraints as hyperedges. 
\begin{definition}[Directed hypergraph of a CSP]
    Consider a $r$-uniform CSP instance $\cI$ with $n$ variables and constraints $E(\cI)$, where each constraint $e \in E(\cI)$ involves $r$ (ordered) variables in $[n]$. Then $\cI$ is associated with the directed hypergraph with vertices $[n]$ and hyperedges $E(\cI)$; i.e. $G([n], E(\cI))$.
\end{definition}
We consider a hyperedge \emph{directed} when the involved vertices are ordered.
For example, if constraint $e$ involves vertices $\{v_1, v_2, v_3\}$
in order $[2,3,1]$, then the associated directed hyperedge is $[v_2, v_3, v_1]$. We denote this ordered list as $\partial e$. We often refer to the CSP instance by its hypergraph. The instance is fully specified by its hypergraph, the CSP function $f$, and the (ordered) randomness $[\epsilon_{e_1}, \dots, \epsilon_{e_r}]$ on each constraint $e \in E(\cI)$.
For each variable $i$, we use $\partial i$ to denote the set of constraints  $i$ participates in.

The \emph{factor graph} of a hypergraph $G(V,E)$ is a bipartite graph of order $|V| + |E|$, with edges between each hyperedge $e \in E$ and each vertex in $\partial e$.
For more information on the factor graph formalism, see~\cite{kschischang2001factor}.
Let $\B_i(\ell)$ be the induced subgraph of $G$ given by the $\ell$-local neighborhood around vertex $i \in V$.\footnote{Note that we define distance as the number of \emph{hypergraph vertices} in the shortest path; vertices that share a hyperedge are distance $1$ apart.}
\begin{definition}
A finite hypergraph $G(V,E)$ is \emph{$(\ell, \epsilon)$-locally treelike} if the factor graph of $\B_i(\ell+1)$ is isomorphic to a tree for at least $(1-\epsilon) \cdot n$ vertices $i \in V$.
\end{definition}
This definition is a generalization of locally treelike graphs used in~\cite{ams21}; see \Cref{fig:hypertree_factorgraph} for an illustration.

\begin{figure}[t]
\centering
\includegraphics[width=0.4\textwidth]{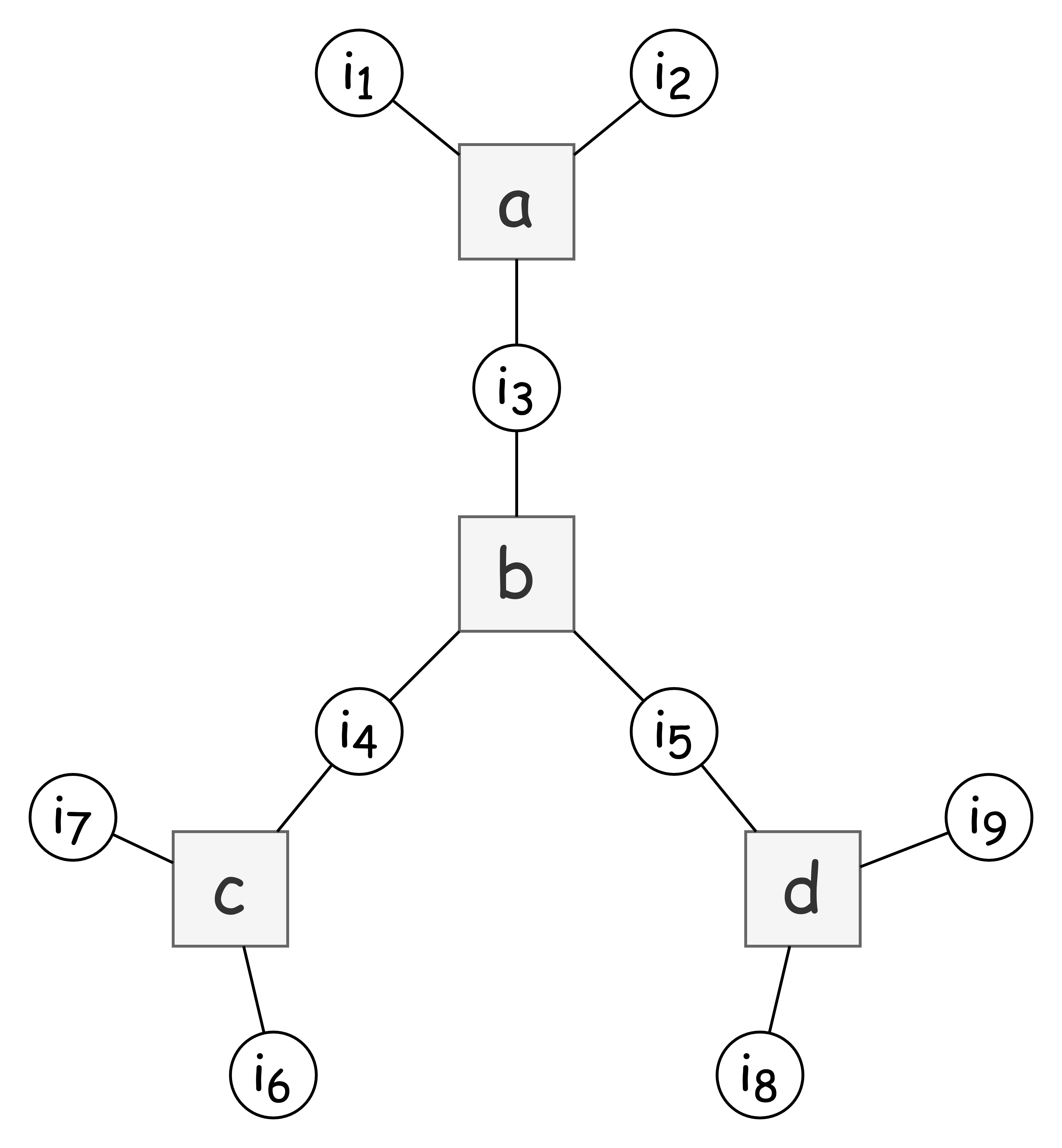}
\caption{\footnotesize Factor graph of the hypergraph with vertices $\{i_1, \dots, i_9\}$ and hyperedges $\{a,b,c,d\}$. Note that this hypergraph is locally treelike.}
\label{fig:hypertree_factorgraph}
\end{figure}

Every \emph{message-passing} algorithm is contained in the model of local algorithms presented in~\cite{chou2022limitations}. In this model, the value associated with each variable is uniquely determined by its local neighborhood. 
\begin{definition}[Generic $p$-local algorithms, {\cite[Definition 3.1]{chou2022limitations}}]
Consider a randomized algorithm $\cA$ that inputs a $r$-uniform hypergraph $G(V,E)$ and outputs a label in an alphabet $\Sigma$ to each vertex in the graph, i.e. $\cA(G) \in \Sigma^V$.
Then $\cA$ is \emph{generic $p$-local} if for any set $S \subseteq V$, the following conditions hold:
\begin{enumerate}
    \item \emph{(Local distribution determination)} The joint marginal distribution $(\cA(G)_v)_{v \in S}$ is identical to $(\cA(G')_v)_{v \in S}$, where $G' = \bigcup_{v \in S} \B_v(p)$.
    \item \emph{(Local independence)} The distribution $\cA(G)_v$ is statistically independent of the joint distribution of $\cA(G)_{v'}$ over all vertices $v' \notin \B_v(2p)$.
\end{enumerate}
\end{definition}
This model also includes the Quantum Approximate Optimization Algorithm~\cite{farhi2014quantum} on typical instances of $\cspfa$. In fact, any depth-$p$ quantum circuit with a fixed architecture of bounded degree can be seen as generic $p$-local by including the architecture in the input hypergraph.
\cite{chou2022limitations} proves that for all $p \le c(\alpha) \cdot \log n$, generic $p$-local algorithms are overlap-concentrated over instances of $\cspfa$ with high probability as $n \to \infty$.

In our analysis, we make some convenient modifications to the input.
We show that these changes have negligible impact on the algorithm's performance in \Cref{sec:indexregular_vs_regular_vs_avgdegree}; see also~\cite{yoon2011belief,dms15,sen16,cgpr19,ams21}.
First, since our algorithm is local and the instances are locally treelike on all but a negligible fraction of vertices\footnote{In particular, for any finite $\ell$, with high probability our instances are $(\ell, \epsilon)$-locally treelike for some $\epsilon$, where $\epsilon \to 0$ as $n \to \infty$.}, we analyze the algorithm on an instance which is locally a \emph{hypertree}, i.e. a hypergraph whose factor graph is a tree. 
We also make the hypergraph \emph{regular}, by embedding our instance of $\cspfa$ in a locally treelike $d$-regular instance with $d = r \cdot \alpha + o_{\alpha}(1)$.

On hypergraphs, we need one additional property, which we call \emph{index-regular}. 
Here, we mean that for each $\iota \in [r]$, each variable is used in  \emph{index} $\iota$ of a constraint the same number of times.
\begin{definition}[$d$-index-regular CSP]
Consider a ($r$-uniform) CSP instance $\cI$ with $n$ variables and constraints $E(\cI)$, such that each constraint $e \in E(\cI)$ has associated variables $\partial e = [v_{e,1}, \dots, v_{e,r}]$.
Then $\cI$ is $d$-\emph{index-regular} if every variable $v \in [n]$ is involved in $d$ constraints, and moreover, for each $\iota \in [r]$,  $v$ is the $\iota^{\text{th}}$-indexed variable in exactly $\frac{d}{r}$ constraints.
\end{definition}

\begin{figure}[t]
\setlength\fboxrule{2pt}
\setlength\fboxsep{3mm}
{\centering
\fbox{
\parbox{44.5em}{
\textbf{Input:} The timestep $\delta > 0$, and a $d$-regular $r$-uniform  hypergraph with girth at least $2 \lfloor \frac{1}{\delta} \rfloor + 2$.

\textbf{Parameters:} The CSP function $f: \{\pm 1\}^r \to \{0,1\}$, constant $K > 0$, and nonlinearities
\begin{align*}
\big\{ F_{\to,\ell}: \R^{\ell} \rightarrow [-K,K] \big\}_{\ell = 0, \ldots, L-1}
\qquad\qquad
\big\{ F_{\ell}: \R^{\ell} \rightarrow [-K,K] \big\}_{\ell = 0, \ldots, L-1}
\end{align*}

\textbf{Algorithm:}
\vspace{-2mm}
\begin{enumerate}
    \item Initialize the following:
\begin{quote}
\vspace{-2mm}
\begin{enumerate}[a.]
\item $z_i^0 \sim \cN(0, \delta)$ for all $i \in V$ and $z_{i \rightarrow a}^0 = z_i^0$ for all $i \in \partial a$ and $a \in E$.
\item $w_i^0 = 0$ for all $i \in V$ and $w_{i \rightarrow a}^0 = 0$ for all $i \in \partial a$ and $a \in E$.
\end{enumerate}
\end{quote}
    \item For each iteration $\ell = 0, \ldots, L-1$ where $L = \lfloor \frac{1}{\delta} \rfloor$:
\begin{quote}
\vspace{-2mm}
\begin{enumerate}[a.]
\item Let $A_{i \rightarrow a}^{\ell} = F_{\to,\ell}\big( u_{i \rightarrow a}^{1}, \ldots, u_{i \rightarrow a}^{\ell} \big)$, and let $A_{i}^{\ell} = F_{\ell}\big( u_{i}^{1}, \ldots, u_{i}^{\ell} \big)$.
\item For each factor $a \in E$ and involved variable $i \in \partial a$, compute:
\begin{align*}
w_{i \rightarrow a}^{\ell + 1}
&= \frac{1}{\sqrt{d-1}} \cdot \sum_{b \in \partial i \setminus a}
\D_{i;b} f(\vz_{\partial b \to b}^\ell) \\
u_{i \rightarrow a}^{\ell + 1}
&= w_{i \rightarrow a}^{\ell + 1} - w_{i \rightarrow a}^{\ell} \\
z_{i \rightarrow a}^{\ell + 1}
&= z_{i \rightarrow a}^{0} + \sum_{s=1}^{\ell+1} A_{i \rightarrow a}^{s-1} \cdot u_{i \rightarrow a}^{s}
\end{align*}

\item For each variable $i \in V$, compute:
\begin{align*}
w_{i}^{\ell + 1}
&= \frac{1}{\sqrt{d}} \cdot \sum_{b \in \partial i}
\D_{i;b} f(\vz_{\partial b \to b}^\ell)  \\
u_{i}^{\ell + 1}
&= w_{i}^{\ell + 1} - w_{i}^{\ell} \\
z_{i}^{\ell + 1}
&= z_{i}^{0} + \sum_{s=1}^{\ell+1} A_{i}^{s-1} \cdot u_{i}^{s}
\end{align*}
\end{enumerate}
\end{quote}
\end{enumerate}

\textbf{Output:} $\sgn(\vz^L)$, where $\sgn$ is applied entrywise.
}}
\par}
\caption{\footnotesize Description of the message-passing algorithm.}
\label{fig:algorithm}
\end{figure}

We describe our choice of algorithm in~\Cref{fig:algorithm}. This involves some notation we describe here. 
A function that depends on zero inputs is a constant; i.e. $A^0_i = F_0$ and $A^0_{i \to a} = F_{\rightarrow,0}$.
For any scalar quantity that depends on vertex, a bolded version represents a vector quantity.
For example, given $\partial a = [i_1, \dots, i_r]$, $\vz^\ell_{\partial a\to a}$ is the \emph{ordered} list $[z^\ell_{i_1 \to a}, \dots, z^\ell_{i_r \to a}]$, and $\vz^\ell_{\partial a}$ is the \emph{ordered} list $[z^\ell_{i_1}, \dots, z^\ell_{i_r}]$.
Similarly, $\vz^L$ is the vector in $\mathbb{R}^V$ such that $\big( \vz^L \big)_i = z_i^L$.

Our analysis makes extensive use of the partial derivative operator.
For any factor $a$ and variable $j \in \partial a$, let $\D_{j;a} f$ be the partial derivative of $f$ at the \emph{coordinate} $\iota$ such that $j$ is the $\iota$th coordinate of $\partial a$.
For example, if $\partial a = [k,i,j]$ and $f(x_1, x_2, x_3) = x_1 x_2 + x_2 x_3$, then $\D_{k;a} f(x_1, x_2, x_3) = x_2$ and $\D_{i;a} f(x_1, x_2, x_3) = x_1 + x_3$. Note the following:
\begin{remark}
\label{remark:f_derivative_doesnt_use_i}
    Since the domain of $f$ is $\{\pm 1\}^r$, $f$ has a representation as a multilinear polynomial. As a result, for all factors $b$ and variables $i \in \partial b$, $\D_{i;b} f$ is independent of the entry with the same coordinate as $i$ in $\partial b$.
\end{remark}

We additionally require the nonlinearities to satisfy the following two properties:

\begin{property}[Similar node and node-to-factor evolution]
\label{hyp:A_deviation_small}
    For every $\ell \ge 0$, variable $i$, factor $a \in \partial i$, and $m \in \mathbb{N}$, we have $\E[|A_i^\ell - A_{i \to a}^{\ell}|^m] = O_d(\frac{1}{d^{m/2}})$.
\end{property}

\begin{property}[Normalization of second moment]
\label{hyp:second_moment}
    For every $\ell \geq 0$, variable $i$, and factor $a \in \partial i$, we have $\E[(A_i^{\ell})^2] = \E[(A_{i \to a}^{\ell})^2] = \frac{\delta}{\nu_{\ell + 1}}$, where $\nu_\ell \defeq \frac{\xi'(\ell \delta) - \xi'((\ell - 1)\delta )}{r}$, and $\xi$ is defined as in \Cref{thm:main_theorem}.
\end{property}

In~\Cref{sec:nonlinearities}, we show how to choose the nonlinearities based on~\Cref{eq.parisi-equations}, and show that our choice indeed satisfies \Cref{hyp:A_deviation_small} and \Cref{hyp:second_moment}.

\begin{remark}
    Our analysis of the algorithm above holds even if the CSP predicate depends on factor $a$, so long as the absolute value of each Fourier component is fixed (this implies that all $f_a$ are associated with the same $\xi$). For example, this implies that the algorithm achieves the same value for \emph{every} choice of signs $\{ \epsilon_{i,\iota} \}_{i \in [n], \iota \in [r]}$ in \Cref{defn:randomcsp}.
\end{remark}

\section{Independence claims}
\label{sec:independence}

Given random variables $\{X_1,\dots, X_k\}$, let $\sigma(X_1,\dots,X_k) = \sigma(\{X_1,\dots,X_k\})$ be the \emph{$\sigma$-algebra} generated by $\{X_1, \ldots, X_k\}$. 
For an introduction to $\sigma$-algebra and related concepts, see for example~\cite{durrett}.

Consider the algorithm in~\Cref{fig:algorithm}.
For every $0 \le \ell \le L$,\footnote{Recall that because the instance has girth at least $2L+2$, all local neighborhoods of radius up to $L$ are treelike.} variable $v$, and factor $a \in \partial v$, we define
\begin{align*}
    \mathcal{G}_v^{\ell} &\defeq \sigma\left( \left \{ z_i^0 \ \mid \ i \in \B_v(\ell) \right \}\right)\,,
    &
   \mathcal{G}_{v \to a}^{\ell} &\defeq \sigma\left( \left \{ z_i^0 \ \mid \ i \in \B_{v \to a}(\ell) \right \}\right)\,.
\end{align*}

Intuitively, these $\mathcal{G}_v^\ell$ algebras are the random variables determined by the initial Gaussian variables in the radius-$\ell$ neighborhood around $v$; here, $\ell$ is a \emph{distance}. 
We'll reuse the notation $\B_v(\ell)$ as the set of \emph{nodes} of distance\footnote{Once again, we consider \emph{distance} of $i$ and $j$ as the number of \emph{nodes} required to pass from $i$ to $j$; the factors are not counted.} at most $\ell$ from $v$ in the corresponding hypergraph. We define the set $\B_{v \to a}(\ell)$ as the set of nodes of distance at most $\ell$ from $v$, \emph{excluding} nodes that reach $v$ through the factor $a$.
We formalize a series of information-theoretic claims using these $\sigma$-algebras, similar to~\cite[Lemma 2.2]{ams21}. See also an illustration in \Cref{fig:independence_factorgraph}.

\begin{figure}[t]
\centering
\includegraphics[width=0.6\linewidth]{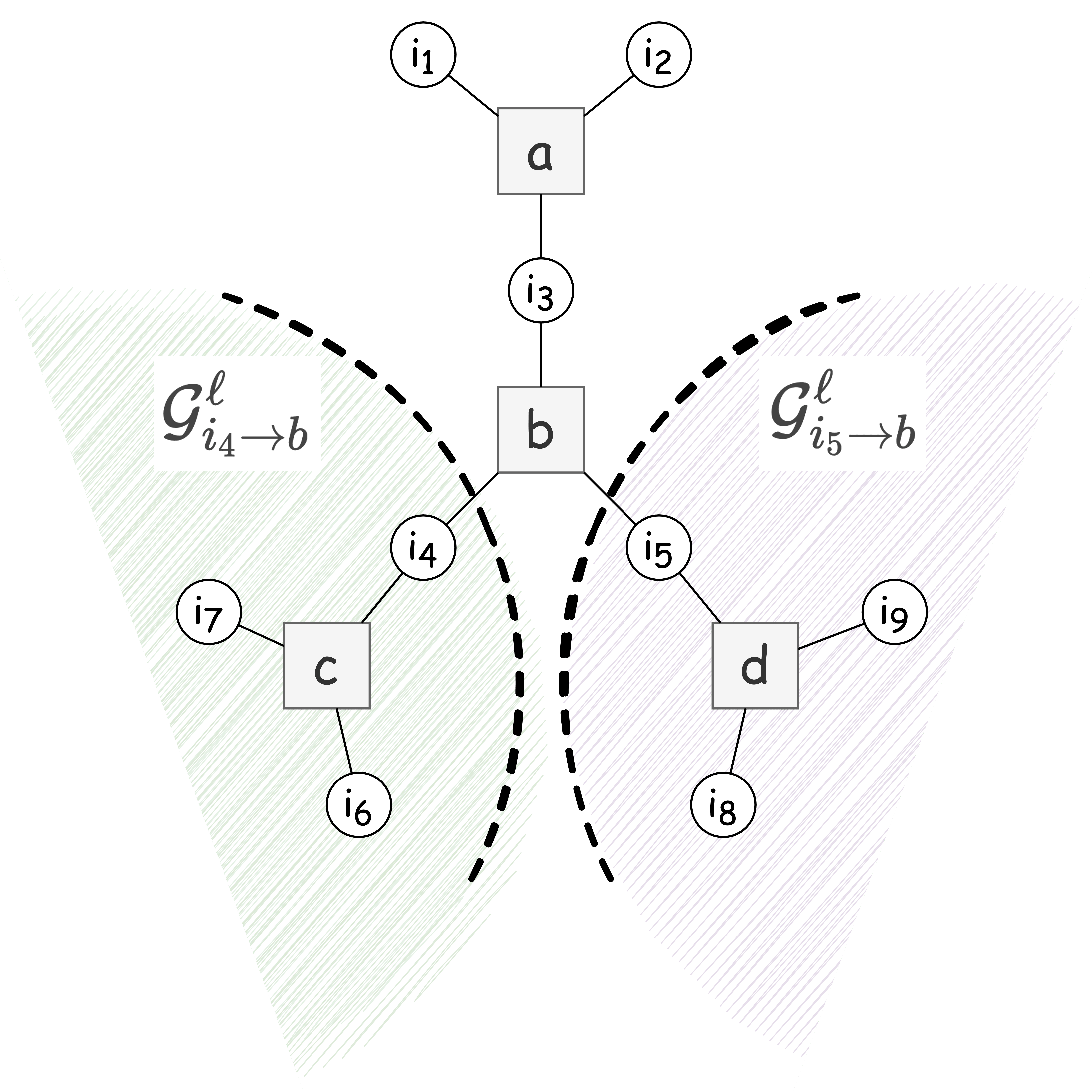}
\caption{\footnotesize An illustration of $\sigma$-algebras on the hypertree. Since the path from $i_4$ to $i_5$ goes through factor $b$, $\cG_{i_4 \to b}^{\ell}$ and $\cG_{i_5 \to b}^{\ell}$ are independent by \Cref{lemma:SigmaAlgebraIndependence}. In addition, $\cG_{i_3 \to a}^{\ell + 1} = \sigma\left(\cG_{i_3}^{0} \bigcup \cG_{i_4 \to b}^{\ell} \bigcup \cG_{i_5 \to b}^{\ell}\right)$ by \Cref{lemma:SigmaAlgebraDecomposition}.}
\label{fig:independence_factorgraph}
\end{figure}

The first fact about independence informally says that node-to-factor messages never use information on ``the other side'' of the factor; once information goes across $a$, it never comes back.
\begin{lemma}\label{lemma:SigmaAlgebraIndependence}
For every $0 \le \ell_1, \ell_2 \le L$, distinct variables $i,j$, and factors $a \in \partial i, b \in \partial j$, if every path of length at most $2L$ from $i$ to $j$ goes through $a$ and $b$, then $\mathcal{G}^{\ell_1}_{i \to a}$ and $\mathcal{G}^{\ell_2}_{j \to b}$ are independent.
\end{lemma}
\begin{proof}
     Recall that all neighborhoods of radius up to $L$ are locally treelike. Since every path of length at most $2L$ from $i$ to $j$ goes through $a$ and $b$, it implies that $\B_{i \to a}(\ell_1)$ and $\B_{j \to b}(\ell_2)$ are disjoint. It follows that $\mathcal{G}^{\ell_1}_{i \to a}$ and $\mathcal{G}^{\ell_2}_{j \to b}$ are independent since they are generated by disjoint sets of independent random variables. 
\end{proof}
\Cref{lemma:SigmaAlgebraIndependence} holds even if the factors are the same, i.e. $a = b$. We refer to this case as \Cref{lemma:SigmaAlgebraIndependence} \emph{around} factor $a$.

The next fact concerns the geometry of local neighborhoods. It allows us to decompose each $\sigma$-algebra into $\sigma$-algebras with distance reduced by one, which is handy for induction.
\begin{lemma}\label{lemma:SigmaAlgebraDecomposition}
    For every $1 \le \ell \le L$, variable $i$, and factor $a \in \partial i$, we have 
    \begin{align*}
        \mathcal{G}_i^{\ell} &= \sigma \left(\mathcal{G}_{i \to a}^{\ell} \cup \left( \bigcup_{j \in \partial a \backslash i} \mathcal{G}_{j \to a}^{\ell-1} \right)\right)\,,
        &
        \mathcal{G}_{i \to a}^{\ell} &= \sigma\left(\mathcal{G}_{i}^{0} \cup  \left(\bigcup_{b \in \partial i \backslash a} \bigcup_{v \in \partial b \backslash i} \mathcal{G}_{v \to b}^{\ell - 1}\right) \right)\,.
    \end{align*}
\end{lemma}
\begin{proof}
    The lemma follows from the equivalent statements about decomposing local neighborhoods; i.e. 
    \begin{align*}
        \B_i(\ell) &= \B_{i \to a}(\ell) \cup \left( \bigcup_{j \in \partial a \backslash i} \B_{j \to a}(\ell - 1) \right)\,,
        &
        \B_{i \to a}(\ell) &= \{i \} \cup \left( \bigcup_{b \in \partial i \backslash a} \bigcup_{v \in \partial b \backslash i} \B_{v \to b}(\ell - 1) \right) \,.
            \tag*{\qedhere}
    \end{align*}
\end{proof}
The next statement sorts quantities appearing in the algorithm (\Cref{fig:algorithm}) into appropriate $\sigma$-algebras. It proves the intuition that messages at iteration number $\ell$ are determined by the local neighborhood of distance at most $\ell$:
\begin{lemma}
\label{lemma:u_is_in_G}
For every $1 \le \ell \le L$, variable $i$, and factor $a$, we have $u_{i \to a}^\ell, w_{i \to a}^\ell, z_{i \to a}^\ell \in \mathcal{G}_{i \to a}^{\ell}$ and $u_{i}^\ell, w_{i}^\ell, z_{i}^\ell \in \mathcal{G}_{i}^{\ell}$.
\end{lemma}
\begin{proof}
    We only prove the first claim, since the second claim has essentially the same proof. We proceed by induction on $\ell$. When $\ell = 1$, we have
    \begin{align*}
    u_{i \to a}^1 = w_{i \to a}^1 = \frac{1}{\sqrt{d-1}}\cdot \sum_{b \in \partial i \backslash a}
    \D_{i;b} f(\vz_{\partial b \to b}^0)\,.
    \end{align*}
    The right-hand side is only a function of variables in $\B_{i\to a}(1)$, so $u_{i \to a}^1 = w_{i \to a}^1 \in \mathcal{G}^{1}_{i \to a}$. For $z_{i \to a}^{1}$, we have
    \begin{align*}
    z_{i \to a}^{1} = z_{i}^0 + A^0_{i \to a} u_{i \to a}^{1} \in \mathcal{G}^{1}_{i \to a}\,,
    \end{align*}
    since $A^0_{i \to a}$ is a constant. For the inductive step, we have
    \begin{align*}
    w^{\ell + 1}_{i \to a} = \frac{1}{\sqrt{d - 1}} \cdot \sum_{b \in \partial i \backslash a}
    \D_{i;b} f(\vz_{\partial b \to b}^\ell)\,.
    \end{align*}
    By \Cref{remark:f_derivative_doesnt_use_i}, each term in the right-hand side is only a function of $z_{j \to b}^{\ell}$ for $j \in \partial b \backslash i$.
    By the inductive hypothesis, $z_{j \to b}^{\ell} \in \mathcal{G}_{j \to b}^{\ell}$. Furthermore, by \Cref{lemma:SigmaAlgebraDecomposition}, we know that $\mathcal{G}_{j \to b}^{\ell} \subseteq \mathcal{G}_{i \to a}^{\ell + 1}$ for $b \in \partial i \backslash a$ and $j \in \partial b \backslash i$. This implies that $w^{\ell + 1}_{i \to a} \in \mathcal{G}_{i \to a}^{\ell + 1}$ and therefore $u^{\ell + 1}_{i \to a} = w^{\ell + 1}_{i \to a}  - w^{\ell}_{i \to a} \in \mathcal{G}_{i \to a}^{\ell + 1}$. Finally, since $A^{s}_{i \to a}$ is a function of $u^{1}_{i \to a}, \ldots, u^{s}_{i \to a}$,
    \begin{align*}
    z_{i \to a}^{\ell + 1} = z_{i}^0 + \sum_{s = 1}^{\ell + 1} A_{i \to a}^{s-1} u_{i \to a}^{s} \in \mathcal{G}^{\ell + 1}_{i \to a}.
    \tag*{\qedhere}
    \end{align*}
\end{proof}
\Cref{lemma:u_is_in_G} implies that the $\ell^\text{th}$ node message at variable $i$ is determined by the initial Gaussian variables in the radius-$\ell$ neighborhood around $i$. Note that the node-to-factor message $u_{i \to a}^\ell$ is \emph{not} influenced by Gaussians at nodes that pass through $a$ to reach $i$.


Next, we prove that messages at node $i$ and iteration number $\ell+1$ are uncorrelated with the initial Gaussians in the radius-$\ell$ neighborhood around $i$. Put another way, the ``extra randomness'' at distance $\ell+1$ removes any correlation with $\mathcal{G}_i^\ell$. To do this, we show that the messages $\{u_{j \to b}^\ell \}_{j \in \partial b \backslash i}$ are uncorrelated with each other, conditioned on $\mathcal{G}_i^\ell$.\footnote{This is subtle. Although $\{u_{j \to b}^\ell \}_{j \in \partial b \backslash i}$ are independent by \Cref{lemma:SigmaAlgebraIndependence} around factor $b$, independence does not imply \emph{conditional} independence.}
\begin{lemma}
    \label{lemma:messages_conditional_independence}
    Choose any variable $i$, factor $b \in \partial i$, and $1 \le \ell \le L$. Choose also $s_j \geq 0$ for every $j \in \partial b \backslash i$. Then 
    \begin{align*}
        \E\left[\prod_{j \in \partial b \backslash i} (u_{j \to b}^\ell)^{s_j} \; \Bigg\vert \; \mathcal{G}_i^\ell \right] = \prod_{j \in \partial b \backslash i} \E \Big[ (u_{j \to b}^\ell)^{s_j} \; \Big\vert \; \mathcal{G}_i^\ell \Big] \,.
    \end{align*}
\end{lemma}
\begin{proof}
    For independent sets of random variables $X, Y$ and any function $h$ such that $\E[|h(X, Y)|]$ exists, it is known that $\E[h(X, Y) | X] = g(X)$, where $g(x) \defeq \E[h(x, Y)]$ (e.g., \cite{durrett}). We may think of $u^\ell_{j \to b}$ as a function $h_j$ of the independent random variables $\{z^0_k \mid k \in \B_{j \to b}(\ell)\}$; conditioning on $\mathcal{G}_i^\ell$ only fixes $X_j \defeq \{z^0_k \mid k \in \B_{j \to b}(\ell) \cap \B_{i}(\ell)\}$. So $u^\ell_{j \to b}$ is some function $g_j$ of the remaining (unfixed) random variables $Y_j$.
    Since $\B_{j \to b}(\ell) \backslash \B_{i}(\ell)$ is disjoint for different $j \in \partial b  \setminus  i$, the sets $\{Y_j\}_{j \in \partial b  \setminus i }$ remain mutually independent even after conditioning on $\mathcal{G}_i^\ell$. Thus, the quantities $\{g_j(X_j)\}_{j \in \partial b  \setminus  i}$ are mutually independent. The lemma directly follows.
\end{proof}

\begin{lemma}
\label{lemma:uy_is_zero}
    Choose any variable $i$, factor $a \in \partial i$, and $0 \le \ell \le L - 1$. Then $\E \big[ u_{i \to a}^{\ell + 1} \; \big\vert \; \mathcal{G}_i^\ell \big] = 0$ and $\E \big[ u_i^{\ell+1} \; \big\vert \; \mathcal{G}_i^\ell \big] = 0$.
\end{lemma}
\begin{proof}
    We prove $\E[u_{i \to a}^{\ell + 1} \mid \mathcal{G}_i^\ell] = 0$; the proof of $\E[u_i^{\ell+1} \mid \mathcal{G}_i^\ell] = 0$ is essentially the same. We proceed by induction on $\ell$. For the base case $\ell = 0$, we have $u^{1}_{i \to a} = w^{1}_{i \to a} = \frac{1}{\sqrt{d - 1}} \sum_{b \in \partial i \backslash a}  \D_{i;b} f(\vz_{\partial b \to b}^0)$. This quantity is independent of $\mathcal{G}_i^0 = \sigma(\{z_i^0\})$ by \Cref{remark:f_derivative_doesnt_use_i}.
    Since $f$ is multilinear, and each $z_{j \to b}^0$ is independent, $\E[\D_{i;b} f(\vz_{\partial b \to b}^0)] = \D_{i;b} f(\E[ \vz_{\partial b \to b}^0]) = \D_{i;b} f(0, \dots, 0)$, which is zero since there are no linear terms in the CSP function $f$. So $\E[ u^{1}_{i \to a} \mid \mathcal{G}_i^0] = \E[ u^{1}_{i \to a}] = 0$.
    
    For the inductive step, assume the statement is true up to $\ell$. We have
    \begin{align*}
    u^{\ell + 1}_{i \to a} = w^{\ell + 1}_{i \to a} - w^{\ell}_{i \to a} = \frac{1}{\sqrt{d - 1}} \sum_{b \in \partial i \backslash a} \left(
  \D_{i;b} f(\vz_{\partial b \to b}^\ell)
    -  \D_{i;b} f(\vz_{\partial b \to b}^{\ell-1})\right)\,.
    \end{align*}
    Notice that for any $b \in \partial i \setminus a$,
 \begin{align*}
     \D_{i;b} f(\vz_{\partial b \to b}^\ell)
    -  \D_{i;b} f(\vz_{\partial b \to b}^{\ell-1})
     &= \sum_{S \subseteq \partial b, i \in S}
     \hat{f}(S) \cdot \left(\prod_{j \in S \backslash i} z^{\ell}_{j \to b} - \prod_{j \in S \backslash i} z^{\ell - 1}_{j \to b}\right) \\
      &= \sum_{S \subseteq \partial b, i \in S}
      \hat{f}(S) \cdot \left(\prod_{j \in S \backslash i} \left(z^{\ell - 1}_{j \to b} + A^{\ell - 1}_{j \to b} \cdot u^{\ell}_{j \to b}\right) - \prod_{j \in S \backslash i} z^{\ell - 1}_{j \to b}\right) \\
      &= \sum_{S \subseteq \partial b, i \in S}
      \hat{f}(S)  \cdot 
      \sum_{T \subsetneq S\backslash i}
      \Big(\prod_{j \in T} z^{\ell - 1}_{j \to b}\Big)
      \Big( \prod_{j \in (S\backslash i) \backslash T}  A^{\ell - 1}_{j \to b} \cdot u^{\ell}_{j \to b}\Big) \,,
    \end{align*}
where $\hat{f}$ is the Fourier transform of $f$.\footnote{We slightly abuse notation; we write $\hat{f}(S)$ for the Fourier component of the set of \emph{coordinates} $\iota \in [r]$ such that the $\iota^{\text{th}}$ element of $\partial b$ is in $S$.} 
For any $j \in \partial b  \setminus  i$, $z_{j \to b}^{\ell - 1} \in \mathcal{G}_{j \to b}^{\ell - 1} \subseteq \mathcal{G}_i^{\ell}$ by \Cref{lemma:SigmaAlgebraDecomposition}. So we have
\begin{align*}
    \E\left[ \D_{i;b} f(\vz_{\partial b \to b}^\ell)
    -  \D_{i;b} f(\vz_{\partial b \to b}^{\ell-1}) \mid \mathcal{G}_i^\ell \right]
    =
    \sum_{S \subseteq \partial b, i \in S}
    \hat{f}(S)
    \cdot 
      \sum_{T \subsetneq S\backslash i}
      \Big(\prod_{j \in T} z^{\ell - 1}_{j \to b}\Big)
      \E\left[ \prod_{j \in (S\backslash i) \backslash T}  A^{\ell - 1}_{j \to b} 
      u^{\ell}_{j \to b} \mid \mathcal{G}_i^\ell \right]\,.
\end{align*}
Note that  $A^{\ell - 1}_{j \to b} \in \mathcal{G}_{j \to b}^{\ell - 1} \subseteq \mathcal{G}_i^\ell$,
so by \Cref{lemma:messages_conditional_independence}, $\E\left[ \prod_{j \in (S\backslash i) \backslash T}  A^{\ell - 1}_{j \to b} 
  u^{\ell}_{j \to b} \mid \mathcal{G}_i^\ell \right] = \prod_{j \in (S\backslash i) \backslash T}  A^{\ell - 1}_{j \to b} \E\left[ 
  u^{\ell}_{j \to b} \mid \mathcal{G}_i^\ell \right]$. By \Cref{lemma:SigmaAlgebraDecomposition}, 
  $\E\left[ 
  u^{\ell}_{j \to b} \mid \mathcal{G}_i^\ell \right] = \E\left[ 
  u^{\ell}_{j \to b} \mid \mathcal{G}_j^{\ell-1} \right]$,
  which is zero by the inductive hypothesis. Since every choice of $S$ and $T$ contains a $j \in (S \setminus i) \setminus T$, the above expression is zero, and so $\E[u^{\ell + 1}_{i \to a} \mid \mathcal{G}_i^\ell ] = 0$.
\end{proof}

\begin{corollary}\label{cor:monomial_is_zero}
Fix any variable $i$, factor $b \in \partial i$, and $0 \le \ell \le L$. Then $\E\left[
\D_{i;b} f(\vz_{\partial b \to b}^\ell)
 \mid \mathcal{G}_i^\ell \right] =  \D_{i;b} f(\vz_{\partial b \to b}^{\ell-1})$ if $\ell \ge 1$ and $0$ otherwise.
\end{corollary}

\subsection{Simplifying the main sum}
What is the average value achieved by the algorithm? By linearity, we can look at the value of any factor $a$; i.e. $\E[ f(\vz_{\partial a}^{L})]$. Note that $\E[ f(\vz_{\partial a}^{0})] = f(0, \dots, 0) = \E[f]$, and
\begin{align*}
\E[ f(\vz_{\partial a}^{L})]
=
\E[ f(\vz_{\partial a}^{0})] + \sum_{\ell = 0}^{L-1} 
\left(\E[ f(\vz_{\partial a}^{\ell + 1})] - \E[ f(\vz_{\partial a}^{\ell})]\right)
=\E[f] + \sum_{\ell = 0}^{L-1} \left( \E[ f(\vz_{\partial a}^{\ell + 1})] - \E[ f(\vz_{\partial a}^{\ell})]\right)
\,.
\end{align*}
It suffices to compute the difference $\E[ f(\vz_{\partial a}^{\ell + 1})] - \E[ f(\vz_{\partial a}^{\ell})]$ for every $0 \leq \ell \leq L - 1$. We use the Fourier expansion of $f$ to express this as a sum of differences of monomials. 

First, we simplify terms of the form $u_i^{s_i} Y_i^{s_i-1} \prod_{v\in \partial a \setminus i} Y_v^{s_v}$, which we encounter along the way.
This relies on the following consequence of the monotone class theorem:
\begin{lemma}[{e.g. \cite[Theorem 5.2.2]{durrett}, \cite[Lemma 2.3]{ams21}}]
\label{lemma:monotoneclassthm}
Let $\{\mathcal{F}_1,\dots,\mathcal{F}_m,\mathcal{A}\}$ be a set of $\sigma$-algebras so that $\mathcal{F} \defeq \sigma(\mathcal{F}_1,\dots,\mathcal{F}_m)  \subseteq \mathcal{A}$. Choose a random variable $X \in \mathcal{A}$,\footnote{Precisely, $X \in \mathcal{A}$ means that the outcomes of $X$ are determined by a measure $\mu: \mathcal{A} \to \R$; i.e $\sigma(X) \subseteq \mathcal{A}$.} and suppose for any $Y_1 \in \mathcal{F}_1,\dots,Y_m \in \mathcal{F}_m$, $\E[X Y_1 \dots Y_m] = 0$ when the expectation exists. Then $\E[X Y] = 0$ (when the expectation exists) for every $Y \in \mathcal{F}$.
\end{lemma}
As a result, to show an expectation is zero, we may replace a variable in $\sigma$-algebra $\mathcal{F}$ with a product of variables from $\sigma$-algebras that cover $\mathcal{F}$.

\begin{lemma}[Generalization of {\cite[Lemma 2.5]{ams21}}]
\label{lemma:uyyyyy_is_simple}
Choose factor $a$, integers $s_v \ge 0$ for $v \in \partial a$, and
let $Y_v^{s_v} \in \mathcal{G}_v^{s_v}$. Let $i \in \partial a$ be such that $s_i + 1 \ge s_v$ for all $v \in \partial a$. Then
    \begin{align*}
        \E \left[ u_i^{s_i + 1} \prod_{v\in \partial a} Y_v^{s_v} \right] 
        =
        \frac{1}{\sqrt{d}}  \E \left[ 
         \left( 
         \D_{i;a} f(\vz_{\partial a \to a}^{s_i})
    -  \mathbbm{1}_{[s_i \ge 1]}\D_{i;a} f(\vz_{\partial a \to a}^{s_i - 1})
         \right)
         \prod_{v \in \partial a} Y_v^{s_v} 
         \right]
         \,,
    \end{align*}
    where $\mathbbm{1}_{[x \ge y]} = 1$ if $x \ge y$ and $0$ otherwise. Intuitively, only factor $a$ contributes to the correlation.
\end{lemma}
\begin{proof}
When $s_i = 0$, notice that
 \begin{align*}
        \E \left[ u_i^{s_i + 1} \prod_{v\in \partial a} Y_v^{s_v} \right] 
        = \E \left[ 
         w_i^1
         \prod_{v \in \partial a} Y_v^{s_v} 
         \right]
         =
          \frac{1}{\sqrt{d}} \sum_{b \in \partial i} \E \left[ 
         \D_{i;b} f(\vz_{\partial b \to b}^{0})
         \prod_{v \in \partial a} Y_v^{s_v} 
         \right]
         \,.
\end{align*}
Consider any $b \ne a$. Then $\D_{i;b} f(\vz_{\partial b \to b}^{0})$ is independent of all other variables by \Cref{remark:f_derivative_doesnt_use_i}. Since $f$ has no linear terms, $\E[ \D_{i;b} f(\vz_{\partial b \to b}^{0})] = \D_{i;b} f(0, \dots, 0) = 0$ and so this term is zero. The statement follows.

Now we analyze when $s_i \ge 1$. By definition of $u_i^{s_i + 1}$, we have
    \begin{align*}
        \E \left[  u_i^{s_i + 1} \prod_{v\in \partial a} Y_v^{s_v} \right] &=
 \frac{1}{\sqrt{d}} \sum_{b \in \partial i} \E \left[   
        \left( 
         \D_{i;b} f(\vz_{\partial b \to b}^{s_i})
    -  \D_{i;b} f(\vz_{\partial b \to b}^{s_i - 1})
         \right)
         \prod_{v \in \partial a} Y_v^{s_v} 
         \right] \,.
    \end{align*}
    To prove the lemma, we show that the expectation is zero for $b \ne a$.
    By \Cref{lemma:SigmaAlgebraDecomposition} and \Cref{lemma:monotoneclassthm}, we can assume 
    $Y_v^{s_v} = Y_{v \to a; v}^{s_v}\prod_{v' \in \partial a \backslash v} Y^{s_v - 1}_{v' \to a; v}$ for every $v \in \partial a$, where $Y_{v' \to a; v}^{s_v-1} \in \mathcal{G}_{v' \to a}^{s_v-1}$ for all $s_v$ and variables $v'$.\footnote{Here we use a second subscript to differentiate terms associated with different $Y_v^{s_v}$.} So,
    \begin{align*}
        & \E \left[  
        \left(
          \D_{i;b} f(\vz_{\partial b \to b}^{s_i})
    -  \D_{i;b} f(\vz_{\partial b \to b}^{s_i - 1})
         \right)
         \prod_{v \in \partial a} Y_v^{s_v} 
         \right]
         \\
         &= \E \left[    
         \left(
          \D_{i;b} f(\vz_{\partial b \to b}^{s_i})
    -  \D_{i;b} f(\vz_{\partial b \to b}^{s_i-1})
         \right)
         \left(
         \prod_{v \in \partial a} Y_{v \to a; v}^{s_v}\prod_{v' \in \partial a \backslash v} Y^{s_v - 1}_{v' \to a; v} 
         \right)\right] \\
         &=
         \E\left[
         \left( 
         \D_{i;b} f(\vz_{\partial b \to b}^{s_i})
    -  \D_{i;b} f(\vz_{\partial b \to b}^{s_i-1})
         \right)
         Y_{i \to a;i}^{s_i}
         \prod_{v \in \partial a  \setminus i} Y_{i \to a;v}^{s_v - 1} 
         \right]
         \cdot 
         \E\left[ \prod_{v \in \partial a \setminus i} Y_{v \to a; v}^{s_v}\prod_{v' \in \partial a \backslash v} Y^{s_v - 1}_{v' \to a; v} \right]\,,
    \end{align*}
    where the last line follows by \Cref{lemma:SigmaAlgebraIndependence} around factor $a$.
    We show that the first term is zero.
    Again using \Cref{lemma:SigmaAlgebraDecomposition} and \Cref{lemma:monotoneclassthm}, assume
    $Y_{i \to a;v}^{k_v} = Y_{i;v}^0\prod_{c \in \partial i \backslash a} \prod_{j \in \partial c \backslash i} Y_{j \to c;v}^{k_v-1}$, where $Y_{i;v}^0 \in \mathcal{G}_i^{0}$ and $Y_{j \to c;v}^{k_v - 1} \in \mathcal{G}^{k_v - 1}_{j \to c}$, for $k_i = s_i$ and $k_v = s_v - 1$ for $v \in \partial a \setminus i$.
    Then, by \Cref{remark:f_derivative_doesnt_use_i} and \Cref{lemma:SigmaAlgebraIndependence}, the first term factors into
    \begin{align*}
         &\E\left[
         \left( 
         \D_{i;b} f(\vz_{\partial b \to b}^{s_i})
    -  \D_{i;b} f(\vz_{\partial b \to b}^{s_i-1})
         \right)
         \left(
         \prod_{j \in \partial b  \setminus  i} Y_{j \to b; i}^{s_i - 1}
         \prod_{v \in \partial a \setminus i} Y_{j \to b; v}^{s_v - 2}
         \right)
         \right]
         \\
          \cdot
         &\E\left[
    \left(
    Y_{i;i}^0 
\prod_{v \in \partial a \setminus i}
         Y_{i;v}^0
         \right)
         \left( 
         \prod_{c \in \partial i \setminus \{a,b\}} 
         \prod_{j \in \partial c \setminus i} 
         Y_{j \to c;i}^{s_i - 1}
         \prod_{v \in \partial a \setminus i} 
         Y_{j \to c;v}^{s_v - 2}
         \right)
         \right]\,.
    \end{align*}
    By the assumed inequality, $\left( Y_{j \to b; i}^{s_i - 1}
         \prod_{v \in \partial a \setminus i} Y_{j \to b; v}^{s_v - 2} \right) \in \mathcal{G}_{j \to b}^{s_i - 1}$ for each $j \in \partial b \setminus i$. The result follows by \Cref{cor:monomial_is_zero}.
\end{proof}

\begin{proposition}\label{prop:product_difference}
Fix $0 \le \ell \le L - 1$, factor $a$, and $S \subseteq \partial a$. Then
    \begin{align*}
    \E\left[\prod_{j \in S} z_j^{\ell + 1}\right] - \E\left[\prod_{j \in S} z_j^{\ell}\right] =  \sum_{i \in S} \E\left[A_i^\ell u_i^{\ell + 1} \cdot \prod_{j \in S \setminus i} z_j^{\ell}\right] + O_d\left(\frac{1}{d}\right).
    \end{align*}
\end{proposition}
\begin{proof}
We have
\begin{align*}
    \E\left[\prod_{j \in S} z_j^{\ell + 1}\right] - \E\left[\prod_{j \in S} z_j^{\ell}\right] & = \E\left[\prod_{j \in S} \left(z_j^{\ell} + A_j^{\ell}u_j^{\ell + 1}\right) - \prod_{j \in S} z_j^{\ell}\right] = \sum_{T \subsetneq S}
    \E\left[
      \Big(\prod_{j \in T} z^{\ell}_{j}\Big)
      \Big( \prod_{j \in S \backslash T}  A^{\ell}_{j} u^{\ell+1}_{j}\Big)
      \right]\,.
\end{align*}
We expand each $u^{\ell + 1}_j$ using \Cref{lemma:uyyyyy_is_simple}. For any $T \subsetneq S$, notice that 
\begin{align*}
    \E\left[
      \Big(\prod_{j \in T} z^{\ell}_{j}\Big)
      \Big( \prod_{j \in S \backslash T}  A^{\ell}_{j} u^{\ell+1}_{j}\Big)
      \right]
      = \left(\frac{1}{\sqrt{d}} \right)^{|S\backslash T|} \E\left[ \Big(\prod_{j \in T} z^{\ell}_{j}\Big)  
       \Big( \prod_{j \in S \backslash T}  A^{\ell}_{j} \cdot 
       \big( \D_{j;a} f(\vz_{\partial a \to a}^{\ell})
    -  \mathbbm{1}_{[\ell \ge 1]} \D_{j;a} f(\vz_{\partial a \to a}^{\ell-1})
    \big) \Big)\right]\,.
\end{align*}
Recall that $A_j^\ell$ is bounded (see~\Cref{fig:algorithm}), and by \Cref{lemma:finite_moments_are_bounded}, all finite moments of $z_{j \to a}^\ell$ and $z_j$ are also bounded. So this term is $O_d(\frac{1}{d^{|S\backslash T|/2}})$. Since $T \subsetneq S$, all terms are $O_d(\frac{1}{\sqrt{d}})$. The proposition follows.
\end{proof}

\begin{corollary}\label{prop:clause_difference_value}
    Choose $0 \le \ell \le L - 1$ and factor $a$. Then
    \begin{align*}
    \E\left[f(\vz_{\partial a}^{\ell + 1})\right] - \E\left[f(\vz_{\partial a}^{\ell})\right]
    = \sum_{i \in \partial a} \E\left[
    A_i^{\ell}u_{i}^{\ell + 1} \cdot \D_{i;a} f(\vz_{\partial a}^{\ell}) \right] + O_d\left(\frac{1}{d}\right).
    \end{align*}
\end{corollary}
\begin{proof}
    Taking the Fourier expansion, we have
    \begin{align*}
        \E\left[f(\vz_{\partial a}^{\ell + 1})\right] - \E\left[f(\vz_{\partial a}^{\ell})\right]
        &= 
        \sum_{S \subseteq \partial a} \hat{f}(S) \cdot \left( \E\left[ \prod_{j \in S} z_j^{\ell + 1}\right] - \E\left[\prod_{j \in S} z_j^\ell\right]\right) \\
        &= \sum_{S \subseteq \partial a} \hat{f}(S) \cdot \sum_{i \in S} \left( \E\left[A_i^\ell u_i^{\ell + 1} \cdot \prod_{j \in S \setminus i} z_j^{\ell}\right] + O_d\left(\frac{1}{d}\right)\right) \\
        &= \sum_{i \in \partial a} \sum_{S \subseteq \partial a, i \in S} \hat{f}(S) \cdot  \E\left[A_i^\ell u_i^{\ell + 1} \cdot \prod_{j \in S \setminus i} z_j^{\ell}\right] + O_d\left(\frac{1}{d}\right) \\
        &= \sum_{i \in \partial a} \E\left[
       A_i^{\ell}u_{i}^{\ell + 1}
       \cdot
       \D_{i;a} f(\vz_{\partial a}^{\ell})\right] + O_d\left(\frac{1}{d}\right)\,.
    \end{align*}
    Here in the second equality we use~\Cref{prop:product_difference}, and the fact that $\partial a$ has a bounded number of elements.
\end{proof}

We simplify this expression further, assuming the algorithm's node and node-to-factor nonlinearities are similar (\Cref{hyp:A_deviation_small}). We choose nonlinearities exhibiting this property in~\Cref{sec:nonlinearities}.

\begin{proposition}
    Assume \Cref{hyp:A_deviation_small}. Let $0 \le \ell \le L - 1$, $a$ be a factor, and $i \in \partial a$. Then
     \begin{align*}
    \frac{1}{\alpha n} \sum_{a \in E} \left( \E\left[f(\vz_{\partial a}^{\ell+1})\right] - \E\left[f(\vz_{\partial a}^{\ell})\right] \right) = \frac{r}{\sqrt{d}} \cdot \frac{1}{n}\sum_{i \in V} \E\left[A_i^{\ell}(u_{i}^{\ell + 1})^2\right] + O_d\left(\frac{1}{d}\right)\,.
    \end{align*}
\end{proposition}
\begin{proof}
    By \Cref{prop:clause_difference_value} we have
    \begin{align*}
    \frac{1}{\alpha n} \sum_{a \in E} \left( \E\left[f(\vz_{\partial a}^{\ell+1})\right] - \E\left[f(\vz_{\partial a}^{\ell})\right] \right)
    = \frac{1}{\alpha n} \sum_{a \in E} \sum_{v \in \partial a} \E\left[
 A_v^{\ell}u_{v}^{\ell + 1}
       \cdot
       \D_{v;a} f(\vz_{\partial a}^{\ell})
    \right] + O_d\left(\frac{1}{d}\right)\,.
    \end{align*}
Rewriting the double sum on the right-hand side, this quantity is equal to
\begin{align*}
    \frac{1}{\alpha n} \sum_{i \in V} \sum_{b \in \partial i}
     \E\left[
            A_i^{\ell}u_{i}^{\ell + 1}
       \cdot
       \D_{i;b} f(\vz_{\partial b}^{\ell})
        \right] + O_d\left(\frac{1}{d}\right)
        =
    \frac{r}{d} \cdot \frac{1}{n} \sum_{i \in V}
     \E\left[
            A_i^{\ell}u_{i}^{\ell + 1}
       \cdot
       \sum_{b \in \partial i} \D_{i;b} f(\vz_{\partial b}^{\ell})
        \right] + O_d\left(\frac{1}{d}\right)\,.
\end{align*}
Notice that the expectation on the right-hand side is very similar to
\begin{align*}
    \frac{1}{d} \sum_{b \in \partial i} \E[   A_i^{\ell} u_i^{\ell + 1} \cdot  \D_{i;b} f(\vz_{\partial b \to b}^{\ell})]
    =
    \frac{1}{\sqrt{d}} \E [A_i^{\ell} u_i^{\ell + 1} \cdot  w_i^{\ell+1}]\,.
\end{align*}
By \Cref{lemma:uy_is_zero}, $\E[A_i^\ell u_i^{\ell + 1} \cdot w_i^\ell] = \E[ A_i^\ell w_i^\ell \cdot  \E[ u_i^{\ell + 1} \mid \mathcal{G}_i^\ell ]] = 0$, and so 
    \begin{align*}
     \frac{1}{\sqrt{d}} \E [A_i^{\ell} u_i^{\ell + 1} \cdot  w_i^{\ell+1}]
     =
      \frac{1}{\sqrt{d}}  \E [A_i^{\ell} u_i^{\ell + 1} \cdot  (w_i^{\ell+1} - w_i^\ell)]
      =
     \frac{1}{\sqrt{d}}  \E [A_i^{\ell} (u_i^{\ell + 1})^2]\,.
    \end{align*}
The proof follows by bounding the difference
\begin{align*}
    \left|  
    \frac{1}{d} \sum_{b \in \partial i}
    \E[   A_i^{\ell} u_i^{\ell + 1} \cdot  \left( \D_{i;b} f(\vz_{\partial b}^{\ell})
   -
    \D_{i;b} f(\vz_{\partial b \to b}^{\ell})
    \right)
   ] 
   \right|
   \le \max_{b \in \partial i}
   \left|
   \E[   A_i^{\ell} u_i^{\ell + 1} \cdot  \left( \D_{i;b} f(\vz_{\partial b}^{\ell})
   -
    \D_{i;b} f(\vz_{\partial b \to b}^{\ell})
    \right)
   ]
      \right|\,.
      \end{align*}
      When $\ell = 0$, $\vz_{\partial b}^\ell = \vz_{\partial b \to b}^\ell$ and so this difference is zero.

      Now suppose $\ell > 0$. Let $c$ be the factor which maximizes the right-hand side. 
      By~\Cref{lemma:uyyyyy_is_simple}, this term is at most
\begin{align*}
    &\frac{1}{\sqrt{d}} \E[ A_i^\ell
    \left(\D_{i;c} f(\vz_{\partial c \to c}^\ell)  -
    \D_{i;c} f(\vz_{\partial c \to c}^{\ell-1})
    \right)
    \cdot  \left( \D_{i;c} f(\vz_{\partial c}^{\ell})
   -
    \D_{i;c} f(\vz_{\partial c \to c}^{\ell})
    \right)
    ]
    \\
    &\le \frac{1}{\sqrt{d}} \sqrt{\E[ (A_i^\ell)^2
    \left(\D_{i;c} f(\vz_{\partial c \to c}^\ell)  -
    \D_{i;c} f(\vz_{\partial c \to c}^{\ell-1})
    \right)^2]}
    \cdot  
    \sqrt{
    \E[\left( \D_{i;c} f(\vz_{\partial c}^{\ell})
   -
    \D_{i;c} f(\vz_{\partial c \to c}^{\ell})
    \right)^2
    ]}\,,
\end{align*}
where the inequality is by Cauchy-Schwarz.
The first expectation is a finite product of terms with bounded moments  by \Cref{lemma:finite_moments_are_bounded}, so it is bounded. By \Cref{lemma:df_deviation_small} (which uses \Cref{hyp:A_deviation_small}), the second expectation is $O_d(\frac{1}{d})$, so the whole expression is $O_d(\frac{1}{d})$.
\end{proof}
By these propositions, we get the following: 
\begin{theorem}\label{thm:simplified_sum}
Assume \Cref{hyp:A_deviation_small}. Fix a constant $L \ge 1$. Then 
\begin{align*}
    \frac{1}{\alpha n} \sum_{a \in E} \E\left[f(\vz_{\partial a}^{L})\right] = \E[f] + \frac{r}{\sqrt{d}} \cdot \sum_{\ell = 0}^{L-1} \left( \frac{1}{n} \sum_{i \in V} \E\left[A_i^{\ell}(u_{i}^{\ell + 1})^2\right]\right) + O_d\left(\frac{1}{d}\right)\,.
\end{align*}
\end{theorem}

\section{State evolution for message-passing on hypergraphs}
\label{sec:stateevo}

So far, we've shown that the value of the algorithm takes a simple form (i.e. \Cref{thm:simplified_sum}).
From here, we prove a \emph{state evolution} statement: with some weak conditions on the nonlinearities (\Cref{hyp:second_moment}), we analyze the asymptotic behavior of statistics of the algorithm. 
State evolution as applied to AMP first appeared in~\cite{bayati11};
this was inspired by a proof in a paper computing the fixed points of the TAP equations for the Sherrington--Kirkpatrick model~\cite{bolthausen2014iterative}.
Since then, state evolution has been used in AMP algorithms applied to a variety of tasks; see for example \cite{feng2021unifying}.

Here, we consider functions of the messages $\{u_{i \to a}^\ell\}_{\ell \in [L]}$ (or $\{u_{i}^\ell\}_{\ell \in [L]}$) that are \emph{pseudo-Lipschitz}:

\begin{definition}
\label{defn:pseudolipschitz}
Fix $s > 0$. We say that a given function $\psi: \R^s \rightarrow \R$ is \emph{pseudo-Lipschitz} if there exists a constant $C > 0$ such that
\begin{equation*}
\lvert \psi(\vx) - \psi(\vy) \rvert 
\leq C \cdot \big( 1 + \lVert \vx \rVert + \lVert \vy \rVert \big) \cdot \lVert \vx - \vy \rVert \, ,
\end{equation*}
for every $\vx, \vy \in \R^s$.
\end{definition}
As $d \rightarrow \infty$, the expected value of a pseudo-Lipschitz function $\psi$ 
with messages $\{u_{i \to a}^\ell\}_{\ell \in [L]}$ (or $\{u_{i}^\ell\}_{\ell \in [L]}$) as input
approaches the expected value of $\psi$ with independent Gaussians as input.

\begin{restatable}{definition}{gaussianVariance}
\label{defn:nu_defn}
For any predicate $f: \{\pm 1\}^r \to \{0,1\}$, let $\xi$ be defined as in \Cref{thm:main_theorem}. For any $\ell \geq 1$, define the quantity\footnote{Note this is the same parameter $\nu_{\ell}$ in \Cref{hyp:second_moment}.}
\begin{equation*}
\nu_{\ell} \defeq \frac{\xi'(\ell \delta) - \xi'((\ell - 1) \delta)}{r} \, .
\end{equation*}
\end{restatable}
\begin{restatable}{theorem}{stateEvolutionSimple}
\label{prop:state_evolution_expectation}
    Assume \Cref{hyp:second_moment}. Then for any pseudo-Lipschitz function $\psi: \mathbb{R}^\ell \to \mathbb{R}$, variable $i$, factor $a \in \partial i$, $1 \le \ell \le L$, we have
    \begin{align}
         \E\left[\psi(u_{i \to a}^1, \ldots, u_{i \to a}^\ell)\right] &= \E\left[\psi(U_1, \ldots, U_\ell)\right] + o_d(1)\,,
        \label{eqn:state_evolution_simple.factor}
        \\
   \E\left[\psi(u_{i}^1, \ldots, u_{i}^\ell)\right] &= \E\left[\psi(U_1, \ldots, U_\ell)\right] + o_d(1)\,,
            \label{eqn:state_evolution_simple.node}
    \end{align}
    where $U_s \sim \cN(0, \nu_s)$ independently.
\end{restatable}
\Cref{prop:state_evolution_expectation} is inspired by the state evolution statement in~\cite{ams21}. These statements are for \emph{sparse} problems, and so 
do not require a correction term to ensure variables remain Gaussian (often named the \emph{Onsager correction}).
In \cite{ams21}, the messages   $\{u_{i \to a}^\ell\}_{\ell \in [L]}$ are Gaussian even at fixed degree!
In our algorithm, this is not the case for arbitrary predicates $f$; we must take additional care showing that these quantities are close to Gaussian.

We defer the proof of \Cref{prop:state_evolution_expectation} to the appendices.
In \Cref{sec:appendix_moments}, we establish that given \Cref{hyp:second_moment}, $u_{i \rightarrow a}^{\ell}$ is \emph{close} to the Gaussian $\cN(0, \nu_{\ell})$ when $d \rightarrow \infty$, independently of $\mathcal{G}_{i}^{\ell - 1}$.
In \Cref{sec:appendix_stateevo}, we show that pseudo-Lipschitz functions are smooth enough so that an input $u_{i \rightarrow a}^{\ell}$ may be exchanged with an independent Gaussian input, and use this argument inductively to finish the proof.

Given \Cref{prop:state_evolution_expectation}, we may convert each spin $\{z_i^{\ell}\}_{0 \le \ell \le L}$ to a certain \emph{Gaussian process}. In \Cref{sec:nonlinearities}, we do exactly this, and choose the nonlinearities $\{A_i^\ell\}_{0 \le \ell \le L}$ to maximize the expression in \Cref{thm:simplified_sum} (subject to $|z_i^L| \le 1$).

\section{Choosing the nonlinearities}
\label{sec:nonlinearities}
We now design the nonlinearities to maximize \Cref{thm:simplified_sum}.
In light of \Cref{sec:stateevo}, we consider the stochastic optimal control problem of \cite{ams20}.
Define $\xi(s) \defeq \sum_{j=1}^r  \| f^{=j}\|^2 s^j$. Consider any $\mu \in \mathscr{L}_\xi$, and $\Phi^\mu$ chosen as in~\Cref{eq.parisi-equations}. Define the stochastic differential equation\footnote{Here, $(B_t)_{t \ge 0}$ is standard Brownian motion.}
\begin{align}
    \label{eqn:sde}
\mathrm{d}X_t = \xi''(t)\mu(t)\Phi_x^\mu(t, X_t)\mathrm{d}t + \sqrt{\xi''(t)}\mathrm{d}B_t\,,
\quad X_0 = 0\,,
\end{align}
along with associated martingale
\begin{align}\label{eqn:zmartingale}
    Z_t = \int_0^t \sqrt{\xi''(s)}\Phi_{xx}^\mu(s, X_s)\mathrm{d}B_s\,.
\end{align}
We now crucially use \Cref{assn:alg_minimizer_exists}. Let $\mu_\ast \in \mathscr{L}$ be the minimizer of \Cref{eqn:alg_defn}.
\begin{lemma}[{\cite[Theorem 3]{ams20}}]\label{lemma:goal}
     Suppose \Cref{assn:alg_minimizer_exists} holds. Then for every $\epsilon > 0$, there exists $\eta > 0$ such that
    \begin{align*}
    \int_0^{1 - \eta} \xi''(t)\E[\Phi_{xx}^{\mu_\ast}(t, X_t)] \mathrm{d}t\geq \ALG_{\xi} - \epsilon\,.
    \end{align*}
\end{lemma}
For mean-field spin glasses, Parisi predicted that the minimizer of $\inf_{\mu \in \mathscr{U}} \P_\xi(\mu)$ led to the correct description of the ground state. While this was true, uniqueness (\Cref{thm:auffingerchen}) required a stochastic representation of the minimization problem. 
Later, \cite{montanari2019optimization,ams20} formulated an algorithm whose optimal value can be written as a stochastic optimal control problem dual to \Cref{eqn:alg_defn}.
Surprisingly, the study of overlap gap phenomena is also linked to \Cref{eqn:alg_defn}~\cite{huang2022tight,jmss22}, so this algorithm is optimal among a broad family.\footnote{This is the family of \emph{overlap-concentrated} algorithms; see \Cref{defn:overlap-concentrated}.}

In a sense, the optimal message-passing algorithm updates its local information in a way concordant with the Parisi minimizer $\mu_\ast$, \emph{assuming} that $\mu_\ast$ exists. Nonetheless, even if \Cref{assn:alg_minimizer_exists} is false, we informally argue that this algorithm succeeds. Since $\Phi^{\mu}_x$ and $\Phi^{\mu}_{xx}$ are continuous in $\mu$ (e.g., \cite[Theorem 4]{chen2017variational}),
using a near-minimizer $\widetilde{\mu}_\ast$ of \Cref{eqn:alg_defn} changes the value in \Cref{lemma:goal} by a negligible amount, which can be absorbed into the choice of $\epsilon$. Furthermore, $\widetilde{\mu}_\ast$ can be precomputed, since finding $\widetilde{\mu}_\ast$ is a task required only once per \emph{CSP predicate $f$}, not once per run of the algorithm. Mathematizing this argument, however, is a valuable open problem; see \cite[Section 6]{ams20}.

In this section, we prove that the algorithm achieves a value affine to \Cref{lemma:goal} after taking two sets of limits.

\subsection{A finite difference equation}
Consider discrete versions of~\Cref{eqn:sde,eqn:zmartingale}. Choose $\delta > 0$, and consider independent \emph{standard} Gaussians $\{B_1, B_2, \dots\}$. Then consider the finite difference
\begin{align}
     X_{\ell+1}^\delta - X_{\ell}^\delta & = \xi''(\delta \ell)\mu_\ast(\delta \ell)\Phi_x^{\mu_\ast}(\delta \ell, X_\ell^\delta) \cdot \delta + \sqrt{\xi'(\delta (\ell + 1)) - \xi'(\delta \ell)} \cdot B_{\ell+1} &\quad X_0^\delta &= 0\,,\label{eqn:finitedifference1}\\
     Z_{\ell+1}^\delta - Z_{\ell}^\delta & = \sqrt{\xi'(\delta (\ell + 1)) - \xi'(\delta \ell)}
     \cdot\Phi_{xx}^{\mu_\ast}(\delta \ell, X_\ell^\delta)\cdot B_{\ell + 1} &\quad Z_0^\delta &= 0\,.\label{eqn:zmartingale_discrete}
\end{align}
The variables are named suggestively. By \Cref{sec:stateevo}, the messages $u_{i \to a}^\ell, u_i^\ell$ behave as independent Gaussians as $d \to \infty$, which we use to imitate $B_\ell$ above. To this end, we define the sequences
\begin{align*}
    x_{i}^{\ell + 1} &= x_{i}^{\ell} + \xi''(\delta \ell)\mu_\ast(\delta \ell)\Phi^{\mu_\ast}_x(\delta \ell, x_{i}^{\ell})  \cdot \delta + \sqrt{r} \cdot u_i^{\ell+1}  & x^0_i &= 0\,,\\
    x_{i \to a}^{\ell + 1} &= x_{i \to a}^{\ell}+ \xi''(\delta \ell)\mu_\ast(\delta \ell)\Phi^{\mu_\ast}_x(\delta \ell, x_{i \to a}^{\ell}) \cdot \delta + \sqrt{r} \cdot u_{i \to a}^{\ell+1} &  x^{0}_{i \to a} &= 0\,.
\end{align*}

We show that in the limit $d \to \infty$, the messages $\{u_i^\ell, u_{i \to a}^\ell\}$ behave as stochastic input to \Cref{eqn:finitedifference1,eqn:zmartingale_discrete}.
Note that $\Phi^\mu_x(t,x)$ and $\Phi^\mu_{xx}(t,x)$ are Lipschitz in $x$, since all  $x$-derivatives are of $\Phi^\mu_x(t,x)$ are bounded:
\begin{lemma}[{\cite[Lemma 6.4]{ams20}}]\label{lem:boundedpartials}
    For every  $\mu \in \mathscr{L}$, $\epsilon > 0$, and $j \geq 2$, there exists $K = K(\mu, j, \epsilon) > 0$ such that the $j^\text{th}$ derivative in $x$ of $\Phi^{\mu}$ is within $[-K, K]$ for all $t \in [0, 1-\epsilon]$ and $x \in \mathbb{R}$. In other words, $|\frac{\partial^j}{\partial x^j}\Phi^{\mu}(t,x)| \leq K$ for all $t \in [0, 1-\epsilon]$ and $x \in \mathbb{R}$.
\end{lemma}
By induction and using \Cref{lem:boundedpartials}, $x_i^\ell$ and $\Phi^{\mu_\ast}_{xx}(t,x_i^\ell)$, and $x_{i \to a}^\ell$ and  $\Phi^{\mu_\ast}_{xx}(t,x_{i \to a}^\ell)$ are Lipschitz in the messages $\{u_i^s\}_{s \le \ell}$ and $\{u_{i \to a}^s\}_{s \le \ell}$, respectively. So we can apply \Cref{prop:state_evolution_expectation}:
\begin{proposition}
\label{cor:large_degree_limit_finite_difference}
Assume \Cref{hyp:second_moment}. Choose $1 \le \ell \le L$ and any variable $i$.
Then for every integer $k \ge 1$, we have 
\begin{align*}
    \E[\Phi_{xx}^{\mu_\ast}(\delta \ell, x_i^{\ell})^k] =\E[\Phi_{xx}^{\mu_\ast}(\delta \ell, X_\ell^\delta)^k] + o_d(1)\,.
\end{align*}
\end{proposition}
We now choose the nonlinearities:
\begin{align}
\label{eqn:nonlinearities1}
A_{i}^{\ell} = \Phi_{xx}^{\mu_\ast}(\delta \ell, x_i^{\ell}) &\cdot    \sqrt{\frac{\delta }{\nu_{\ell + 1} \cdot \E[ \Phi_{xx}^{\mu_\ast}(t, x_i^\ell)^2]}}\,,
\\
\label{eqn:nonlinearities2}
A_{i \to a}^{\ell} = \Phi_{xx}^{\mu_\ast}(\delta \ell, x_{i  \to a}^{\ell})  &\cdot \sqrt{\frac{\delta }{\nu_{\ell + 1} \cdot \E[ \Phi_{xx}^{\mu_\ast}(t, x_{i \to a}^\ell)^2]}}\,.
\end{align}

These nonlinearities behave as we promised earlier. First of all, by \Cref{lem:boundedpartials}, the nonlinearities are uniformly bounded for all $d$ (presumed in \Cref{fig:algorithm}).
Furthermore, by construction, the second moment of the nonlinearities satisfy \Cref{hyp:second_moment}. Finally, the nonlinearities do satisfy \Cref{hyp:A_deviation_small}:
\begin{proposition}
    \Cref{eqn:nonlinearities1,eqn:nonlinearities2} satisfy \Cref{hyp:A_deviation_small}. That is, for every $m \in \mathbb{N}$, we have $\E[|A_i^\ell - A_{i \to a}^{\ell}|^m] = O_d(\frac{1}{d^{m/2}})$.
\end{proposition}
\begin{proof}
    We first bound the difference of $\E[ \Phi_{xx}^{\mu_\ast}(t, x_i^\ell)^2]$ and $\E[ \Phi_{xx}^{\mu_\ast}(t, x_{i \to a}^\ell)^2]$:
    \begin{align*}
        \left|\E[ \Phi_{xx}^{\mu_\ast}(t, x_i^\ell)^2] - \E[ \Phi_{xx}^{\mu_\ast}(t, x_{i \to a}^\ell)^2]\right| 
        & = \left|\E[ \left(\Phi_{xx}^{\mu_\ast}(t, x_i^\ell) - \Phi_{xx}^{\mu_\ast}(t, x_{i \to a}^\ell)\right)\cdot \left(\Phi_{xx}^{\mu_\ast}(t, x_i^\ell) + \Phi_{xx}^{\mu_\ast}(t, x_{i \to a}^\ell)\right)]\right|\\
        & \leq \E\left[ \left|\left(\Phi_{xx}^{\mu_\ast}(t, x_i^\ell) - \Phi_{xx}^{\mu_\ast}(t, x_{i \to a}^\ell)\right)\cdot \left(\Phi_{xx}^{\mu_\ast}(t, x_i^\ell) + \Phi_{xx}^{\mu_\ast}(t, x_{i \to a}^\ell)\right)\right|\right]\\
        & \leq 2K \cdot\E\left[ \left|\left(\Phi_{xx}^{\mu_\ast}(t, x_i^\ell) - \Phi_{xx}^{\mu_\ast}(t, x_{i \to a}^\ell)\right)\right|\right]\,,
    \end{align*}
    where the last inequality follows by boundedness of $\Phi_{xx}^{\mu_\ast}$ (\Cref{lem:boundedpartials}). By Lipschitzness of $\Phi_{xx}^{\mu_\ast}$ and \Cref{lemma:u_deviation_small}, this term is $O_d(\frac{1}{\sqrt{d}})$. It follows that $\left( \E[ \Phi_{xx}^{\mu_\ast}(t, x_i^\ell)^2] \right)^{-1/2} - 
        \left( \E[ \Phi_{xx}^{\mu_\ast}(t, x_{i \to a}^\ell)^2] \right)^{-1/2}$, which equals
    \begin{align*}
    \left|\frac{\E[ \Phi_{xx}^{\mu_\ast}(t, x_i^\ell)^2] - \E[ \Phi_{xx}^{\mu_\ast}(t, x_{i \to a}^\ell)^2]}{\sqrt{\E[ \Phi_{xx}^{\mu_\ast}(t, x_i^\ell)^2]\E[ \Phi_{xx}^{\mu_\ast}(t, x_{i \to a}^\ell)^2]}(\sqrt{\E[ \Phi_{xx}^{\mu_\ast}(t, x_i^\ell)^2} + \sqrt{\E[ \Phi_{xx}^{\mu_\ast}(t, x_{i \to a}^\ell)^2]})}\right| \,,
    \end{align*}
    is also $O_d(\frac{1}{\sqrt{d}})$ by boundedness of $\Phi_{xx}^{\mu_\ast}$ (\Cref{lem:boundedpartials}). So, $A_i^\ell - A_{i \to a}^{\ell}$ can be written as
    \begin{align*}
         &\sqrt{\frac{\delta}{\nu_{\ell + 1}}} \cdot \left(\frac{\Phi_{xx}^{\mu_\ast}(\delta \ell, x_i^{\ell}) }{ \sqrt{ \E[ \Phi_{xx}^{\mu_\ast}(t, x_i^\ell)^2]}} - \frac{\Phi_{xx}^{\mu_\ast}(\delta \ell, x_{i \to a}^{\ell}) }{ \sqrt{ \E[ \Phi_{xx}^{\mu_\ast}(t, x_{i \to a}^\ell)^2]}}\right) \\ 
         = \,, &\sqrt{\frac{\delta}{\nu_{\ell + 1}}} \cdot \left(\frac{\Phi_{xx}^{\mu_\ast}(\delta \ell, x_i^{\ell}) - \Phi_{xx}^{\mu_\ast}(\delta \ell, x_{i \to a}^{\ell}) }{ \sqrt{ \E[ \Phi_{xx}^{\mu_\ast}(t, x_i^\ell)^2]}} + \Phi_{xx}^{\mu_\ast}(\delta \ell, x_{i \to a}^{\ell})\Big(\sqrt{\frac{1}{\E[ \Phi_{xx}^{\mu_\ast}(t, x_i^\ell)^2]}} - \sqrt{\frac{1}{\E[ \Phi_{xx}^{\mu_\ast}(t, x_{i \to a}^\ell)^2]}}\Big)\right). \\
    \end{align*}
    Within the parentheses, each term is $\frac{1}{\sqrt{d}}$ times a quantity with bounded moments by Lipschitzness of $\Phi_{xx}^{\mu_\ast}$ and \Cref{lemma:u_deviation_small}. This completes the proof.
\end{proof}

Note that with this choice of nonlinearities, the value of the spin $\vz^\ell$ differs slightly from \Cref{eqn:zmartingale_discrete} in the large degree limit. By \Cref{fig:algorithm}, the spin value evolves as
\begin{align*}
\widetilde{Z}_{\ell+1}^\delta  - \widetilde{Z}_{\ell}^\delta
= B_{\ell + 1} \cdot \frac{\sqrt{\delta}\cdot
  \Phi_{xx}^{\mu_\ast}(\delta \ell, X_\ell^\delta)}
  {\sqrt{\E[\Phi_{xx}^{\mu_\ast}(\delta \ell, X_\ell^\delta)^2]}}\,, 
\quad \widetilde{Z}_0^\delta \sim N(0, \delta)\,.
\end{align*}
However, the difference between $Z_\ell^\delta$ and $\widetilde{Z}_\ell^\delta$ is inconsequential for small enough $\delta$:
\begin{proposition}[{\cite[Equation 5.7]{ams20}}]\label{prop:A_normalization}
    Suppose \Cref{assn:alg_minimizer_exists} holds. Then there exist $\delta_0, C > 0$ such that for every $\delta < \delta_0$ and $\ell < \frac{1}{\delta}$, we have
    \begin{align*}
    \left|\frac{\xi'((\ell + 1)\delta) - \xi'(\ell \delta)}{\delta} \cdot \E[\Phi_{xx}^{\mu_\ast}(\delta \ell, X_\ell^\delta)^2] - 1\right| \leq C\sqrt{\delta}\,.
    \end{align*}
\end{proposition}

\subsection{Achieving the extended Parisi value}
Here we show that \Cref{thm:simplified_sum} achieves a value affine to $\ALG_\xi$, the infimum of a Parisi minimization problem.
First, note that the expectation $\E[A_i^\ell (u_i^{\ell + 1})^2]$ can be simplified at large degree:
\begin{proposition}\label{prop:replace_u2}
For every variable $i$ and $0 \le \ell \le L - 1$, we have
\begin{align*}
\left| \E\left[A_i^\ell(u_i^{\ell + 1})^2\right] - \nu_{\ell + 1}\cdot  \E\left[A_i^\ell\right]\right| = O_d\left(\frac{1}{\sqrt{d}}\right)\,.
\end{align*}
\end{proposition}
\begin{proof}
First, note that by \Cref{lemma:u_deviation_small} and boundedness of $A_i^\ell$, we may replace $u_i^{\ell + 1}$ with $u_{i \to a}^{\ell + 1}$ in the expression for a penalty of $O_d\left(\frac{1}{\sqrt{d}}\right)$, for any $a \in \partial i$.
    Next, since $A_i^{\ell} \in \mathcal{G}_{i}^{\ell}$, we have
    \begin{align*}
    \E\left[A_i^\ell(u_{i \to a}^{\ell + 1})^2\right] - \nu_{\ell + 1} \cdot  \E\left[A_i^\ell\right] 
    = \E\left[ A_i^\ell \cdot  \E\left[\left((u_{i \to a}^{\ell + 1})^2 - \nu_{\ell + 1}\right) | ~\mathcal{G}_{i}^{\ell}\right]\right] 
    = \E\left[A_i^\ell \cdot \left((\tau_{i \to a}^{\ell + 1})^2 - \nu_{\ell + 1}\right)\right]\,,
    \end{align*}
    where $\tau_{i \to a}^{\ell+1}$ is defined in \Cref{def:message_variance}. By Cauchy-Schwarz and boundedness of $A_i^\ell$,
    \begin{align*}
        \left| \E\left[A_i^\ell\left((\tau_{i \to a}^{\ell + 1})^2 - \nu_{\ell + 1}\right)\right]\right| 
        \leq \sqrt{\E\left[\left(A_i^\ell\right)^2\right] \cdot \E\left[\left((\tau_{i \to a}^{\ell + 1})^2 - \nu_{\ell + 1}\right)^2\right]} \leq K \cdot \sqrt{\E\left[\left((\tau_{i \to a}^{\ell + 1})^2 - \nu_{\ell + 1}\right)^2\right]}\,.
    \end{align*}
    This last expectation is $O_d(\frac{1}{d})$ by \Cref{lemma:mixture_variance_bound}. The statement follows.
\end{proof}
After the large degree limit, we take a second limit sending the timestep to $0$; i.e. $\delta \to 0$ or $\frac{1}{\delta} \to \infty$. As a result, the finite difference describing $x_i^\ell, x_{i \to a}^\ell$ converges to \Cref{eqn:sde}, and the finite difference describing $z_i^\ell, z_{i \to a}^\ell$ converges to \Cref{eqn:zmartingale}. We bound the error when making this approximation, using a statement from \cite{ams20}:\footnote{The curious reader may refer to \Cref{table:notation} to convert between the notation here and in \cite{ams20}.}
\begin{proposition}[{\cite[Proposition 5.3]{ams20}, see also \cite[Proposition 4.8]{ams21}}]\label{prop:x_coupling}
Suppose \Cref{assn:alg_minimizer_exists} holds. Then there exists a coupling between $(X_t, Z_t)_{t \geq 0}$ and $(X_\ell^\delta, \widetilde{Z}_\ell^\delta)_{\ell \in \mathbb{N}}$ and a constant $C > 0$ such that for every small enough $\delta$ and $\ell \leq \frac{1}{\delta}$, we have
    $\E[|X_{\delta \ell} - X_\ell^{\delta}|^2] \leq C\delta$ and $\E[|Z_{\delta \ell} - \widetilde{Z}_\ell^{\delta}|^2] \leq C\delta$.
\end{proposition}

We are now in a position to simplify \Cref{thm:simplified_sum}:
\begin{theorem}\label{thm:sde}
    Suppose \Cref{assn:alg_minimizer_exists} holds.  Then for all $\epsilon > 0$, there exist $d_0 \in \mathbb{N}$ and $\delta_0 > 0$ such that for all $d \geq d_0$ and $0 < \delta \leq \delta_0$, 
    \begin{align*}
        \frac{1}{\alpha n}\sum_{a \in E} \E\left[f(\vz_{\partial a}^{L})\right] \ge \E[f] + \frac{\ALG_\xi - \epsilon}{\sqrt{\alpha}}\,.
    \end{align*}
\end{theorem}
\begin{proof}
Recall that $\alpha = \frac{d}{r}$ is the clause density.
By \Cref{thm:simplified_sum} and \Cref{prop:replace_u2}, we have
    \begin{align*}   
     \frac{1}{\alpha n}\sum_{a \in E} \E\left[f(\vz_{\partial a}^{L})\right] = 
    \E[f]
    +
    \frac{r}{\sqrt{d}} \cdot \sum_{\ell = 0}^{L - 1}
    \left(  \frac{1}{n} \sum_{i \in V}  \nu_{\ell+1} \cdot \E\left[A_i^{\ell}\right] \right) + O_d\left(\frac{1}{d}\right)\,.
    \end{align*}
Let $\Delta_{\ell,i} \defeq  \nu_{\ell+1} \cdot \E[A_i^\ell]$. At large enough $d$, \Cref{cor:large_degree_limit_finite_difference} and \Cref{prop:A_normalization} imply the following for all $i$:\footnote{Note that an error bound in \Cref{cor:large_degree_limit_finite_difference} could improve the scaling beyond $o_d(1)$.}
\begin{align*}
   \Delta_{\ell,i} =
    \nu_{\ell+1} \cdot \E[\Phi_{xx}^{\mu_\ast}(\delta \ell, X_\ell^\delta)] \cdot \sqrt{r} (1 + O_{\frac{1}{\delta}}(\sqrt{\delta}))
    + o_d(1)\,.
\end{align*}
By definition of $\nu_{\ell+1}$,
\begin{align*}
   \Delta_{\ell,i}  = \frac{\xi'(\delta(\ell + 1)) - \xi'(\delta\ell)}{\sqrt{r}} \cdot  \E[\Phi_{xx}^{\mu_\ast}(\delta \ell, X_\ell^\delta)] \cdot (1 + O_{\frac{1}{\delta}}(\sqrt{\delta}))
    + o_d(1)\,.
\end{align*}
By \Cref{prop:x_coupling} and Lipschitzness of $\Phi_{xx}^{\mu_\ast}(t,x)$ in $x$, using the continuous $X_{\delta \ell}$ introduces multiplicative error:
\begin{align*}
          \Delta_{\ell,i} = \frac{\xi'(\delta(\ell + 1)) - \xi'(\delta\ell)}{\sqrt{r}} \cdot  \E[\Phi_{xx}^{\mu_\ast}(\delta \ell, X_{\delta\ell})] \cdot (1 + O_{\frac{1}{\delta}}(\sqrt{\delta}))^2
    + o_d(1)\,.
\end{align*}
Putting it all together, 
\begin{align*}
    \frac{1}{\alpha n} \sum_{a \in E} \E\left[f(\vz_{\partial a}^{L})\right] 
     &=     \E[f] +
    \frac{r}{\sqrt{d}} \cdot \sum_{\ell = 0}^{L - 1} \left( \frac{1}{n} \sum_{i \in V} 
    \Delta_{\ell,i} \right)
    + O_d\left(\frac{1}{d}\right)
    \\
    &= \E[f] + 
    \sqrt{\frac{r}{d}}\cdot  \sum_{\ell = 0}^{L - 1}
    \left( \xi'(\delta(\ell + 1)) - \xi'(\delta\ell) \right) 
    \cdot  \E[\Phi_{xx}^{\mu_\ast}(\delta \ell, X_{\delta\ell})] \cdot (1 + O_{\frac{1}{\delta}}(\sqrt{\delta}))^2
    + o_d\left(\frac{1}{\sqrt{d}}\right)
    \\
    &= \E[f] + 
    \frac{1}{\sqrt{\alpha}}\cdot 
    \sum_{\ell = 0}^{L - 1}
    \left( \xi'(\delta(\ell + 1)) - \xi'(\delta\ell) \right) 
    \cdot 
    \E[\Phi_{xx}^{\mu_\ast}(\delta \ell, X_{\delta\ell})]
    \cdot 
    (1 + o_{\frac{1}{\delta}}(1))
    + o_d\left(\frac{1}{\sqrt{d}}\right)
    \,.
\end{align*}
Applying \Cref{lemma:goal} with large enough $d$ and small enough $\delta$ proves the result.
\end{proof}

\subsection{Rounding}

We have nearly achieved our objective. However, the entries of $\vz^L$ are \emph{continuous}, not in $\{\pm 1\}$. 
We account for this by \emph{rounding} each entry of $\vz^L$ without significant impact to the value of the algorithm.

Let $\trunc: \mathbb{R} \to [-1, 1]$ be the truncation function, defined by
\begin{align*}
\trunc(x) = \left\{\begin{array}{ll}
    -1 & \text{if } x < -1, \\
    x & \text{if } -1 \leq x \leq 1, \\
    1 & \text{if } x > 1.
\end{array}\right.
\end{align*}
We use a randomized rounding procedure $R$ that acts independently on each vertex $i$: it outputs $+1$ at vertex $i$ with probability $\frac{1+\trunc(z_i^L)}{2}$, and $-1$ with probability $\frac{1-\trunc(z_i^L)}{2}$. 
The average value of the algorithm is unchanged when replacing $\trunc(\vz^L)$ with $R(\vz^L)$.\footnote{For convenience, we use the notation $\trunc(\vz^L)$ or $R(\vz^L)$ to apply $\trunc$ or $R$ entrywise on $\vz^L$.}
This is because averaging over the randomness in $R$, we have for every set $S \subseteq V$,
\begin{align*}
\E\left[ \prod_{i \in S} R(z^L_i)\right]
=
 \prod_{i \in S} \trunc(z^L_i)\,.
\end{align*}
We also compare the average value of the algorithm when replacing  $\vz^L$ with $\trunc(\vz^L)$; i.e. the average value of $\E[f(\trunc(\vz_{\partial a}^L))] - \E[f(\vz_{\partial a}^L)]$ across factors $a$. Intuitively, $\trunc$ has no effect because each entry of $\vz^L$ is almost surely bounded by $1$ after taking the limits $d \to \infty$ and $\delta \to 0$:
\begin{lemma}[{\cite[Lemma 6.2 and 6.5]{ams20}}]
\label{lemma:zt_atmost_1}
    Suppose \Cref{assn:alg_minimizer_exists} holds with minimizer $\mu_\ast \in \mathscr{L}$, and consider the martingale $Z_t$ in \Cref{eqn:zmartingale} corresponding to $\mu_\ast$. Then for all $t \in [0, 1]$, $|Z_t| \leq 1$ almost surely.
\end{lemma}
For each factor $a$, the difference $\E[f(\trunc(\vz_{\partial a}^L))] - \E[f(\vz_{\partial a}^L)]$ can be split into at most $2^r$ terms of the form
\begin{align*}
    \E\left[   \left( \prod_{i \in S} \trunc(z_i^L) \right) -  \left( \prod_{i \in S} z_i^L \right) \right]\,,
\end{align*}
for some set $S \subseteq \partial a$ where $|S| \geq 2$ (since $f$ has no linear terms). We show that each of these terms is at $\delta \cdot O_d(\frac{1}{\sqrt{d}})$, where $\delta$ is the timestep of the algorithm.

Arbitrarily \emph{order} the elements of set $S$; then the above term can be written as
\begin{align*}
  \E\left[   \left( \prod_{i \in S} \trunc(z_i^L) \right) -  \left( \prod_{i \in S} z_i^L \right) \right] 
  = \sum_{\kappa = 1}^{|S|} \E\left[ 
  \left( \prod_{\iota = 1}^{\kappa-1} \trunc(z_{S[\iota]}^L)  \right) 
  \left( \trunc(z_{S[\kappa]}) - z_{S[\kappa]} \right)
  \left( \prod_{\iota = \kappa + 1}^{|S|} z_{S[\iota]}^L  \right) 
  \right]\,.
\end{align*}
We show that the average value of such terms across factors $a$ is negligible with $d$ as $n \to \infty$.
\begin{proposition}\label{prop:trci_to_trcitoa}
    For every integer $L \ge 1$, clause $a$, and disjoint subsets $T_1, T_2 \subseteq \partial a$ such that $T_1 \cup T_2 \neq \varnothing$, we have 
    \[
    \E\left[\left(
    \Big( \prod_{i \in T_1}z_i^L\Big) \prod_{i \in T_2}\trunc(z_i^L) - 
    \Big( \prod_{i \in T_1}z_{i \to a}^L\Big) \prod_{i \in T_2}\trunc(z_{i \to a}^L)
    \right)^2\right] = O_d\left(\frac{1}{d}\right).
    \]
\end{proposition}
\begin{proof}
Note that by definition, $|\trunc(x) - \trunc(y)| \leq |x - y|$ for every $x,y \in \mathbb{R}$.
As a result, for every $m \in \mathbb{N}$, $\E\left[\left|\trunc(z^{L}_{i \to a}) - \trunc(z^{L}_i)\right|^m\right] \le  \E\left[\left|z^{L}_{i \to a} - z^{L}_i\right|^m\right]  = O_d(\frac{1}{d^{m/2}})$ by \Cref{lemma:z_deviation_small}.
The statement then follows from the same proof as in \Cref{lemma:df_deviation_small}.    
\end{proof}
\begin{proposition}
    Fix an integer $L \ge 1$ and variable $i$.
    For each $a \in \partial i$, fix disjoint subsets $T_{1,a}, T_{2,a} \subseteq \partial a \backslash i$ with $T_{1,a} \cup T_{2,a} \neq \varnothing$. Then
    \begin{align*}
    \left|
    \E\left[(z_i^L - \trunc(z_i^L))
    \cdot 
    \frac{1}{d} \sum_{a \in \partial i}
    \prod_{j \in T_{1,a}}z_j^L\prod_{j \in T_{2,a}}\trunc(z_j^L)\right]\right| 
    =
    O_d\left(\frac{1}{\sqrt{d}}\right) \cdot \E[(z_i^L - \trunc(z_i^L))^2]^{1/2}\,.
    \end{align*} 
\end{proposition}
\begin{proof}
We split the left-hand side into two parts; i.e.
    \begin{align*}
    &\E\left[(z_i^L - \trunc(z_i^L))
     \cdot 
    \frac{1}{d} \sum_{a \in \partial i}
    \prod_{j \in T_{1,a}}z_{j \to a}^L \prod_{j \in T_{2,a}}\trunc(z_{j \to a}^L)\right]
    \\
    +
    &\frac{1}{d} \cdot
    \E\left[(z_i^L - \trunc(z_i^L))  
    \cdot
    \sum_{a \in \partial i}
    \left(\prod_{j \in T_{1,a}}z_j^L \prod_{j \in T_{2,a}}\trunc(z_j^L) - \prod_{j \in T_{1,a}}z_{j \to a}^L \prod_{j \in T_{2,a}}\trunc(z_{j \to a}^L)\right)\right]\,.
    \end{align*}
    The second part, by \Cref{prop:trci_to_trcitoa} and Cauchy-Schwarz, is bounded by $O_d\left(\frac{1}{\sqrt{d}}\right) \cdot \E[(z_i^L - \trunc(z_i^L))^2]^{1/2}$. 
    
    We bound the first part using a symmetry argument. 
    For any $a \in \partial i$, define $P_{a} \defeq \prod_{j \in {T_{1,a}}}z_{j \to a}^L \prod_{j \in {T_{2,a}}}\trunc(z_{j \to a}^L)$. 
    Using Cauchy-Schwarz,
    \begin{align*}
    \frac{1}{d}\cdot\left|\E\left[(z_i - \trunc(z_i)) \cdot \left(\sum_{a \in \partial i}P_a\right)\right]\right|
        &\le 
        \frac{1}{d} \cdot
        \E[(z_i - \trunc(z_i))^2]^{1/2} \cdot \E\left[\left(\sum_{a \in \partial i}P_a\right)^2\right]^{1/2}\,.
    \end{align*}
    Note that $\{P_b ~| ~b \in \Delta_{i:a} \}$ are mutually independent by \Cref{lemma:SigmaAlgebraIndependence}. So the last expectation is the variance of a sum of $d$ centered (by \Cref{lemma:uy_is_zero}) random variables with bounded moments (by \Cref{lemma:finite_moments_are_bounded}), and so has size $O_d(d)$. The statement follows.    
\end{proof}
Although this is enough to show that the difference $\frac{1}{\alpha n} \sum_{a \in E} \left( \E[f(\trunc(\vz_{\partial a}^L))] - \E[f(\vz_{\partial a}^L)] \right) = O_d(\frac{1}{\sqrt{d}})$, we want this quantity to be negligible even among terms of size $O_d(\frac{1}{\sqrt{d}})$. We show this as the timestep $\delta$ also goes to $0$.
\begin{proposition}
    Suppose \Cref{assn:alg_minimizer_exists} holds. Then for every integer $L \ge 1$, there is a $C$ such that
        \begin{align*}
            \frac{1}{n} \sum_{i \in V}  \E[(z_i^L - \trunc(z_i^L))^2] \leq C\delta + o_d(1)\,.
        \end{align*}
\end{proposition}
\begin{proof}
Consider the pseudo-Lipschitz function $\psi(x) \defeq (x - \trunc(x))^2$. By \Cref{prop:state_evolution_expectation}, 
\begin{align*}
    \E[\psi(z_i^L)] = \E[\psi(\widetilde{Z}_L^\delta)] + o_d(1)\,.
\end{align*}
We consider the behavior as $\delta$ goes to zero. Since $\psi$ is bounded by a quadratic function,  
\begin{align*}
    \E[\psi(\widetilde{Z}_L^\delta)] \le \E[|\psi(Z_{\delta L})|] + \E[|\widetilde{Z}_L^\delta - Z_{\delta L}|^2]\,.
\end{align*}
The first term is zero by \Cref{lemma:zt_atmost_1}. The statement follows by \Cref{prop:x_coupling}.
\end{proof}
We combine these statements with \Cref{thm:sde} to achieve our goal:
\begin{theorem}
    Suppose \Cref{assn:alg_minimizer_exists} holds.  Fix a factor $a$.
    Then for all $\epsilon > 0$, there exist $d_0 \in \mathbb{N}$ and $\delta_0 > 0$ such that for all $d \geq d_0$ and $0 < \delta \leq \delta_0$,
    \begin{align*}
        \frac{1}{\alpha n} \sum_{a \in E} \E\left[f(R(\vz_{\partial a}^{L}))\right] \ge \E[f] + \frac{\ALG_\xi - \epsilon}{\sqrt{\alpha}}\,.
    \end{align*}
\end{theorem}    
We conclude this section by showing the value of the algorithm \emph{concentrates} around its mean. This implies that the algorithm is not only good on average, but it finds a good solution with high probability as $n \to \infty$.
\begin{proposition}[{Generalization of \cite[Lemma 4.10]{ams21}}]
Fix a integer $L \ge 1$. Then there exists a constant $C$ such that for all $t \ge 0$,
    \begin{align*}
        \mathbb{P} \left(
        \left| \frac{1}{\alpha n} \sum_{a \in E} \left( f(R(\vz_{\partial a}^L)) - \E[f(R(\vz_{\partial a}^L))]
        \right)
        \right|
        \ge t
        \right) \le C e^{-nt^2/(2C)}\,.
    \end{align*}
\end{proposition}
\begin{proof}
Partition the nodes of the graph $G$ into $C = C(L)$ sets $\{W_1, \dots, W_{C}\}$ such that any two nodes $i,j \in W_q$ have distance at least $2L+3$. Then the value $\sum_{a \in E} \left( f(R(\vz_{\partial a}^L)) - \E[f] \right)$ can be split into $C$ sums as in
\begin{align*}
    \Delta \defeq \sum_{a \in E} \left( f(R(\vz_{\partial a}^L)) - \E[f] \right)
    = 
    \frac{1}{r} \sum_{i \in V} R(z_i^L) \sum_{a \in \partial i} \D_{i;a} f(R(\vz_{\partial a}^{L}))
    =
    \frac{1}{r} \sum_{\iota=1}^C   \Delta_\iota\,,
\end{align*}
where $\Delta_\iota \defeq \sum_{i \in W_\iota} R(z_i^L) \sum_{a \in \partial i} \D_{i;a} f(R(\vz_{\partial a}^{L}))$.
By a union bound, $\mathbb{P}( (\Delta - \E[\Delta]) \ge t \alpha n) \le \sum_{\iota = 1}^C \mathbb{P}( (\Delta_\iota - \E[\Delta_\iota]) \ge t d n / C)$. Since each $\Delta_\iota$ is a sum of up to $n$ independent variables each bounded by $d$, this probability is at most $C\exp(-\frac{(tdn/C)^2}{8nd^2}) = C\exp(-\frac{nt^2}{8C^2})$.
\end{proof}

\section{Open questions}
The first challenge is to either rigorously prove \Cref{assn:alg_minimizer_exists}, or mathematize the continuity argument that using a near-minimizer $\widetilde{\mu}_\ast \in \mathscr{L}_\xi$ in the algorithm incurs negligible error. Although such a result is expected (see \cite[Section 6]{ams20}), it has not been written down. This would address the chief concern of \Cref{cor:no_quantum_advantage}.

As mentioned before, our algorithm assumes the CSP predicate $f$ has no linear terms; e.g. $f$ cannot represent $k$-SAT. We expect that a preprocessing step fixes this issue, similar to \cite{sellke2021optimizing}'s improvement to \cite{ams20}.

Current techniques obstructing overlap-concentrated algorithms~\cite{huang2022tight,jmss22} are restricted to \emph{even} Ising spin glasses and CSP predicates $f$. However, a very recent technique~\cite{huang2023algorithmic} on spherical spin glasses seems to circumvent this issue. Applying this technique to Ising spin glasses could lead to obstructions beyond even models, which automatically transfer to $\cspfa$ by \cite{jmss22}. We suspect that message-passing is optimal among \emph{overlap-concentrated} algorithms for \emph{all} Ising spin glasses and CSP predicates $f$.

Furthermore, this work deepens a curious connection between $\cspfa$ and Ising spin glasses~\cite{dms15,jmss22}. For example, this algorithm and the one in \cite{ams20} are very similar, and achieve the same performance up to affine shift. In fact, the same is true for the QAOA for any constant depth and choice of hyperparameters~\cite{basso2021quantum,basso_spinglass_qaoa}! (This is known for monomial predicates $f$ but likely extends to all predicates.)
Perhaps there is a way to characterize the equivalence in energy landscape that explains why some algorithms ``can't see the difference'' between the two models. This equivalence extends beyond overlap gap phenomena, which are defined only at near-optimal energies.

It seems that understanding optimal algorithms for $\cspfa$ without asymptotics in $\alpha$ requires new techniques.

\section*{Acknowledgements}
Thanks to 
Ahmed El Alaoui,
Brice Huang,
Chris Jones,
Juspreet Singh Sandhu,
Mark Sellke,
and Subhabrata Sen
for answering questions about the overlap gap property, graphical models, and message-passing algorithms.
Thanks to Chris Jones for comments on a draft of this manuscript.

AC would like to thank Lorenzo Orecchia for his input on the presentation, Sekhar Tatikonda for references on message-passing, and discussion on proving state evolution, and Jonathan Shi for discussion on choosing non-linearities. AC would also like to thank Marc Mezard; content taught during his 2022 course on Statistical Physics provided crucial insight when first designing the algorithm present in this paper.

This work was done in part while KM visited the Simons Institute for the Theory of Computing. AC and KM are supported by the National Science Foundation Graduate Research Fellowship Program under Grant No. DGE-2140001. Any opinions, findings, and conclusions or recommendations expressed in this material are those of the author(s) and do not necessarily reflect the views of the National Science Foundation.

\clearpage
\newpage
\bibliography{_references}
\bibliographystyle{_alpha-betta-url}
\clearpage
\newpage
\appendix


\section{Equivalence of sparse random CSP models}
\label{sec:indexregular_vs_regular_vs_avgdegree}

We demonstrate that a CSP instance with $\alpha n$ randomly chosen clauses can be converted to an index-regular instance by changing only a $o_d \big( \frac{1}{\sqrt{d}} \big)$ fraction of edges. Our reduction follows a style that has appeared several times before, for example in~\cite{dms15,sen16}.

Given a $f: \{\pm 1\}^r \to \{0,1\}$, let $\mathcal{I}$ be an instance drawn from $\cspfa$. 
For $d \defeq r \cdot \alpha$ the average degree of a variable in $\mathcal{I}$, we define parameters $\alpha' \defeq \left\lceil \frac{d + d^{1/2}\log d}{r} \right\rceil$ and $d' \defeq r \cdot \alpha'$.
For every variable $v$ in $\mathcal{I}$ and $\iota \in [r]$, let $\deg(v,\iota)$ be the number of times $v$ appears in the $\iota^{\text{th}}$ index of a clause. Finally, fix a radius $L \in \mathbb{N}$.

The modification proceeds as follows:
\begin{enumerate}
\item First, remove clauses until $\deg(v, \iota) \le \alpha'$ for all variables $v$ and $\iota \in [r]$.

\item Then, add clauses to variables until $\deg(v, \iota) = \alpha'$ for all variables $v$ and $\iota \in [r]$. We add a clause by sampling its variables: for the $\iota^{\text{th}}$ index, we choose each variable $v$ with probability proportional to $\alpha' - \deg(v, \iota)$.
\end{enumerate} 
We show that, with high probability, this process only removes a $o_d \big( \frac{1}{\sqrt{d}} \big)$ fraction of initial clauses, and all but a vanishing fraction of radius-$L$ neighborhoods in the modified instance are locally treelike.
As a result, when we run the algorithm on the modified instance, the asymptotic performance up to $O_d \big( \frac{1}{\sqrt{d}} \big)$ does not change.\footnote{Note that the degree increase from $d$ to $d'$ also preserves the asymptotic performance, since $O_d(d') = O_d(d)$.}

We start by counting the number of removed clauses.\footnote{We do not worry about counting the number of \emph{added} clauses. This is because the algorithm cannot distinguish them from the initial clauses when the neighborhoods are locally treelike.}
In the first phase, as long as $\deg(v, \iota) > \alpha'$ for some $v$ and $\iota$, we pick a pair $(v, \iota)$ and remove an arbitrary clause in which $v$ appears in the $\iota^{\text{th}}$ index. 

\begin{proposition}
    With high probability over choice of $~\mathcal{I}$ as $n \to \infty$, the number of removed clauses is $n \cdot O_d \big( \frac{1}{d^{c \log d}} \big)$ for some $c > 0$.
\end{proposition}
\begin{proof}
Note that identically for each variable $v$ and $\iota \in [r]$, the quantity $\deg(v, \iota)$ is initially distributed as the binomial distribution $\text{Bin}(\alpha n, \frac{1}{n})$. We study the number of clauses removed because $\deg(v,\iota) > \alpha'$. This value is at most $\Delta_{v,\iota} \defeq \max(\deg(v, \iota) - \alpha', 0)$. We bound the first and second moment of $\Delta_{v,\iota}$:
\begin{align*}
    \E[ \Delta_{v,\iota} ] 
    &= \sum_{k > \alpha'} \pr\big( \deg(v, \iota) = k \big) \cdot (k - \alpha')
    = \sum_{k > \alpha'} \pr\big( \deg(v, \iota) \geq k \big) 
    \leq \int_{\alpha'}^\infty \pr \big( \deg(v, \iota) \geq t \big) \mathrm{d}t\,,
    \\
    \E[ \Delta_{v,\iota}^2 ] 
    &= \sum_{k > \alpha'} \pr \big( \deg(v, \iota) = k \big) \cdot (k - \alpha')^2
    \leq \int_{\alpha'}^\infty 2(t-\alpha') \pr \big( \deg(v, \iota) \geq t \big) \mathrm{d}t \,.
\end{align*}
Note that both of these quantities are at most $\int_{\alpha'}^\infty 2(t-\alpha) \pr \big( \deg(v, \iota) \geq t \big) \mathrm{d}t$.
Chernoff's bound says that for $t > \alpha$, we have $\pr \big( \deg(v, i) \geq t \big) \leq 2 \cdot \exp(\frac{-(t - \alpha)^2}{3\alpha})$. It follows that
\begin{align*}
\int_{\alpha'}^\infty 2(t - \alpha) \cdot \exp\left( -\frac{(t - \alpha)^2}{3\alpha} \right) \mathrm{d}t
= 3\alpha \cdot \exp\left(-\frac{(\alpha' -\alpha)^2}{3\alpha}\right) 
= O_d\left(d \cdot \exp(-c' \frac{(d^{1/2} \log d)^2}{d})\right) \,,
\end{align*}
for some $c' > 0$. So both $\E[\Delta_{v,\iota}]$ and $\E[\Delta_{v,\iota}^2]$ are bounded by $O_d \big( \frac{1}{d^{c \log d}} \big)$ for some $c > 0$.
 
Now we consider the total number of removed clauses $\Delta$. By a union bound, $\Delta \le \widetilde{\Delta}$, where $\widetilde{\Delta} \defeq \sum_{v \in [n], \iota \in [r]} \Delta_{v, \iota}$. By linearity, $\E[\widetilde{\Delta}] = rn \cdot \E[\Delta_{v, \iota}]$. Note that $\Delta_{v,\iota}$ and $\Delta_{w,\kappa}$ are independent when $\iota \ne \kappa$, so
\begin{align*}
\Var\big( \widetilde{\Delta} \big) 
= \Var \bigg( \sum_{v \in [n], \iota \in [r]} \Delta_{v, \iota} \bigg) 
\le r n \cdot \E[\Delta_{v, \iota}^2] + 2 r \cdot \sum_{w,v \in [n]} \textup{Cov} \big( \Delta_{v, \iota}, \Delta_{w, \iota} \big) \,.
\end{align*}
Since $\{\deg(v, \iota)\}_{v \in [n]}$ is distributed as a multinomial distribution, the covariance term is negative. It follows that both the mean and variance of $\widetilde{\Delta}$ are bounded by $n \cdot O_d(\frac{1}{d^{c \log d}})$ for some $c > 0$. Thus, with high probability as $n \to \infty$, $\widetilde{\Delta}$ (and thus $\Delta$) is at most $n \cdot O_d \big( \frac{1}{d^{c \log d}} \big)$ for some $c > 0$.
\end{proof}

Now we ensure that nearly all radius-$L$ neighborhoods in the modified instance are locally treelike. We start by noting that nearly all $L$-local neighborhoods in the initial instance are treelike:
\begin{proposition}
Fix any $d$ and $L$. Then with high probability over choice of $~\mathcal{I}$ as $n \to \infty$, a $1 - o(1)$ fraction of variables' $L$-local neighborhoods are treelike.
\end{proposition}
\begin{proof}
For every fixed $k \in \mathbb{N}$, we first show that the number of cycles of length $k$ is small compared to $n$. There are $O \big( k!\cdot n^k \big) = O \big( n^k \big)$ size-$k$ arrangements of variables as $[v_1, \dots, v_k]$, and $O \big( k! \cdot (\alpha n)^k \bigg) = O \big( n^k \big)$ size-$k$ arrangements of clauses $[a_1, \dots, a_k]$. The chance that each clause $a_i$ involves the variables $\{v_i, v_{i + 1~\text{mod}~k}\}$ is $O \big( (\frac{1}{n^2} )^k \big)$. By linearity, the average number of length-$k$ cycles is $O \big( \frac{n^{2k}}{n^{2k}} \big) = O(1)$, and so the average number of constant-length cycles is also $O(1)$.
By Markov's inequality, the chance that there exist more than $\log n$ constant-length cycles is at most $O \big( \frac{1}{\log n} \big)$, which vanishes as $n \to \infty$.

As $n \to \infty$, the degree of each variable is Poisson distributed with mean $d$. By a union bound, all variables have degree at most $\log n$ with high probability as $n \to \infty$. So the $L$-local neighborhood of all variables are at most $\log^c n$-sized for some $c > 0$. Using the argument above, at most $O \big( \frac{\log n}{n/\log^c n} \big) = o(1)$ of these variables can include a cycle in the $L$-local neighborhood; the others must have locally treelike neighborhoods.
\end{proof}
Removing clauses does not create any cycles. Moreover, only a small number of $L$-local neighborhoods are no longer treelike after adding clauses:
\begin{proposition}
Fix any $d$ and $L$. Then with high probability as $n \to \infty$, adding clauses preserves treelikeness in a $1 - o(1)$ fraction of $L$-local neighborhoods.
\end{proposition}
\begin{proof}
Consider the instance after removing clauses, but before adding clauses. Then every variable $v$ has degree at most $d'$, and so for fixed $d$ and $L$, the $L$-local neighborhood $\B_v(L)$ around any variable $v$ is of constant size.

We add at least $\alpha'n - \alpha n = \Theta(n)$ clauses; but only a constant number of clauses can be added to variables in $\B_v(L)$. As a result, during this process, the chance that two variables in $\B_v(L)$ are involved the same \emph{added} clause is $O\big( \frac{1}{n} \big)$. 

Thus, on average, at most a constant number of local neighborhoods are no longer treelike.
By Markov's inequality, the chance that $\log n$ local neighborhoods are no longer treelike is at most $O \big( \frac{1}{\log n} \big)$, which vanishes as $n \to \infty$. \qedhere
\end{proof}
\clearpage
\newpage

\section{Computing moments of algorithmic quantities}
\label{sec:appendix_moments}

First, we show that all quantities in the algorithm have moments bounded uniformly over $d$. We use this to show that the node and node-to-factor quantities are similar, assuming \Cref{hyp:A_deviation_small}. Second, we show that $\big\{ u_{i \to a}^{\ell} \big\}_{\ell \in [L]}$ approach Gaussian random variables as $d \to \infty$; this allows us to prove a state evolution statement in \Cref{sec:appendix_stateevo}.

\subsection{All moments are bounded}

\begin{lemma}
\label{lemma:finite_moments_are_bounded}
   Fix any $L \ge 1$. Then for every $k\in \mathbb{N}$, there is a constant $C_{k,L}$, independent of $d$, such that
   for every factor $a$ and every quantity $x \in \{w_i^\ell, w_{i \to a}^\ell, z_i^\ell, z_{i \to a}^\ell  \mid i \in \partial a, 0 \leq \ell \leq L\}$, we have $\E[|x|^k] \le C_{k,L}$.
\end{lemma}
\begin{proof}
It suffices to show that for every $\ell$, there is a constant $C_{k,L,\ell}$ which bounds the $k^{\text{th}}$ moment of $\{w_i^\ell, w_{i \to a}^\ell, z_i^\ell, z_{i \to a}^\ell \mid i \in \partial a, a \in E\}$. Then, we can take $C_{k,L} = \max_{0 \leq \ell \leq L}\{C_{k, L, \ell}\}$. Furthermore, we only need to study even $k$, since for odd $k$ we have $\E[|x|^k] \leq \sqrt{\E[x^{2k}]}$ by Cauchy-Schwarz.

We prove the statement by induction on $\ell$. The base case $\ell = 0$ follows immediately by definition, as $w^0_i = w^{0}_{i \to a} = 0$ and $z^0_i = z^0_{i \to a} \sim N(0, \delta)$. Now assume the statement holds up to some $\ell$. By definition, we have that
    \begin{align*}
    \E[(w_{i \to a}^{\ell + 1})^k] = \frac{1}{(d-1)^{k / 2}} \cdot \E\left[\left(\sum_{b \in \partial i \setminus a} \D_{i;b} f(\vz_{\partial b \to b}^\ell)\right)^k\right] =  \frac{1}{(d-1)^{k / 2}} \cdot \E\left[\sum_{b_1, \ldots, b_k \in \partial i \setminus a}\prod_{\iota = 1}^k \D_{i;b_\iota} f(\vz_{\partial b_\iota \to b_\iota}^\ell)\right].
    \end{align*}
    Note that $\D_{i;b} f(\vz_{\partial b \to b}^\ell)$ has mean zero by \Cref{cor:monomial_is_zero}, and is independent for different $b \in \partial i \setminus a$ by \Cref{remark:f_derivative_doesnt_use_i} and \Cref{lemma:SigmaAlgebraIndependence}.
    It follows that if any factor $b$ appears once in $\{b_1, \ldots, b_k\}$, then the expectation of the corresponding product is zero.
    By a counting argument and the induction hypothesis, the contribution from products with fewer than $\frac{k}{2}$ distinct factors is $O_d\left(\frac{1}{\sqrt{d}}\right)$, so it can be bounded for all $d$ by some $c_{\ell, L, k}$.
    This implies
    \begin{align*}
    \E[(w_{i \to a}^{\ell + 1})^k] =
        &\leq\,
        c_{\ell, L, k} +
        \frac{1}{(d-1)^{k / 2}} \cdot k! \cdot \E\left[\sum_{b_1, \ldots, b_{k/2} \in \partial i \setminus a}\prod_{\iota = 1}^{k/2} \left( \D_{i;b_\iota} f(\vz_{\partial b_\iota \to b_\iota}^\ell)\right)^2\right] \\
     &\leq\,
     c_{\ell, L, k} + k! \cdot \max_{b_1, \ldots, b_{k/2} \in \partial i \setminus a}\E\left[\prod_{\iota = 1}^{k/2} \left( \D_{i;b_\iota} f(\vz_{\partial b_\iota \to b_\iota}^\ell)\right)^2\right]\,.
    \end{align*}
    By independence, the expectation of the product splits into a finite product of terms  $\E\left[\left(\D_{i;b} f(\vz_{\partial b \to b}^\ell)\right)^m \right]$ for $m \le k$. We bound such terms. By H\"older's inequality,
    \begin{align*}
    \left|\D_{i;b} f(\vz_{\partial b \to b}^\ell)\right|^m
    = \left|\sum_{S \subseteq \partial b, i \in S} \hat{f}(S)\prod_{j \in S \backslash i}z^{\ell}_{j \to b}\right|^m
    \le
    \left|
    \sum_{S \subseteq \partial b, i \in S} |\hat{f}(S)|^{\frac{m}{m-1}} \right|^{m-1} \cdot
    \sum_{S \subseteq \partial b, i \in S} \left|\prod_{j \in S \backslash i}z^{\ell}_{j \to b}\right|^m\,.
    \end{align*}
    Since $\hat{f}(S) \le 1$ for all $S \subseteq \partial b$, the first term is some absolute constant $\widetilde{c}_m$. As a result,
    \begin{align*}
        \E\left[\left(\D_{i;b} f(\vz_{\partial b \to b}^\ell)\right)^m \right]
        \le \widetilde{c}_m  \cdot
        \sum_{S \subseteq \partial b, i \in S}  \E \left[\prod_{j \in S \backslash i} \left( z^{\ell}_{j \to b} \right)^m \right]
        = \widetilde{c}_m  \cdot
        \sum_{S \subseteq \partial b, i \in S} \prod_{j \in S \backslash i} \E \left[ \left( z^{\ell}_{j \to b} \right)^m \right]\,,
    \end{align*}
    where the equality is by \Cref{lemma:SigmaAlgebraIndependence} around factor $b$.
    Each $\E[(z_{j \to b}^\ell)^m]$ is at most $C_{m,L,\ell}$ by the inductive hypothesis. So $\E\left[\left( w_{i \to a}^{\ell + 1} \right)^k\right]$ is bounded by a function of $C_{m,L,\ell}$. Note $\E\left[\left( w_{i}^{\ell + 1} \right)^k\right]$ is bounded by essentially the same argument.

    Now we study $z^{\ell + 1}_{i \to a}$. Recall that we may suppose $k$ is even. By Jensen's inequality on $x \mapsto x^k$, we have
    \begin{align*}
    \E[(z_{i \to a}^{\ell + 1})^k] = \E[(z_{i \to a}^{\ell} + A_{i \to a}^{\ell} (w_{i \to a}^{\ell + 1} - w_{i \to a}^{\ell}))^k] \leq 2^{k-1}\E[(z_{i \to a}^{\ell})^k + (A_{i \to a}^{\ell} (w_{i \to a}^{\ell + 1} - w_{i \to a}^{\ell}))^k]\,.
    \end{align*}
    Recall that $A_{i \to a}^\ell$ is at most an absolute constant $K$. So by the inductive hypothesis, the right-hand side is bounded by a function of $\{C_{k,L,\ell}, \E[(w_{i \to a}^{\ell + 1})^k]\}$. Note $\E[(z_{i}^{\ell + 1})^k]$ is bounded by essentially the same argument.
\end{proof}

\begin{lemma}
\label{lemma:w_deviation_small}
    Choose $0 \le \ell \le L - 1$, variable $i$, and factor $a \in \partial i$. Then  $\E\left[\left|w^{\ell+1}_{i \to a} - w^{\ell+1}_i \right|^m\right] = O_d(\frac{1}{d^{m/2}})$ for every $m \in \mathbb{N}$.
\end{lemma}
\begin{proof}
    The difference simplifies as
\begin{align*}    w^{\ell+1}_i - w^{\ell+1}_{i \to a}
    &= \frac{1}{\sqrt{d}}\D_{i;a} f(\vz_{\partial a \to a}^\ell) + \left(\sqrt{\frac{d-1}{d}}-1\right)\cdot w^{\ell+1}_{i \to a}
    =
    \frac{1}{\sqrt{d}}
    \left(
    \D_{i;a} f(\vz_{\partial a \to a}^\ell) + \left(\sqrt{d-1} - \sqrt{d}\right)\cdot w^{\ell+1}_{i \to a}
    \right)\,.
    \end{align*}
    By \Cref{lemma:finite_moments_are_bounded}, the expression in the parentheses is a finite sum of terms with bounded moments, so it too has bounded moments. The asymptotic statement follows.
\end{proof}
\begin{corollary}
\label{lemma:u_deviation_small}
    Choose $0 \le \ell \le L - 1$, variable $i$, and factor $a \in \partial i$. Then  $\E\left[\left|u^{\ell+1}_{i \to a} - u^{\ell+1}_i\right|^m\right] = O_d(\frac{1}{d^{m/2}})$ for every $m \in \mathbb{N}$.
\end{corollary}
\begin{proof}
    When $\ell = 0$, $u_{i \to a}^1 = w_{i \to a}^1$. Otherwise, the difference
    \begin{align*}
        \sqrt{d} \left( u^{\ell+1}_i - u^{\ell+1}_{i \to a} \right)  = \sqrt{d} (w^{\ell+1}_i - w^{\ell+1}_{i \to a} ) - \sqrt{d}(w^{\ell}_i - w^{\ell}_{i \to a} )\,,
    \end{align*}
    is a sum of two terms with bounded  moments, so it too has bounded moments. The asymptotic statement follows.
\end{proof}
\begin{lemma}
\label{lemma:z_deviation_small}
     Assume \Cref{hyp:A_deviation_small}. Fix $0 \le \ell \le L$, variable $i$, and factor $a \in \partial i$. Then $\E\left[\left|z^{\ell}_{i \to a} - z^{\ell}_i\right|^m\right] = O_d(\frac{1}{d^{m/2}})$ for every $m \in \mathbb{N}$.
\end{lemma}
\begin{proof}
    We prove this by induction on $\ell$. When $\ell = 0$, the quantity is zero. Now assume this is true at $\ell = p$. Then
    \begin{align*}
        \sqrt{d} (z^{p+1}_{i \to a} - z^{p+1}_i) = \sqrt{d} (z^{p}_{i \to a} - z^{p}_i)  +
        \sqrt{d} \left(
        A_{i \to a}^p u_{i \to a}^{p+1} -  A_{i}^p u_{i}^{p+1}
        \right)\,.
    \end{align*}
    The first term has bounded moments by the induction hypothesis. The second term can be written in two parts; i.e.
    \begin{align*}
    \sqrt{d} (A_{i \to a}^p u_{i \to a}^{p+1} -  A_{i}^p u_{i}^{p+1}) = A_{i \to a}^p \cdot \sqrt{d}  (u_{i \to a}^{p+1} - u_i^{p + 1}) +
    u_i^{p + 1} \cdot \sqrt{d} (A_{i \to a}^{p} - A_i^{p}) \,.
    \end{align*}
The first part has bounded moments by boundedness of $A_{i \to a}^p$ and \Cref{lemma:u_deviation_small}. The second part is a product of two terms with bounded moments (by \Cref{lemma:finite_moments_are_bounded} and \Cref{hyp:A_deviation_small}), and so it too has bounded moments. So the full expression has bounded moments, and the asymptotic statement follows.
\end{proof}
\begin{lemma}
\label{lemma:df_deviation_small}
   Assume \Cref{hyp:A_deviation_small}. Choose $0 \le \ell \le L$, variable $i$, and factor $a \in \partial i$.
   Then $\E\left[\left|\D_{i;a}f(\vz_{\partial a}^\ell) - \D_{i;a}f(\vz_{\partial a \to a}^\ell)\right|^m\right] = O_d(\frac{1}{d^{m/2}})$ for every $m \in \mathbb{N}$.
\end{lemma}
\begin{proof}
    We can decompose the difference $\sqrt{d}\left(\D_{i;a}f(\vz_{\partial a}^\ell) -
   \D_{i;a}f(\vz_{\partial a \to a}^\ell)\right)$ into a finite sum of terms of the form
   \begin{align*}
       \sqrt{d}\left(\prod_{j \in S} z_{j \to a}^\ell - \prod_{j \in S} z_{j}^\ell\right)\,,
   \end{align*}
for nonempty $S \subseteq \partial a  \setminus i$. Choose an ordering of $S$, and let $S[\iota]$ be the $\iota$th element of $S$ for $\iota \in [|S|]$. Then
\begin{align*}
    \sqrt{d}\left(\prod_{j \in S} z_{j \to a}^\ell - \prod_{j \in S} z_{j}^\ell\right)
    =
    \sum_{\kappa=1}^{|S|}
        \left(
\prod_{\iota < \kappa} z_{S[\iota] \to a}^\ell
    \right)
    \cdot \sqrt{d}(z_{S[\kappa] \to a}^\ell - z_{S[\kappa]}^\ell) \cdot
    \left(
       \prod_{\iota > \kappa} z_{S[\iota]}^\ell
    \right)\,.
\end{align*}
By~\Cref{lemma:z_deviation_small}, $\sqrt{d}(z_{S[\iota] \to a}^\ell - z_{S[\iota]}^\ell)$ has bounded moments, and the other terms have bounded moments by \Cref{lemma:finite_moments_are_bounded}. So each term has bounded moments, and the statement follows.
\end{proof}

\subsection{Messages are asymptotically Gaussian}
Here, we show that the conditional distribution $u_{i \rightarrow a}^{\ell} \; \vert \; \cG_{i}^{\ell - 1}$, as a function of $d$,
has Wasserstein distance from a deterministic Gaussian distribution
converging in probability to zero.
This requires three high-level steps. 
First, in \Cref{sec:computing_the_moments}, we compute the second moment of $u_{i \rightarrow a}^{\ell} \; \vert \; \cG_{i}^{\ell - 1}$.
Then, in \Cref{subsub:tau_to_nu}, we compute the higher moments and show how the moments converge in probability to that of a Gaussian $U \sim \cN\big( 0, \nu_\ell\big)$.
Finally, in \Cref{sec:messages_converge_to_a_gaussian_in_wasserstein}, we use
standard tools from probability theory to finish the claim.
\subsubsection{Computing the moments}
\label{sec:computing_the_moments}

\begin{restatable}{definition}{messageConditionalVariance}
\label{def:message_variance}
Fix a predicate $f: \pmset^r \rightarrow \{ 0, 1 \}$. For every variable $i$ and factor $a \in \partial i$, and for every $1 \le \ell \le L$, define the \emph{variance parameter} $\tau_{i \to a}^{\ell} \ge 0$ via the recurrence relation
\begin{align*}
\left( \tau_{i \to a}^{1} \right)^2
&\defeq \frac{1}{d-1} \sum_{b \in \partial i  \setminus  a} \sum_{\substack{S \subseteq \partial b, i \in S}} \hat{f}(S)^2 \cdot \delta^{\lvert S \rvert - 1}\,, 
\\
\left( \tau_{i \to a}^{\ell + 1} \right)^2
&\defeq
\frac{1}{d-1} \sum_{b \in \partial i \setminus a}
\sum_{\substack{S_1, S_2 \subseteq \partial b \\ i \in S_1, S_2}} \hat{f}(S_1) \hat{f}(S_2) \cdot \phi_{i}^\ell(S_1, S_2)\,,
\end{align*}
where $\phi^\ell_{i}$ is
\begin{align*}
\phi^{\ell}_{i}(S_1, S_2) &\defeq
\sum_{\varnothing \subsetneq \widetilde{T} \subseteq (S_1 \setminus i) \cap (S_2 \setminus i)}
\left(
\prod_{\iota \in \{1,2\}}
\prod_{j \in (S_\iota  \setminus  i)  \setminus  \widetilde{T}} z_{j \to b}^{\ell - 1}
\right)
  \left(
    \prod_{j \in \widetilde{T}} \big( A_{j \to b}^{\ell - 1} \big)^2 \cdot \big( \tau_{j \to b}^{\ell} \big)^2
  \right)\,. 
\end{align*}
\end{restatable}

Notice that $\tau_{i \to a}^\ell$ is measurable with respect to $\mathcal{G}^{\ell - 1}_{i}$. We first show that $(\tau_{i \to a}^\ell)^2$ is the second moment of $u_{i \rightarrow a}^{\ell} \; \vert \; \cG_{i}^{\ell - 1}$, justifying \Cref{def:message_variance}.
\begin{lemma}
\label{lemma:message_second_moment}
Choose any variable $i$, factor $a \in \partial i$, and $1 \le \ell \le L$.
Then $\E \Big[ \big( u_{i \to a}^{\ell} \big)^2 \;\Big\vert\; \cG_{i}^{\ell - 1} \Big]
= \big( \tau_{i \to a}^{\ell} \big)^2$.
\end{lemma}
\begin{proof}
We proceed via induction.
In the base case $\ell = 1$, $u_{i \to a}^{1} = w_{i \to a}^1$ is independent of $\mathcal{G}_i^0 = \sigma(\{z_i^0\})$ by \Cref{remark:f_derivative_doesnt_use_i}:
\begin{align*}
    \E\left[   \left( u_{i \to a}^1 \right)^2 \mid \mathcal{G}_{i}^0        \right]
    =
     \E\left[  \left( w_{i \to a}^1 \right)^2  \right]
     =
     \frac{1}{d-1} \cdot \E\left[  \bigg(      \sum_{b \in \partial i  \setminus  a} \D_{i;b} f\left( \vz_{\partial b \to b}^{0} \right)                \bigg)^2 \right]\,.
\end{align*}
As in \Cref{lemma:uy_is_zero}, we use the Fourier decomposition of $f$. For any $b_1, b_2 \in \partial i  \setminus  a$,
\begin{align*}
    \D_{i;b_1} f\big( \vz_{\partial b_1 \to b_1}^{0} \big) \cdot \D_{i;b_2} f\big( \vz_{\partial b_2 \to b_2}^{0} \big)
&= \bigg(
  \sum_{S_1 \subseteq \partial b_1, i \in S_1} \hat{f}(S_1) \prod_{j \in S_1 \setminus i} z_{j \to b_1}^0
\bigg) \bigg(
  \sum_{S_2 \subseteq \partial b_2, i \in S_2} \hat{f}(S_2) \prod_{j \in S_2 \setminus i} z_{j \to b_2}^0
\bigg) \\
&= \sum_{\substack{S_1 \subseteq \partial b_1, S_2 \subseteq \partial b_2 \\ i \in S_1, S_2}} \hat{f}(S_1) \hat{f}(S_2)
\cdot \prod_{j \in (S_1 \cup S_2) \setminus i} \left(  z_{j \to b}^0 \right)^{h(j,S_1,S_2)}\,,
\end{align*}
where $h(j,S_1,S_2) = \mathbbm{1}_{[j \in S_1]} + \mathbbm{1}_{[j \in S_2]}$, and  $\mathbbm{1}_{[j \in S]} = 1$ if $j \in S$ and $0$ otherwise.

By definition (see \Cref{fig:algorithm}), $z_{j \to b}^0$ are independent for different $j$, so
\begin{align*}
    \E\left[   \left( u_{i \to a}^1 \right)^2 \mid \mathcal{G}_{i}^0        \right]
=
\frac{1}{d-1} \cdot \sum_{b_1, b_2 \in \partial i  \setminus  a}
\sum_{\substack{S_1 \subseteq \partial b_1, S_2 \subseteq \partial b_2 \\ i \in S_1, S_2}} \hat{f}(S_1) \hat{f}(S_2)
\cdot \prod_{j \in (S_1 \cup S_2) \setminus i} \E \left[  \left(  z_{j \to b}^0 \right)^{h(j,S_1,S_2)} \right]\,.
\end{align*}
Since $\E[z_{j \to b}^0] = 0$, the only terms that are nonzero correspond to $h(j,S_1, S_2) = 2$ for all $j \in (S_1 \cup S_2) \setminus  i$. This occurs only when $b_1 = b_2$ and $S_1 = S_2$, so
\begin{align*}
     \E\left[   \left( u_{i \to a}^1 \right)^2 \mid \mathcal{G}_{i}^0        \right]
&=
\frac{1}{d-1} \cdot
\sum_{b \in \partial i \setminus a}
\sum_{S \subseteq \partial b, i \in S} \hat{f}(S)^2 \cdot
\prod_{j \in S \setminus i} \E \left[  \left(  z_{j \to b}^0 \right)^{2} \right]
=
\frac{1}{d-1} \cdot
\sum_{b \in \partial i \setminus a}
\sum_{S \subseteq \partial b, i \in S} \hat{f}(S)^2 \cdot \delta^{|S|-1}\,.
\end{align*}

For the inductive case, suppose the statement holds up to some $\ell > 0$; we consider $u_{i \to a}^{\ell + 1} \mid \mathcal{G}_{i}^\ell$. Recall that
\begin{align*}
    u^{\ell + 1}_{i \to a}  = w_{i \to a}^{\ell + 1} - w_{i\to a}^\ell =
      \frac{1}{\sqrt{d-1}} \sum_{b \in \partial i \setminus a} \left(  \D_{i;b} f \big( \vz_{\partial b \rightarrow b}^{\ell} \big) - \D_{i;b}f \big( \vz_{\partial b \rightarrow b}^{\ell - 1} \big) \right) \,.
\end{align*}
The right-hand side can be written with the Fourier decomposition of $f$, as calculated in \Cref{lemma:uy_is_zero}:
\begin{align*}
   \D_{i;b} f \big( \vz_{\partial b \rightarrow b}^{\ell} \big) - \D_{i;b}f \big( \vz_{\partial b \rightarrow b}^{\ell - 1} \big) =  \sum_{S \subseteq \partial b, i \in S} \hat{f}(S)  \cdot
      \sum_{T \subsetneq S\backslash i}
      \Big(\prod_{j \in T} z^{\ell - 1}_{j \to b}\Big)
      \Big( \prod_{j \in (S\backslash i) \backslash T}  A^{\ell - 1}_{j \to b} \cdot u^{\ell}_{j \to b}\Big) \,.
\end{align*}
This gives an expression for $\left( u^{\ell + 1}_{i \to a} \right)^2$:
\begin{align*}
    \left( u^{\ell + 1}_{i \to a} \right)^2
    =
    \frac{1}{d - 1}
    \cdot
    \sum_{b_1, b_2 \in \partial i \setminus a}
\sum_{\substack{S_1 \subseteq \partial b_1, S_2 \subseteq \partial b_2 \\ i \in S_1, S_2}} \hat{f}(S_1) \hat{f}(S_2)
\cdot
\sum_{
\substack{T_1 \subsetneq S_1 \setminus i
\\ T_2 \subsetneq S_2 \setminus i}
}
\prod_{\iota \in \{1,2\}}
\left(
\prod_{j \in T_\iota} z_{j \to b_\iota}^{\ell - 1}
\right)
\left(
\prod_{j \in (S_\iota  \setminus  i)  \setminus  T_\iota} A_{j \to b_\iota}^{\ell - 1} u_{j \to b_\iota}^\ell
\right)\,.
\end{align*}
For every $b \in \partial i \setminus a$ and $j \in \partial b  \setminus  i$, we have $z_{j \to b}^{\ell - 1}, A_{j \to b}^{\ell - 1} \in \mathcal{G}_{j \to b}^{\ell - 1} \subseteq \mathcal{G}_{i}^{\ell}$ by \Cref{lemma:SigmaAlgebraDecomposition}. So the only parts of
$\left( u^{\ell + 1}_{i \to a} \right)^2$ that are \emph{not} measurable with respect to $\mathcal{G}_{i}^\ell$ are the messages $u_{j \to b_\iota}^\ell$. We take the expectation $\E\left[
  \left( u^{\ell + 1}_{i \to a} \right)^2
  \mid
  \mathcal{G}_{i}^\ell
  \right]$, and partition the terms in the sum into cases.
\begin{itemize}
    \item Each term with $b_1 \ne b_2$ contains the expression
    \begin{align*}
        \E\left[ \left( \D_{i;b_1} f \big( \vz_{\partial b_1 \rightarrow b_1}^{\ell} \big) - \D_{i;b_1}f \big( \vz_{\partial b_1 \rightarrow b_1}^{\ell - 1} \big)  \right)
        \cdot
\left( \D_{i;b_2} f \big( \vz_{\partial b_2 \rightarrow b_2}^{\ell} \big) - \D_{i;b_2}f \big( \vz_{\partial b_2 \rightarrow b_2}^{\ell - 1} \big)  \right)
\mid \mathcal{G}_{i}^{\ell}
    \right]\,.
    \end{align*}
The two differences are independent by \Cref{remark:f_derivative_doesnt_use_i} and \Cref{lemma:SigmaAlgebraIndependence}, and each has mean zero by \Cref{cor:monomial_is_zero}.
\item All other terms contain, for some $b \in \partial i  \setminus  a$ and $\{S_1, T_1, S_2, T_2\}$,  the expression
\begin{align*}
    \E\left[
    \prod_{\iota \in \{1,2\}}
\prod_{j \in (S_\iota  \setminus  i)  \setminus  T_\iota} u_{j \to b}^\ell
\mid
\mathcal{G}_i^\ell
    \right]\,.
\end{align*}
By \Cref{lemma:messages_conditional_independence}, the expression splits into $\prod_{j \in \partial b}
\E\left[
\left( u_{j \to b}^\ell \right)^{h(j, (S_1  \setminus  i)  \setminus  T_1, (S_2  \setminus  i)  \setminus  T_2)}
\mid
\mathcal{G}_i^\ell
    \right]$.
    \begin{itemize}
        \item By \Cref{lemma:uy_is_zero}, this expression is zero if $h(j, (S_1  \setminus  i)  \setminus  T_1, (S_2  \setminus  i)  \setminus  T_2) = 1$ for any $j \in \partial b$.
        \item  Otherwise, $(S_1  \setminus  i)  \setminus  T_1 = (S_2  \setminus  i)  \setminus  T_2$. The expression is $\prod_{j \in (S_1 \setminus i) \setminus T_1} \left( \tau_{j \to b}^{\ell} \right)^2$ by the inductive hypothesis.
    \end{itemize}
\end{itemize}
Thus, we may write the expectation $\E\left[
  \left( u^{\ell + 1}_{i \to a} \right)^2
  \mid
  \mathcal{G}_{i}^\ell
  \right]$ as
  \begin{align*}
    &\frac{1}{d-1}
    \cdot
    \sum_{b \in \partial i  \setminus  a }
\sum_{\substack{S_1, S_2 \subseteq \partial b \\ i \in S_1, S_2}} \hat{f}(S_1) \hat{f}(S_2)
\cdot
\sum_{
\substack{T_1 \subsetneq S_1 \setminus i
\\ T_2 \subsetneq S_2 \setminus i
\\ (S_1  \setminus i)  \setminus  T_1 = (S_2 \setminus i)  \setminus  T_2}
}
\left(
\prod_{\iota \in \{1,2\}}
\prod_{j \in T_\iota} z_{j \to b}^{\ell - 1}
\right)
\prod_{j \in (S_1  \setminus  i)  \setminus  T_1} \left(A_{j \to b}^{\ell - 1}\right)^2 \left(\tau_{j \to b}^\ell\right)^2
\\
=
  &\frac{1}{d-1}
    \cdot
    \sum_{b \in \partial i  \setminus  a}
\sum_{\substack{S_1, S_2 \subseteq \partial b \\ i \in S_1, S_2}} \hat{f}(S_1) \hat{f}(S_2)
\cdot
\sum_{\varnothing \subsetneq \widetilde{T} \subseteq (S_1 \setminus i) \cap (S_2  \setminus i)}
\left(
\prod_{\iota \in \{1,2\}}
\prod_{j \in (S_\iota  \setminus  i)  \setminus  \widetilde{T}} z_{j \to b}^{\ell - 1}
\right)
\prod_{j \in \widetilde{T}} \left(A_{j \to b}^{\ell - 1}\right)^2 \left(\tau_{j \to b}^\ell\right)^2
\,,
\end{align*}
where the equality follows by choosing $\widetilde{T} \defeq (S_1 \setminus i) \setminus T_1$.
\end{proof}
\subsubsection{Convergence in probability of moments}
\label{subsub:tau_to_nu}
We now show that $\left(\tau_{i \to a}^\ell \right)^2$ approximates $\nu_\ell$ (as in \Cref{defn:nu_defn}) in the large degree limit. This requires the observation that $\tau_{i \to a}^\ell$ is uniformly bounded over $d$:
\begin{proposition}
\label{cor:message_moment_bounded}
Fix any $L \ge 1$. Then for every $k\in \mathbb{N}$, there is a constant $C_{k,L}$, independent of $d$, such that
   for every factor $a$ and every quantity $x \in \{u_i^\ell, u_{i \to a}^\ell, \tau_{i \to a}^{\ell} \mid i \in \partial a, 1 \leq \ell \leq L\}$, we have $\E[|x|^k] \le C_{k,L}$.
\end{proposition}
\begin{proof}
Recall that $u_i^{\ell} = w_i^{\ell} - \mathbbm{1}_{[\ell \ge 1]} w_i^{\ell - 1}$. Since both terms' moments are bounded independent of $d$ by \Cref{lemma:finite_moments_are_bounded}, $u_i^{\ell}$'s moments are bounded independent of $d$. The same is true for the node-to-factor message $u_{i \to a}^{\ell}$.

By \Cref{lemma:message_second_moment}, $\left( \tau_{i  \to a}^\ell \right)^2$ is the second moment of $u_{i  \to a}^{\ell}$ conditioned on $\mathcal{G}_i^{\ell - 1}$, so for any integer $s \ge 1$, 
\begin{align*}
    \E\left[ \left( \tau_{i  \to a}^\ell \right)^{2s}
    \right]
    =
      \E\left[ \left( \E\left[ \left( u_{i  \to a}^\ell \right)^2 \mid \mathcal{G}_i^{\ell - 1}
    \right] \right)^{s}  \right]
    \le
         \E\left[  \E\left[  \left( u_{i  \to a}^\ell \right)^{2s} \mid \mathcal{G}_i^{\ell - 1} \right] \right]
    =
    \E\left[ \left( u_{i  \to a}^\ell \right)^{2s} \right]
    \,,
\end{align*}
where the inequality follows by Jensen's on $x \mapsto x^s$ for non-negative inputs. 
So all even moments, and thus all moments of $\tau_{i  \to a}^\ell$ are  bounded independent of $d$.
\end{proof}
\begin{lemma}
\label{lemma:mixture_variance_bound}
Assume \Cref{hyp:second_moment}.
    Then for all variables $i$, factors $a \in \partial i$,  and $1 \le \ell \le L$, we have $\E\left[ \left( \left( \tau_{i  \to a} ^\ell \right)^2 - \nu_\ell \right)^2 \right] = O_d\left(\frac{1}{d}\right)$, and $\left| \E[(z_{i \to a}^\ell)^2 ]- (\ell + 1)\delta \right| = O_d\left(\frac{1}{\sqrt{d}}\right)$, where the constant may depend on $\ell$.
\end{lemma}
\begin{proof}
Recall that there are no linear terms in the predicate $f$. By definition, for any factor $b$ and $s \in \mathbb{R}$,
\begin{align*}
    \xi'(s) = \sum_{j=2}^r \|f^{=j}\|^2 \cdot j \cdot s^{j-1}
    =
    \sum_{S \subseteq \partial b, |S| \ge 2} \hat{f}(S)^2 \cdot |S| \cdot s^{|S|-1}
    =
    \sum_{v \in \partial b} \sum_{S \subseteq \partial b, v \in S}\hat{f}(S)^2  \cdot  s^{|S|-1}\,.
\end{align*}
Since the instance is index-regular, we may rewrite the right-hand side for any variable $i$ as
\begin{align}
\label{eqn:xi_index_regular_expansion}
    \xi'(s) = \frac{r}{d} \cdot \sum_{b \in \partial i} \sum_{S \subseteq \partial b, i \in S} \hat{f}(S)^2 \cdot s^{|S|-1} \,.
\end{align}
We proceed by induction. For the base case $\ell = 1$, note that $\nu_1 = \frac{\xi'(\delta) - \xi'(0)}{r} = \frac{\xi'(\delta)}{r} = \left( \tau_{i \to a}^1 \right)^2 + O_d\left(\frac{1}{d} \right)$ by definition. Furthermore, $\E[(z_{i \to a}^1)^2 ]$ is bounded by
\begin{align*}
    \E[(z^{1}_{i \to a})^2] = \E\left[\left(z^{0}_{i \to a} + A^{0}_{i \to a}u^{1}_{i \to a}\right)^2\right] 
    &= \E\left[\left(z^{0}_{i \to a}\right)^2\right] + 2\E\left[\left(z^{0}_{i \to a}\right)\left(A^{0}_{i \to a}u^{1}_{i \to a}\right)\right] + \E\left[\left(A^{0}_{i \to a}u^{1}_{i \to a}\right)^2\right]
    \\
    &= \delta + 0 + \E\left[\left(A^{0}_{i \to a}\cdot \tau^{1}_{i \to a}\right)^2\right]\,,
    \end{align*}
    where the last equality follows by \Cref{lemma:uy_is_zero} and \Cref{lemma:message_second_moment}. The remaining expectation is $\delta + O_d\left(\frac{1}{d} \right)$ by $\left( \tau^1_{i \to a} \right)^2 = \nu_1 + O_d\left(\frac{1}{d} \right)$ and by \Cref{hyp:second_moment}, so $ \E[(z^{1}_{i \to a})^2]  = 2 \delta + O_d\left( \frac{1}{d} \right)$ as required.
    
For the inductive step, assume that the lemma holds up to $\ell$; we consider the statement at $\ell + 1$. We write $\left( \tau_i^{\ell + 1} \right)^2 = \frac{1}{d} \sum_{b \in \partial i  \setminus  a} (R_{i,b,\ell} + \widetilde{R}_{i,b,\ell})$, where 
\begin{align*}
    R_{i,b,\ell} &= \sum_{S \subseteq \partial b, i \in S} \hat{f}(S)^2 \cdot \phi_i^{\ell}(S,S)\,,
    \\
        \widetilde{R}_{i,b,\ell} &= \sum_{\substack{S_1, S_2 \subseteq \partial b, S_1 \ne S_2 \\ i \in S_1,S_2}} \hat{f}(S_1) \hat{f}(S_2) \cdot \phi_i^{\ell}(S_1,S_2)\,.
\end{align*}
Since $(x+y)^2 \le 2(x^2 + y^2)$, we have
\begin{align*}
    \E\left[ \left( \left( \tau_{i \to a}^\ell \right)^2 - \nu_{\ell+1} \right)^2 \right]
    \le
    2 \E\left[ \left(  \left( \frac{1}{d-1} \sum_{b \in \partial i  \setminus  a} R_{i,b,\ell} \right) - \nu_{\ell+1} \right)^2 \right]
    +
        2 \E\left[ \left( \frac{1}{d-1} \sum_{b \in \partial i  \setminus  a} \widetilde{R}_{i,b,\ell} \right)^2 \right]\,.
\end{align*}
We first show that the term involving $\widetilde{R}_{i,b,\ell}$ contributes at most $O_d(\frac{1}{d})$. Recall that $\E[z_j^{\ell}] = 0$ for all variables $j$ by \Cref{lemma:uy_is_zero}. So, by definition of $\phi_i^\ell(S_1, S_2)$ and \Cref{lemma:SigmaAlgebraIndependence} around factor $b$, $\E[\phi_i^{\ell}(S_1, S_2)] = 0$ when $S_1 \ne S_2$.
This implies $\E[\widetilde{R}_{i,b,\ell}] = 0$ by linearity. By \Cref{remark:f_derivative_doesnt_use_i} and \Cref{lemma:SigmaAlgebraIndependence}, $R_{i,b,\ell}$ are independent for different $b \in \partial i$, so 
\begin{align*}
    \E\left[\left(  \frac{1}{d-1} \sum_{b \in \partial i  \setminus  a}\widetilde{R}_{i,b,\ell} \right)^2 \right] = \frac{1}{(d-1)^2} \sum_{b \in \partial i  \setminus  a} \E \left[ \widetilde{R}_{i,b,\ell}^2 \right] = O_d\left(\frac{1}{d} \right)\,,
\end{align*}
where the last equality follows by \Cref{lemma:finite_moments_are_bounded}. 

Now we inspect $\nu_{\ell+1}$. By \Cref{eqn:xi_index_regular_expansion}, we may express it with any variable $i$ as
\begin{align*}
    \nu_{\ell + 1} = \frac{\xi'\left((\ell + 1) \delta \right) - \xi'(\ell \delta)}{r} = \frac{1}{d} \cdot \sum_{b \in \partial i} \sum_{S \subseteq \partial b, i \in S} \hat{f}(S)^2 \cdot \left(
    \left((\ell \delta + \delta \right)^{|S| - 1} - \left(\ell \delta \right)^{|S|-1} \right)\,.
\end{align*}
Similarly, we may simplify $\phi_i^\ell(S,S)$ as 
\begin{align*}
    \phi_i^{\ell}(S,S) 
    &= \sum_{\varnothing \subsetneq \widetilde{T} \subseteq S \setminus i}
    \left( \prod_{j \in (S \setminus i) \setminus \widetilde{T}} \left( z_{j \to b}^{\ell - 1} \right)^2 \right)
    \cdot 
        \left( \prod_{j \in \widetilde{T}} \left(A_{j \to b}^{\ell - 1} \cdot
        \tau_{j \to b}^{\ell - 1} \right)^2  \right) 
        \\
        &= \prod_{j \in S \setminus i}  \left( 
        \left( z_{j \to b}^{\ell - 1} \right)^2 + 
        \left(A_{j \to b}^{\ell - 1} \cdot
        \tau_{j \to b}^{\ell - 1} \right)^2  \right)
        -
    \prod_{j \in S \setminus i} \left( z_{j \to b}^{\ell - 1} \right)^2 \,.
\end{align*}
We first show that $ \phi_i^{\ell}(S,S)$ is close to $\rho_i^{\ell}(S,S) \defeq  \prod_{j \in S \setminus i} 
       \left( \left( z_{j \to b}^{\ell - 1} \right)^2 + 
        \left(A_{j \to b}^{\ell - 1} \right)^2 \cdot
        \nu_{\ell} \right)
        -
    \prod_{j \in S \setminus i} \left( z_{j \to b}^{\ell - 1} \right)^2$:
\begin{align*}
    \left( \phi_i^{\ell}(S,S)  - \rho_i^{\ell}(S,S)  \right)^2
    &=
\left( \prod_{j \in S \setminus i}  \left(
        \left( z_{j \to b}^{\ell - 1} \right)^2 + 
        \left(A_{j \to b}^{\ell - 1} \cdot
        \tau_{j \to b}^{\ell - 1} \right)^2  \right)
        -
\prod_{j \in S \setminus i} \left( 
        \left( z_{j \to b}^{\ell - 1} \right)^2 + 
        \left(A_{j \to b}^{\ell - 1} \right)^2 \cdot
        \nu_{\ell}   \right)
    \right)^2
        \\
        &= 
        \left(
        \sum_{\varnothing \subsetneq \widetilde{T} \subseteq S \setminus i}
 \left( \prod_{j \in (S \setminus i) \setminus \widetilde{T}} \left( z_{j \to b}^{\ell - 1} \right)^2 \right)
 \cdot 
 \left( \prod_{j \in \widetilde{T}} \left(A_{j \to b}^{\ell - 1}
 \right)^2
 \right)
 \cdot
 \left(
 \prod_{j \in \widetilde{T}}
 \left(
        \tau_{j \to b}^{\ell - 1} \right)^2 
        -
        \prod_{j \in \widetilde{T}} \nu_{\ell}
        \right)
    \right)^2
    \\
    &\le 
    2^r \cdot 
  \sum_{\varnothing \subsetneq \widetilde{T} \subseteq S \setminus i}
 \left( \prod_{j \in (S \setminus i) \setminus \widetilde{T}} \left( z_{j \to b}^{\ell - 1} \right)^4 \right)
 \cdot 
 \left( \prod_{j \in \widetilde{T}} \left(A_{j \to b}^{\ell - 1}
 \right)^4
 \right)
 \cdot
 \left(
 \prod_{j \in \widetilde{T}}
 \left(
        \tau_{j \to b}^{\ell - 1} \right)^2 
        -
        \prod_{j \in \widetilde{T}} \nu_{\ell}
        \right)^2\,,
\end{align*}
where the inequality follows by Cauchy-Schwarz. By \Cref{lemma:finite_moments_are_bounded} and boundedness of $A_{j \to b}^{\ell - 1}$, this expected value of the difference is controlled by the right-most term; i.e. 
\begin{align*}
    \E\left[ \left( \phi_i^{\ell}(S,S)  - \rho_i^{\ell}(S,S)  \right)^2 \right]
    =
    O_d\left(
      \sum_{\varnothing \subsetneq \widetilde{T} \subseteq S \setminus i}
\E\left[ 
\left(
 \prod_{j \in \widetilde{T}}
 \left(
        \tau_{j \to b}^{\ell - 1} \right)^2 
        -
        \prod_{j \in \widetilde{T}} \nu_{\ell}
        \right)^2
        \right]
    \right)\,.
\end{align*}
Each term in the sum above is $O_d(\frac{1}{d})$ by the inductive hypothesis and \Cref{cor:message_moment_bounded}, so the total value is $O_d(\frac{1}{d})$.

Now we compare $\rho_i^\ell(S,S)$ with $\left(
    \left((\ell \delta + \delta \right)^{|S| - 1} - \left(\ell \delta \right)^{|S|-1} \right)$. First, note that by \Cref{lemma:SigmaAlgebraIndependence} around factor $b$,
    \begin{align*}
        \E[ \rho_i^\ell(S,S) ] &=
        \prod_{j \in S \setminus i} \left( 
      \E\left[   \left( z_{j \to b}^{\ell - 1} \right)^2  \right] + 
       \E \left[ \left(A_{j \to b}^{\ell - 1} \right)^2 \cdot
        \nu_{\ell} \right] \right) 
        -
    \prod_{j \in S \setminus i}  \E\left[  \left( z_{j \to b}^{\ell - 1} \right)^2 \right]
    \\
    &=
   \prod_{j \in S \setminus i}  \left(
      \E\left[   \left( z_{j \to b}^{\ell - 1} \right)^2  \right] + 
       \delta \right) 
        -
    \prod_{j \in S \setminus i}  \E\left[  \left( z_{j \to b}^{\ell - 1} \right)^2 \right]\,;
    \end{align*}
     the last equality is by \Cref{hyp:second_moment}. By the inductive hypothesis, this value is $\left(
    \left((\ell \delta + \delta \right)^{|S| - 1} - \left(\ell \delta \right)^{|S|-1} \right) + O_d\left(\frac{1}{\sqrt{d}} \right)$.

We now put these facts together. First, since $(x+y)^2 \le 2 (x^2 + y^2)$,
\begin{align*}
    &\E\left[ \left(  \left( \frac{1}{d-1} \sum_{b \in \partial i  \setminus  a} R_{i,b,\ell} \right) - \nu_{\ell+1} \right)^2 \right]
    \\
    =\, &\E\left[ \left(  \frac{1}{d-1} \sum_{b \in \partial i \setminus a} \sum_{S \subseteq \partial a, i \in S} \hat{f}(S)^2 \cdot 
    \left( 
    \phi_i^{\ell}(S,S) - 
    \left((\ell \delta + \delta \right)^{|S| - 1}+ \left(\ell \delta \right)^{|S|-1}
    \right) \right)^2 \right]
    +
    O_d\left(\frac{1}{d} \right)
    \\
    =\, 
    &2\E\left[ \left(  \frac{1}{d-1} \sum_{b \in \partial i \setminus a} \sum_{S \subseteq \partial b, i \in S} \hat{f}(S)^2 \cdot 
    \left( 
    \rho_i^{\ell}(S,S) - 
    \left((\ell \delta + \delta \right)^{|S| - 1}+ \left(\ell \delta \right)^{|S|-1}
    \right) \right)^2 \right]
    +
    O_d \left(\frac{1}{d} \right)\,.
\end{align*}
We inspect the remaining term; there are $(d-1)^2$ choices of $b_1, b_2 \in \partial i  \setminus  a$. Recall that $\hat{f}(S)^2 \le 1$. The on-diagonal parts ($b_1 = b_2$) are at most $O_d(\frac{1}{d})$ by \Cref{lemma:finite_moments_are_bounded}. For each off-diagonal part, the expectation splits by \Cref{lemma:SigmaAlgebraIndependence} as
\begin{align*}
     O_d\left( 
     \sum_{
     \substack{ S_1 \subseteq \partial b_1, S_2 \subseteq \partial b_2
     \\ i \in S_1, S_2}
     } 
 \left( 
     \E\left[ \rho_i^{\ell}(S_1,S_1) \right] - 
    \left((\ell \delta + \delta \right)^{|S| - 1}+ \left(\ell \delta \right)^{|S|-1}
    \right)
    \left( 
 \E\left[ \rho_i^{\ell}(S_2,S_2)  \right] - 
    \left((\ell \delta + \delta \right)^{|S| - 1}+ \left(\ell \delta \right)^{|S|-1}
    \right)
     \right)\,,
\end{align*}
which is $O_d\left( \frac{1}{d} \right)$ for every choice of $S_1, S_2$ as shown earlier.
Furthermore,  $\E[(z_{i \to a}^{\ell + 1})^2 ]$ is bounded by
\begin{align*}
    \E[(z^{\ell + 1}_{i \to a})^2] = \E\left[\left(z^{\ell}_{i \to a} + A^{\ell}_{i \to a}u^{\ell + 1}_{i \to a}\right)^2\right] 
    &= \E\left[\left(z^{\ell}_{i \to a}\right)^2\right] + 2\E\left[\left(z^{\ell}_{i \to a}\right)\left(A^{\ell}_{i \to a}u^{\ell + 1}_{i \to a}\right)\right] + \E\left[\left(A^{\ell}_{i \to a}u^{\ell + 1}_{i \to a}\right)^2\right]
    \\
    &=  \ell \delta + O_d \left(\frac{1}{d} \right) +  0 + \E\left[\left(A^{\ell}_{i \to a}\cdot \tau^{\ell + 1}_{i \to a}\right)^2\right] \,,
    \end{align*}
    where the last equality follows the inductive hypothesis, \Cref{lemma:uy_is_zero}, and \Cref{lemma:message_second_moment}. 
    The remaining expectation is  $\delta + O_d\left(\frac{1}{d} \right)$, since we just proved $\left( \left( \tau^{\ell + 1}_{i \to a} \right)^2 - \nu_{\ell} \right)^2 = O_d\left(\frac{1}{d}\right)$, and by \Cref{hyp:second_moment}. So $ \E[(z^{1}_{i \to a})^2]  = (\ell + 1) \delta + O_d\left(\frac{1}{d} \right)$ as required.
\end{proof}

\begin{corollary}\label{cor:second_moment_nu}
Choose any variable $i$, factor $a \in \partial i$, and $1 \le \ell \le L$. We have
\[
\plim_{d \to \infty}\, \E \Big[ \big( u_{i \to a}^{\ell} \big)^2 \;\Big\vert\; \cG_{i}^{\ell - 1} \Big] = \nu_\ell\,.
\]
\end{corollary}
\begin{proof}
    Follows from \Cref{lemma:message_second_moment,lemma:mixture_variance_bound}, and since convergence in quadratic mean implies convergence in probability.
\end{proof}

We now compute the higher-order moments of $u_{i \to a}^{\ell} \;\big\vert\; \cG_{i}^{\ell - 1}$; they approach that of a Gaussian as $d \to \infty$.
\begin{lemma}
\label{lemma:message_higher_moments}
Assume \Cref{hyp:second_moment}.
Choose any variable $i$, and $0 \le \ell \le L - 1$. Then for any $p \ge 1$,
\begin{align*}
&\plim_{d \to \infty} \,\E\left[\left( u_{i \to a}^{\ell+1} \right)^{2p} \mid \mathcal{G}_i^{\ell} \right]
= (2p-1)!! \cdot \left( \nu_{\ell+1} \right)^{p}\,, \\
&\plim_{d \to \infty} \,\E\left[\left( u_{i \to a}^{\ell+1} \right)^{2p+1} \mid \mathcal{G}_i^{\ell} \right] = 0\,.
\end{align*}
\end{lemma}
\begin{proof}
By definition, we have that
\begin{align*}
      u^{\ell + 1}_{i \to a}  = w_{i \to a}^{\ell + 1} - \mathbbm{1}_{[\ell \ge 1]} w_{i \to a}^\ell =
      \frac{1}{\sqrt{d}} \sum_{b \in \partial i  \setminus  a} \left(  \D_{i;b} f \big( \vz_{\partial b \rightarrow b}^{\ell} \big) - \mathbbm{1}_{[\ell \ge 1]} \cdot  \D_{i;b}f \big( \vz_{\partial b \rightarrow b}^{\ell - 1} \big) \right) \,.
\end{align*}
So then for any $s \ge 2$, we may write the $s^{\text{th}}$ moment as
\begin{align*}
      \E\left[   \left( u_{i \to a}^{\ell + 1} \right)^s \mid \mathcal{G}_{i}^\ell        \right]
     &=
     \frac{1}{(d-1)^{s/2}}\cdot \E\left[  \bigg(      \sum_{b \in \partial i  \setminus  a} \left(
     \D_{i;b} f\left( \vz_{\partial b \to b}^{\ell} \right)
     -
      \mathbbm{1}_{[\ell \ge 1]} \cdot \D_{i;b} f\left( \vz_{\partial b \to b}^{\ell-1} \right)
     \right)
     \bigg)^s
     \mid \mathcal{G}_{i}^\ell       \right]
     \\
     &=
          \frac{1}{(d-1)^{s/2}}\cdot \sum_{b_1, \dots, b_s \in \partial i  \setminus  a}\E\left[       \prod_{\iota = 1}^s
           \left(
     \D_{i;b_\iota} f\left( \vz_{\partial b_\iota \to b_\iota}^{\ell} \right)
     -
      \mathbbm{1}_{[\ell \ge 1]} \cdot
      \D_{i;b_\iota} f\left( \vz_{\partial b_\iota \to b_\iota}^{\ell-1} \right)
     \right)
       \mid \mathcal{G}_{i}^\ell
          \right]\,.
\end{align*}
Note that each term has mean zero by \Cref{cor:monomial_is_zero}, and is independent for different $b \in \partial i  \setminus  a$ by \Cref{remark:f_derivative_doesnt_use_i} and \Cref{lemma:SigmaAlgebraIndependence}.
So, if any factor $b$ appears exactly \emph{once} in $\{b_1, \ldots, b_s\}$, then the expectation of the corresponding product is zero. Let $c$ be the number of distinct factors in $\{b_1, \ldots, b_s\}$; then surviving terms must satisfy $c \le \lfloor \frac{s}{2} \rfloor$.

How large are the terms with $c$ distinct factors? The distinct factors can be chosen from $\partial i  \setminus  a$ in $O_d(d^c)$ ways, and the degeneracy from permuting the set of factors is at most $s!$.
By \Cref{lemma:finite_moments_are_bounded}, each term in the sum is bounded, so the terms with $c$ distinct factors altogether contribute $\frac{1}{(d-1)^{s/2}} \cdot s! \cdot O_d(d^c) = O_d(d^{c-s/2})$ to the sum. Note that when $c < s/2$, this value is $O_d\left(\frac{1}{\sqrt{d}}\right)$, so the only contribution larger than this is when $c = \frac{s}{2}$.
\begin{itemize}
    \item Suppose $s$ is \emph{odd}. Then $c$ never equals $\frac{s}{2}$, since $c \le \lfloor \frac{s}{2} \rfloor = \frac{s-1}{2}$. So the $s^{\text{th}}$ moment has order $O_d\left(\frac{1}{\sqrt{d}}\right)$, and thereby $\plim_{d \to \infty} \,\E\left[\left( u_{i \to a}^{\ell+1} \right)^{s} \mid \mathcal{G}_i^{\ell} \right] = 0$.
 \item  Suppose $s$ is \emph{even}. When $c = \frac{s}{2}$, each distinct factor appears exactly twice in $\{b_1, \dots, b_s\}$. As a result, we may write the $s^{\text{th}}$ moment  $ \E\left[   \left( u_{i \to a}^{\ell + 1} \right)^s \mid \mathcal{G}_{i}^\ell        \right]$ as
 \begin{align*}
         &\frac{(s-1)!!}{(d-1)^{s/2}}\cdot
         \sum_{b_1, \dots, b_{s/2} \in \partial i  \setminus  a}       \prod_{\iota = 1}^{s/2}
         \E\left[
          \left(
     \D_{i;b_\iota} f\left( \vz_{\partial b_\iota \to b_\iota}^{\ell} \right)
     -
      \mathbbm{1}_{[\ell \ge 1]} \cdot
      \D_{i;b_\iota} f\left( \vz_{\partial b_\iota \to b_\iota}^{\ell-1} \right)
     \right)^2
       \mid \mathcal{G}_{i}^\ell
          \right] + O_d\left(\frac{1}{\sqrt{d}}\right)
          \\
=\,
          &\frac{(s-1)!!}{(d-1)^{s/2}}\cdot
    \left(
\sum_{b \in \partial i  \setminus  a}
\E \left[
    \left(
     \D_{i;b} f\left( \vz_{\partial b \to b}^{\ell} \right)
     -
      \mathbbm{1}_{[\ell \ge 1]} \cdot
      \D_{i;b} f\left( \vz_{\partial b \to b}^{\ell-1} \right)
     \right)^2
       \mid \mathcal{G}_{i}^\ell
\right]
    \right)^{s/2} + O_d\left(\frac{1}{\sqrt{d}}\right)
    \\
=\,
&\frac{(s-1)!!}{(d-1)^{s/2}}\cdot
    \left(
\E \left[
\sum_{b_1, b_2 \in \partial i  \setminus  a}
\prod_{\iota \in \{1,2\}}
    \left(
     \D_{i;b_\iota} f\left( \vz_{\partial b_\iota \to b_\iota}^{\ell} \right)
     -
      \mathbbm{1}_{[\ell \ge 1]} \cdot
      \D_{i;b_\iota} f\left( \vz_{\partial b_\iota \to b_\iota}^{\ell-1} \right)
     \right)
       \mid \mathcal{G}_{i}^\ell
\right]
    \right)^{s/2} + O_d\left(\frac{1}{\sqrt{d}}\right)
    \\
=\,
&(s-1)!! \cdot \left( \left(\tau_{i \to a}^{\ell + 1} \right)^2 \right)^{s/2} + O_d\left(\frac{1}{\sqrt{d}}\right)\,,
\end{align*}
where the second equality follows because the terms where $b_1 \ne b_2$ have zero mean (i.e. see proof of \Cref{lemma:message_second_moment}). 
By \Cref{lemma:mixture_variance_bound} (and \Cref{hyp:second_moment}), we have $\plim_{d \to \infty} \,\E\left[\left( u_{i \to a}^{\ell+1} \right)^{s} \mid \mathcal{G}_i^{\ell} \right]
= (s-1)!! \cdot \left( \nu_{\ell+1} \right)^{s/2}$.\qedhere
\end{itemize}
\end{proof}

\subsubsection{Wasserstein distance to Gaussian goes to zero}
\label{sec:messages_converge_to_a_gaussian_in_wasserstein}

We now show that the Wasserstein distance of $u^\ell_{i \to a} \; \big\vert \; \mathcal{G}_{i}^{\ell-1}$  and $\cN \big( 0, \nu_\ell\big)$ converges in probability to zero in the large degree limit. 
We describe the intuition here. Let $X_d \defeq u^\ell_{i \to a} \; \big\vert \; \mathcal{G}_{i}^{\ell-1}$. 
Because the moments converge in probability
to that of a Gaussian random variable with distribution $\cN\big(0, ( \tau^\ell_{i \to a} )^2 \big)$,
and the moments of a Gaussian determine its distribution (e.g., \cite[Example 30.1]{billingsley}), $X_d$ must converge in probability in distribution to some $X \sim \cN\big(0, ( \tau^\ell_{i \to a} )^2 \big)$.
Furthermore, for each even $p > 0$, because $X_d \to X$ and the $p^{\text{th}}$ moment converges in probability, the Wasserstein distance of $X_d$ and $X$ converges in probability to zero.

Let us formalize this reasoning. Denote $\cP(\R)$ by the set of all real probability measures. We recall the definition of a sequence of measures converging in Wasserstein distance:
\begin{definition}[Convergence in $p$-Wasserstein distance]
\label{def:wasserstein_convergence}
Given $p > 0$, let $\mu$ be a measure in $\cP(\R)$, and $(\mu_n)_{n \geq 1}$ a sequence of measures in $\cP(\R)$. We say that $(\mu_n)_{n \geq 1}$ converges to $\mu$ in $p$-Wasserstein distance if
\begin{align}
\label{eqn:wasserstein_defn}
\lim_{n \rightarrow \infty} \lVert \mu_n - \mu \rVert_{W_p}
= \lim_{n \rightarrow \infty} \, \Big( \inf_{\gamma \in \Gamma(\mu_n, \mu)} \, \E_{x, y \sim \gamma} \, \left[\lvert x - y \rvert^p \right] \Big)^{1/p}
= 0\,,
\end{align}
where $\Gamma(\mu_n, \mu)$ is the set of all real measures with marginals $\mu_n$ and $\mu$.
\end{definition}
If a distribution is uniquely determined by its moments, convergence of moments implies convergence in distribution:
\begin{lemma}[{e.g., \cite[Theorem 30.2]{billingsley}}]
\label{lemma:moments_to_distribution}
Let $\mu$ be a probability measure in $\cP(\R)$ uniquely determined by its moments, and let $(\mu_n)_{n \geq 1}$ be a sequence of probability measures in $\cP(\R)$. If
\begin{align*}
    \int_{\R} t^p \, d\mu_n(t)
\rightarrow \int_{\R} t^p \, d\mu(t) \, ,
\end{align*}
then $(\mu_n)_{n \ge 1}$  converges weakly to $\mu$.
\end{lemma}
From here, convergence in $p$-Wasserstein distance follows by convergence in distribution and convergence in the $p^{\text{th}}$ absolute moment around $0$:
\begin{lemma}[{e.g., \cite[Definition 6.8(i) and Theorem 6.9]{villani2009optimal}}]
\label{lemma:distribution_to_wasserstein}
Fix $p > 0$. Let $\mu$ be a measure in $\cP(\R)$, and $(\mu_n)_{n \geq 1}$ be a sequence of measures in $\cP(\R)$ that converges weakly to $\mu$, and
\begin{equation*}
\int_{\R} \lvert t \rvert^p \, d\mu_n(t)
\rightarrow \int_{\R} \lvert t \rvert^p \, d\mu(t) \, .
\end{equation*}
Then $(\mu_n)_{n \geq 1}$ converges to $\mu$ in $p$-Wasserstein distance.
\end{lemma}

\begin{lemma}[{e.g., \cite[Theorem 20.5(ii)]{billingsley}}]
\label{lemma:subsubsequence}
    Choose a sequence of random variables $(X_n)_{n \geq 1}$  and a random variable $X$. Then $\plim_{n \to \infty} X_n = X$ if and only if every subsequence $(n_i)_{i \in I}$ has a sub-subsequence $(n_j)_{j \in J}$ ($J \subseteq I$) such that $(X_{n_{j}})_{j \in J}$ converges almost surely to $X$.
\end{lemma}

Altogether, this implies the following:
\begin{theorem}
\label{thm:message_gaussian_convergence}
Assume \Cref{hyp:second_moment}.
Fix any variable $i \in V$, and integer $1 \le \ell \le L$. Then there exists a random variable $U \sim \cN\left(0, \nu_\ell  \right)$ independent from $\cG_{i}^{\ell - 1}$ such that for any $p > 0$,
\begin{align*}
  \plim_{d \to \infty}\E \left[ \big\lvert u_{i \to a}^\ell - U \big\rvert^{p} \mid \mathcal{G}_i^{\ell - 1} \right] = 0
\end{align*}
\end{theorem}
\begin{proof}
Consider the sequence of random variables $\left(u^\ell_{i \to a} \; \big\vert \; \mathcal{G}_{i}^{\ell-1}\right)_{d \geq 1}$.
By \Cref{lemma:uy_is_zero}, \Cref{cor:second_moment_nu}, and \Cref{lemma:message_higher_moments} (and \Cref{hyp:second_moment}), 
each finite moment of this sequence converges in probability as $d \to \infty$ to that of a Gaussian $U \sim \cN\left(0, \nu_\ell \right)$.

Now consider any subsequence $\left(u^\ell_{i \to a} \; \big\vert \; \mathcal{G}_{i}^{\ell-1}\right)_{d \in I}$.
By repeated application of \Cref{lemma:subsubsequence}, there exists a sub-subsequence $\left(u^\ell_{i \to a} \; \big\vert \; \mathcal{G}_{i}^{\ell-1}\right)_{d \in J}$ ($J \subseteq I$) such that \emph{for every} $k$, the $k^{\text{th}}$ moment converges \emph{almost surely} to the $k^{\text{th}}$ moment of $\cN\left(0, \nu_\ell \right)$.
Because a Gaussian distribution is uniquely determined by its moments (\cite[Example 30.1]{billingsley}), \Cref{lemma:moments_to_distribution,lemma:distribution_to_wasserstein} imply 
that the probability measures associated with $\left(u^\ell_{i \to a} \; \big\vert \; \mathcal{G}_{i}^{\ell-1}\right)_{d \in J}$ converge almost surely in $p$-Wasserstein distance to that of $\cN\left(0, \nu_\ell \right)$, for all $p > 0$. Thus, for any $I$, there exists a $J \subseteq I$ such that $\left( \E \left[ \big\lvert u_{i \to a}^\ell - U \big\rvert^{p} \mid \mathcal{G}_i^{\ell - 1} \right] \right)_{d \in J}$ converges almost surely to $0$.
The statement follows by \Cref{lemma:subsubsequence}.
\end{proof}
\clearpage
\newpage

\section{Proof of state evolution}
\label{sec:appendix_stateevo}
Here, we prove \Cref{prop:state_evolution_expectation}.
First, we show that the average of a pseudo-Lipschitz function is not affected by replacing an input message $u_{i \to a}^\ell$ with a Gaussian $U \sim \mathcal{N}\left( 0, \nu_\ell \right)$.
We use this inductively to complete the proof.
\begin{lemma}
\label{lemma:pseudo_lipschitz}
Assume \Cref{hyp:second_moment}. 
Let $1 \le \ell \le L$, and $\psi: \R^{\ell} \rightarrow \R$ be a pseudo-Lipschitz function.
For any factor $a$ and variable $i \in \partial a$, we have
\begin{align*}
    \E \left[\left|\psi \left(u_{i \to a}^1, \dots, u_{i \to a}^{\ell - 1}, u_{i \to a}^\ell \right) - \psi \left(u_{i \to a}^1, \dots, u_{i \to a}^{\ell - 1}, U\right)\right|\right] = o_d(1),
\end{align*}
where $U \sim \cN\left( 0, \nu_\ell \right)$ is a Gaussian random variable independent of $\cG_{i}^{\ell - 1}$.
\end{lemma}

\begin{proof}

By \Cref{defn:pseudolipschitz}, there exists an absolute constant $C \ge 0$ such that for any $\vx, \vy \in \R^{\ell}$, 
\begin{align*}
\E \big[ \lvert \psi(\vx) - \psi(\vy) \rvert \big]
&\leq \E \big[ C \cdot \big( 1 + \lVert \vx \rVert + \lVert \vy \rVert \big) \cdot \lVert \vx - \vy \rVert \big]\,.
\end{align*}
Applying Cauchy-Schwarz to each term, we see
\begin{align}
    \E \big[ \lvert \psi(\vx) - \psi(\vy) \rvert \big]
    \le C
    \cdot \left(
  1 + \sqrt{\E \big[ \lVert \vx \rVert^2 \big]} + \sqrt{\E \big[ \lVert \vy \rVert^2 \big]}
\right)
    \cdot \sqrt{\E \big[ \lVert \vx - \vy \rVert^2 \big]} \label{eqn:pseudo-lipschitz.identity-1}\,.
\end{align}
We fix any $1 \le \ell \le L$. Given $\mathcal{G}^{\ell - 1}_i$, let $U \sim \mathcal{N}(0, \nu_\ell)$ be the Gaussian random variable associated with $u_{i \to a}^\ell$ in \Cref{thm:message_gaussian_convergence}. We define $\vx, \vy \in \R^{\ell}$ as
\begin{equation*}
\vx = \begin{pmatrix}
  u_{i \rightarrow a}^{1} \\
  \vdots \\
  u_{i \rightarrow a}^{\ell - 1} \\
  u_{i \rightarrow a}^{\ell}
\end{pmatrix}\,,
\qquad\qquad
\vy = \begin{pmatrix}
  u_{i \rightarrow a}^{1} \\
  \vdots \\
  u_{i \rightarrow a}^{\ell - 1} \\
  U
\end{pmatrix} \, .
\end{equation*}
Note that by \Cref{lemma:message_second_moment}, $\Delta_x^2 \defeq \E\left[ \|\vx\|^2  \mid \mathcal{G}_i^{\ell - 1} \right] =  \left( \tau_{i \to a}^\ell \right)^2  + \sum_{s = 1}^{\ell-1} \left( u_{i \to a}^s \right)^2$, and $\Delta_y^2 \defeq \E\left[ \|\vy\|^2  \mid \mathcal{G}_i^{\ell - 1} \right] =
\nu_\ell  + \sum_{s = 1}^{\ell-1} \left( u_{i \to a}^s \right)^2$.
Applying \Cref{eqn:pseudo-lipschitz.identity-1}, we have
\begin{align*}
    \E \big[ \lvert \psi(\vx) - \psi(\vy) \rvert  \mid \mathcal{G}_i^{\ell - 1} \big]
    &\le C
    \cdot \left(
  1 + \sqrt{\E \big[ \lVert \vx \rVert^2  \mid \mathcal{G}_i^{\ell - 1} \big]} + \sqrt{\E \big[ \lVert \vy \rVert^2  \mid \mathcal{G}_i^{\ell - 1} \big]}
\right)
    \cdot \sqrt{\E \big[ \lVert \vx - \vy \rVert^2  \mid \mathcal{G}_i^{\ell - 1} \big]}
    \\
    &= C \cdot \left( 1 + \Delta_x + \Delta_y \right) \cdot 
    \sqrt{\E \left[ \left( u_{i \to a}^\ell - U \right)^2  \mid \mathcal{G}_i^{\ell - 1} \right]}\,.
\end{align*}
Note that by \Cref{thm:message_gaussian_convergence}, the right-most expectation converges in probability to zero. Taking the total expectation,
\begin{align*}
    \E\left[
    \left|
\psi(\vx) - \psi(\vy) \right|
\right]
&=
\E\left[
\E\left[
\left| \psi(\vx) - \psi(\vy)  \right|
\mid 
\mathcal{G}_i^{\ell - 1}
\right]
    \right]
    \\
&\le
    \E\left[ 
    C \left( 1 + \Delta_x  + \Delta_y \right) \cdot 
    \sqrt{\E \big[ \left( u_{i \to a}^\ell - U \right)^2  \mid \mathcal{G}_i^{\ell - 1} \big]}
    \right]
\\
&\le
\sqrt{\E\left[ 
    C^2 \left( 1 + \Delta_x  + \Delta_y \right)^2 \right]}
    \cdot
    \sqrt{\E \big[ \left( u_{i \to a}^\ell - U \right)^2  \big]}
\,.
\end{align*}
By definition and \Cref{cor:message_moment_bounded}, $\E[\Delta_x^2]$ and $\E[\Delta_y^2]$ are bounded independent of $d$, so the first expectation is bounded independent of $d$. The second expectation is $o_d(1)$ by \Cref{thm:message_gaussian_convergence}. The statement follows.
\end{proof}

We are now ready to prove the state evolution statement.
\stateEvolutionSimple*
\begin{proof}
Note that \Cref{eqn:state_evolution_simple.node} follows from \Cref{eqn:state_evolution_simple.factor} by \Cref{lemma:u_deviation_small} and pseudo-Lipschitzness of $\psi$.
So we focus on \Cref{eqn:state_evolution_simple.factor}, proceeding via induction.
The base case $\ell = 1$ follows directly from \Cref{lemma:pseudo_lipschitz} (using our assumption of \Cref{hyp:second_moment}). 

For the inductive step, assume that the statement holds up to $\ell$.
Define the function $\tilde{\psi}: \R^{\ell} \rightarrow \R$ given by
\begin{equation*}
\tilde{\psi} (x_1, \ldots, x_{\ell})
\defeq \E_{U \sim \mathcal{N}(0, \nu_{\ell + 1})} \big[ \psi(x_1, \ldots, x_{\ell}, U) \big] \, .
\end{equation*}
We may rewrite $\E \big[ \psi \big( U_1, \ldots, U_{\ell+1} \big) \big]$ as 
\begin{align*}
\E \big[ \psi \big( U_1, \ldots, U_{\ell+1} \big) \big]
&= \E_{U_1, \ldots, U_{\ell}} \Big[
  \E_{U_{\ell+1}} \big[ \psi \big( U_1, \ldots, U_{\ell+1} \big) \big]
\Big] 
&= \E_{U_1, \ldots, U_{\ell}} \big[ \tilde{\psi} \big( U_1, \ldots, U_{\ell} \big) \big] 
&= \E \big[ \tilde{\psi} \big( U_1, \ldots, U_{\ell} \big) \big] \, .
\end{align*}
Note that $\tilde{\psi}$ is pseudo-Lipschitz because $\psi$ is pseudo-Lipschitz.
As a result, using the induction hypothesis, 
\begin{align*}
    \E \left[ 
    \tilde{\psi} \big( u_{i \to a}^1, \ldots, u_{i \to a}^\ell \big) 
    \right]
    =
    \E \big[ \tilde{\psi} \big( U_1, \ldots, U_{\ell} \big) \big] + o_d(1)
    =
    \E \big[ \psi \big( U_1, \ldots, U_{\ell+1} \big) \big] + o_d(1)\,.
\end{align*}
It remains to show that the following term is $o_d(1)$:
\begin{align*}
    \Delta_{\ell+1} \defeq \E \left[ 
    \tilde{\psi} \big( u_{i \to a}^1, \ldots, u_{i \to a}^\ell \big) 
    \right]
    -
        \E \left[ 
    \psi \big( u_{i \to a}^1, \ldots, u_{i \to a}^\ell, u_{i \to a}^{\ell + 1} \big) 
    \right]\,.
\end{align*}
We rewrite this difference as
\begin{align*}
     \Delta_{\ell+1} = \E \left[ 
\E_{U \sim \mathcal{N}(0, \nu_{\ell + 1})}
\left[
    \psi \big( u_{i \to a}^1, \ldots, u_{i \to a}^\ell , U \big) 
    \mid \mathcal{G}_i^{\ell}
    \right]
    -
        \E \left[ 
    \psi \big( u_{i \to a}^1, \ldots, u_{i \to a}^\ell, u_{i \to a}^{\ell + 1} \big) 
    \mid \mathcal{G}_i^{\ell}
    \right]
        \right]
\,.
\end{align*}
We choose $U$ as in \Cref{lemma:pseudo_lipschitz} (using our assumption of \Cref{hyp:second_moment}) to achieve $\Delta_{\ell + 1} = o_d(1)$, as desired.
\end{proof}

\clearpage

\end{document}